\documentclass[a4paper, UKenglish,cleveref, autoref, thm-restate, numberwithinsect]{lipics-v2021}

\title{Fine-Grained Complexity of Computing \\ Degree-Constrained Spanning Trees}

\titlerunning{Degree-Constrained Spanning Trees}  
\author{Narek Bojikian}{Humboldt Universität zu Berlin, Germany}{bojikian@hu-berlin.de}{https://orcid.org/0000-0003-1072-4873}{}
\author{Alexander Firbas}{TU Wien, Austria}{alexander.firbas@tuwien.ac.at}{https://orcid.org/0009-0007-2049-2144}{FWF Project 10.55776/Y1329 and WWTF Project 10.47379/ICT22029}
\author{Robert Ganian}{TU Wien, Austria}{rganian@gmail.com}{https://orcid.org/0000-0002-7762-8045}{FWF Project 10.55776/Y1329 and WWTF Project 10.47379/ICT22029}
\author{Hung P. Hoang}{TU Wien, Austria}{phoang@ac.tuwien.ac.at}{https://orcid.org/0000-0001-7883-4134}{FWF Projects 10.55776/Y1329 and ESP1136425}
\author{Krisztina Szil\'{a}gyi}{Czech Technical University, Czechia}{krisztina.szilagyi@fit.cvut.cz}{https://orcid.org/0000-0003-3570-0528}{Supported under the project Robotics and advanced industrial production (reg. no. CZ.02.01.01/00/22\_008/0004590) and CTU Global Postdoc Fellowship Program.}

\Copyright{Narek Bojikian, Alexander Firbas, Robert Ganian, Hung P. Hoang, Krisztina Szil\'{a}gyi}

\authorrunning{N. Bojikian, A. Firbas, R. Ganian, H. P. Hoang, K. Szil\'{a}gyi}


\category{}  
\relatedversion{}   
 
\hideLIPIcs   

\nolinenumbers 

\usepackage{mathtools, mathrsfs}
\usepackage{multirow}

\newtheorem{longlemma}[theorem]{Lemma}

\newtheorem{longcorollary}[theorem]{Corollary}
\newtheorem{longclaim}[theorem]{Claim}
\newtheorem{longdefinition}[theorem]{Definition}
\newtheorem{longalgorithm}[theorem]{Algorithm}

\RequirePackage[T1]{fontenc}
\RequirePackage{complexity} 
 \RequirePackage{wrapfig}
\RequirePackage{subcaption}
\RequirePackage{todonotes}
\RequirePackage{comment}
\RequirePackage{soul}
 \RequirePackage{makecell}
\RequirePackage{todonotes}
\RequirePackage{array}
\RequirePackage{color}
\RequirePackage{cases}
\RequirePackage{xparse}
\RequirePackage{xargs}
\RequirePackage{appendix}
\RequirePackage{xcolor}
\RequirePackage{relsize}
\RequirePackage{xspace}
\RequirePackage{tabularx}
\RequirePackage{thm-restate}

\newcolumntype{h}{>{\hsize=.6\hsize}X}  \newcolumntype{L}[1]{>{\raggedright\arraybackslash}m{#1}}
\newcolumntype{C}[1]{>{\centering\arraybackslash}m{#1}}

\newcommand{\bigoh}{\mathcal{O}}

\newcommand{\idv}{\mathbf{e}}

\newcommand{\ord}{\mathscr{O}}

\renewcommand{\deg}{\ensuremath{\mathrm{deg}}}

\newtheorem{procedure}{Procedure}
\newtheorem{algorithm}{Algorithm}

\newtheorem*{definition*}{Definition}

\newcommand\xqed{   \leavevmode\unskip\penalty9999 \hbox{}\nobreak\hfill
  \quad\hbox{\textcolor{black}{$\blacktriangle$}}}

\newcommand{\para}[1]{\subparagraph*{#1}}

  \usepackage{import}

 \def\hmargin{0.0em}
\newcommand{\problembox}[3]{
  \par\addvspace{\abovedisplayskip}
  \noindent
  \makebox[\textwidth]{     \fboxsep=0.5em
    \fbox{       \begin{minipage}{\dimexpr\textwidth-2\fboxsep-2\fboxrule-\hmargin\relax}
        \textsc{#1}\par
        \smallskip
        \begin{tabularx}{\textwidth}{@{}l@{\hspace{0.5em}}X@{}}
          \textbf{Input:} & #2 \\
          \textbf{Question:} & #3
        \end{tabularx}
        \vspace{-0.33em}
      \end{minipage}     }   }
  \par\addvspace{\belowdisplayskip}
}

\newcommand{\sdmstp}{\textup{\textsc{Specified Degree MST}}\xspace}

\newcommand{\setmstp}{\textup{\textsc{Set of Degrees MST}}\xspace}

\newcommand{\dsets}{\ensuremath{D}}
\newcommand{\dvecs}{\ensuremath{\mathcal{D}}}
\newcommand{\weightf}{\ensuremath{w}}
 
\DeclareMathOperator{\ctw}{ctw}
\DeclareMathOperator{\pw}{pw}
\DeclareMathOperator{\tw}{tw}
\DeclareMathOperator{\cw}{cw}

\DeclareMathOperator{\repr}{patt}

\DeclareMathOperator{\leaf}{L}
\newcommand{\sleaf}{\mathbf{L}}
\newcommand{\sinner}{\mathbf{I}}
\DeclareMathOperator{\iso}{z}
\DeclareMathOperator{\mset}{M}
\DeclareMathOperator{\elpat}{A}
\DeclareMathOperator{\eltarget}{t}
\DeclareMathOperator{\allpairs}{\alpha_{\infty}}

 \newcommand{\costbd}{\ensuremath{B}}
 \newcommand{\degreq}{\ensuremath{D}}
 \newcommand{\pwinst}{\mathcal{I}_{\mu,\nu}}
\newcommand{\ctwinst}{\mathcal{I}_{M}}

 \newcommand{\dtree}{\ensuremath{\mathcal{T}}}
\newcommand{\bags}{\ensuremath{\mathcal{B}}}
\newcommand{\nodes}{\ensuremath{\mathcal{V}}}
\newcommand{\droot}{\ensuremath{r}}

 \DeclareMathOperator{\lab}{lab}
\newcommand{\syntaxtree}{\mathcal{T}}
\newcommand{\rrule}{\pi}

\newcommand{\ostar}{\oh^*}
\newcommand{\oh}{\mathcal{O}}

 \newcommand{\LL}{\ensuremath{\boldsymbol{L}}}
\newcommand{\RR}{\ensuremath{\boldsymbol{R}}}
\newcommand{\noc}{\ensuremath{\uparrow}}

\newcommand{\indices}{\ensuremath{\mathcal{I}}}

 \newcommand{\maxreq}{\ensuremath{r}}  \newcommand{\cspdom}{\ensuremath{\beta}}

\newcommand{\ctbags}[1][{\ell}]{\ensuremath{{\mathcal{B}_{#1}}}}
\newcommand{\nctbags}[1][{\ell}]{\ensuremath{{\mathcal{B}_{#1}^{*}}}}

 \definecolor{defblue}{rgb}{0, 0.4, 0.796}
\newcommand{\defi}[1]{\textcolor{defblue}{\emph{#1}}}

\newcommand{\ncuts}{\ensuremath{C}}

\newcommand{\zfr}{\ensuremath{\mathbb{Z}_4}}

\usepackage{drawings}

\newcommand{\cnc}{Cut\&Count\xspace}

\newcommand{\pos}{positive\xspace}

\graphicspath{{./graphics/}}

\newcommand{\sumstack}[1]{{\sum\limits_{\substack{#1}}}}

\begin{document}

\maketitle

\begin{abstract}
We investigate the computation of minimum-cost spanning trees satisfying prescribed vertex degree constraints: Given a graph $G$ and a constraint function $D$, we ask for a (minimum-cost) spanning tree $T$ such that for each vertex $v$, $T$ achieves a degree specified by $D(v)$. Specifically, we consider three kinds of constraint functions ordered by their generality---$D$ may either assign to each vertex a list of admissible degrees, an upper bound on the degree, or a specific degree. Using a combination of novel techniques and state-of-the-art machinery, we obtain an almost-complete overview of the fine-grained complexity of these problems taking into account the most classical structural graph parameters of the input graph $G$. In particular, we present SETH-tight upper and lower bounds for these problems when parameterized by pathwidth and cutwidth, an ETH-tight algorithm parameterized by clique-width, and a nearly SETH-tight algorithm parameterized by treewidth.

In order to obtain our upper bound for clique-width, we develop a novel technique of double representation through ``requirement shifting''.
Using this technique, we also obtain an ETH-tight single-exponential \XP\ algorithm for the \textsc{Exact Leaf Spanning Tree} problem parameterized by clique-width, which settles the final remaining open case for clique-width from the classical Cut and Count of Cygan et al.\ [FOCS 2011, TALG 2022]. This shows the versatility of our technique and its potential applicability to other problems as well. 
Additionally, in order to establish our lower and upper bounds we introduce a number of tools which may be of independent interest, including lazy coloring and ``asymptotic'' SETH-based reductions for structural parameters.
\end{abstract}

\newpage
\setcounter{page}{1}
 \section{Introduction}\label{sec:intro}
\noindent
Algorithms for computing minimum-cost spanning trees of graphs are, in many ways, cornerstones of computer science: classical results such as the algorithms of Prim or Kruskal are often among the first graph algorithms presented to undergraduate students. And yet, in many situations it is necessary to compute not only a spanning tree, but one that satisfies additional constraints. In this article, we investigate the computation of (minimum-cost) spanning trees that satisfy prescribed constraints on the degrees of the tree's vertices. 

More precisely, given a graph $G$ with polynomially-bounded edge weights and a constraint function $D \colon V(G)\rightarrow 2^{\mathbb{N}}$, our aim is to determine whether there is a spanning tree $T$ of $G$ such that for each $v\in V(G)$, $\deg_T(v)\in D(v)$---and if the answer is positive, output one of minimum cost. Depending on the form of the constraint function, we distinguish between the following three computational problems:

\begin{enumerate}
\item in \textsc{Set of Degrees Minimum Spanning Tree (MST)}, $D$ maps each vertex to a set of integers;
\item in \textsc{Bounded Degree MST}, all sets in the image of $D$ are of the form $\{1,\dots,d\}$ for some $d\in\mathbb{N}$;
\item in \textsc{Specified Degree MST}, the image of $D$ is a set of singletons.
\end{enumerate}

While the above problems are in fact ordered from most to least general\footnote{Every instance of \textsc{Specified Degree MST} can either be trivially rejected or stated as an equivalent instance of \textsc{Bounded Degree MST} on the same input graph $G$; see Section~\ref{sec:prelims}.}, the considered types of constraints can each be seen as natural and fundamental in their own right. Indeed, the constraints in \textsc{Bounded Degree MST} appear as a crucial base case in the \emph{Thin Tree Conjecture}~\cite{Goddyn,DBLP:conf/focs/KleinO23} and have been extensively studied in the approximation setting where one is allowed to violate the constraints by an additive constant~\cite{FurerR92,DBLP:conf/focs/Goemans06,SinghL15}. At the same time, \textsc{Specified Degree MST} and \textsc{Set of Degrees MST} form natural counterparts to the classical \textsc{$f$-Factor} and \textsc{General $f$-Factor} problems~\cite{tutte1952factors,Cornuejols88,GabowS21,GabowS21a}, respectively. The \NP-hard connected variants of the latter two problems have been studied in the literature as well~\cite{ellingham2002connected,CornelissenHMNR18,GanianNORR19}.
It is worth noting that both \textsc{Set of Degrees MST} and \textsc{Specified Degree MST} have also been studied in the literature under the same name ``\textsc{Degree Constrained Spanning Tree}''~\cite{FominGLS12,Nederlof13} and were shown to admit $2^{\bigoh(n)}$ algorithms.

While classical minimum-cost spanning trees can be computed  on general graphs in almost linear time~\cite{Chazelle00a,Chazelle00}, 
we cannot hope to achieve such an outcome for even the easiest of the three problems considered here---indeed, \textsc{Specified Degree MST} is \NP-hard as it admits a straightforward reduction from \textsc{Hamiltonian Path}. Naturally, this only rules out efficient algorithms on general graphs; in reality, the actual running time bounds will necessarily depend on the structural properties of the input graphs---for instance, all three degree-constrained MST problems admit trivial linear-time algorithms on trees. In this article, we investigate the fine-grained running time bounds for solving these problems taking into account the structure of the input graphs and obtain a surprisingly tight classification under the \emph{Exponential Time Hypothesis} (ETH)~\cite{ImpagliazzoPZ01} along with its \emph{strong} variant (SETH)~\cite{ImpagliazzoP01}.

In fact, we are not the first to investigate tight running time bounds for a connectivity problem with respect to structural parameters. The \cnc technique introduced by Cygan et al.~\cite{DBLP:conf/focs/CyganNPPRW11, CyganNPPRW22} was designed specifically to target connectivity problems under structural parameterizations, and in particular treewidth.
For many of the studied problems, they also provided SETH-based lower bounds that rule out any significant improvements to the base of the exponent.
In particular, they obtained a SETH-tight $\bigoh^*(4^{\tw(G)})$ algorithm for the \textsc{Exact Leaf Spanning Tree} problem---a variant of the spanning tree problem, where one requires the solution tree to have a given number $k$ of degree one vertices---where $\tw(G)$ is the treewidth of the input graph and $\bigoh^*$ suppresses constant as well as polynomial factors of the input.

After the aforementioned work, there has been a series of papers studying the fine-grained complexity of different connectivity problems (using \cnc and other techniques) under structural parameters~\cite{BodlaenderCKN15, DBLP:conf/stacs/BojikianCHK23, DBLP:conf/soda/CurticapeanLN18, HegerfeldK20, DBLP:conf/wg/HegerfeldK23}.
In particular, as mentioned in \cite{bojikianK25}, Hegerfeld and Kratsch~\cite{HegerfeldK23} and Bojikian and Kratsch~\cite{DBLP:conf/icalp/BojikianK24, DBLP:conf/iwpec/BojikianK25, bojikianK25} provided SETH-tight algorithms for all but one of the connectivity problems, whose SETH-tight complexity relative to treewidth was studied in~\cite{DBLP:conf/focs/CyganNPPRW11, CyganNPPRW22}, when parameterized by clique-width, which is a well-established dense counterpart to treewidth. The single problem which was left open in that series of works was \textsc{Exact Leaf Spanning Tree}.

Aside from the above line of research, the complexity of constrained variants of the \textsc{(Minimum) Spanning Tree} problem was investigated in other related works as well. For example, Gargano and Rescigno~\cite{GarganoR25} studied the variant of computing a spanning tree with minimum number of branching vertices parameterized by modular-width. Moreover, Ganian et al.~\cite{DBLP:conf/stacs/GanianHO11} studied the directed version of degree constrained spanning tree, providing an XP algorithm for this problem when parameterized by clique-width.

\smallskip
\noindent
\textbf{Contributions.}\quad
The techniques used to obtain the above results---or any other known results on computing constrained spanning trees (see, e.g.,~\cite{viana2021unifying})---do not seem to translate into our setting. In fact, as our first contribution (Theorem~\ref{thm:tw_lb}) we rule out an algorithm with single exponential running time with constant dependence on the parameter, in contrast to the treewidth-based bounds discussed above.
More precisely, we establish an SETH-based lower bound that essentially excludes any algorithm solving \textsc{Specified Degree MST} (the easiest of the three problems considered here) in time faster than $\bigoh^*((2r)^{\pw(G)})$, where $r$ is the maximum degree requirement in the image of $D$, even if $G$ is edge-unweighted.
We complement this lower bound with a tight algorithmic upper bound that holds for the most general of the three kinds of degree constraints and for polynomially bounded weight functions. In particular, we obtain a randomized algorithm solving \textsc{Set of Degrees MST} in time $\bigoh^*((2r)^{\pw(G)})$ (Theorem~\ref{thm:pw-ub}) and show that this algorithm can also be translated to treewidth, albeit with an ``almost optimal'' running time of $\bigoh^*((2r+2)^{\tw(G)})$ (Theorem~\ref{thm:tw-ub}).

The above results raise the question of whether one can also solve these problems on more general inputs, notably unweighted dense graphs.\footnote{Weights in simple dense graphs such as cliques can be used to model instances on arbitrary graphs.} More concretely, can we efficiently compute degree-constrained spanning trees by exploiting the classical graph parameter \emph{clique-width} ($\cw$)? 
As our next contribution, in Theorem~\ref{thm:cw-ub} we obtain a highly non-trivial deterministic $n^{\bigoh(\cw)}$-time algorithm solving \textsc{Set of Degrees MST} on unweighted $n$-vertex graphs\footnote{As is common in related fine-grained upper bounds, we assume that a corresponding decomposition is provided as a witness~\cite{FominGLS14,BergougnouxKK20}.}---a result which is tight under the ETH~\cite{FominGLSZ19} and generalizes the known algorithm for \textsc{Hamiltonian Path} under the same parameterization~\cite{BergougnouxKK20}.

Under current complexity assumptions, neither \cref{thm:pw-ub} nor \cref{thm:cw-ub} can be improved to yield fixed-parameter algorithms for our problems of interest under the considered structural parameterizations---indeed, any such algorithm would necessarily violate the aforementioned lower bounds for instances with sufficiently large $r$.
However, the question still remains whether one can obtain (tight) fixed-parameter algorithms under different, more restrictive parameterizations. As our third and final set of contributions, we obtain tight ``fixed-parameter'' running time bounds with respect to the graph parameter \emph{cutwidth} $\ctw$~\cite{chung1985cutwidth,ThilikosSB05}. In particular, we design a randomized algorithm solving \textsc{Set of Degrees MST} in time $\bigoh^*(3^{\ctw(G)})$ (Theorem~\ref{thm:ctw-ub}) and a complementary SETH-based lower bound which excludes running times of the form, e.g., $\bigoh^*(2.999^{\ctw(G)})$, even for \sdmstp in the edge-unweighted case (Theorem~\ref{thm:ctw_lb}).

We summarize all these results in Table~\ref{tbl:results}. We remark that all our upper bounds work for the most general variant, namely \setmstp, while all our lower bounds are obtained for the most restrictive variant, \sdmstp, implying that, surprisingly, all three problems have the same complexity for each of the considered parameters.

Finally, we show that our technique for the clique-width upper bound can be adapted to obtain an ETH-tight single-exponential XP algorithm for the \textsc{Exact Leaf Spanning Tree} problem parameterized by clique-width. This settles the last remaining fine-grained question arising from the classical work of Cygan et al.~\cite{DBLP:conf/focs/CyganNPPRW11, CyganNPPRW22}, as mentioned in the very recent work of Bojikian and Kratsch~\cite{bojikianK25}.

\begin{table}
\centering
    \begin{tabular}{c|c|c}
        \multirow{2}{*}{Parameter} & \textsc{Specified Degree MST} & \setmstp \\
        & Lower bound & Upper bound \\
        \hline
        Clique-width &  $n^{o(k)}$~\cite{FominGLSZ19} & $n^{\oh(k)}$ (Thm.~\ref{thm:cw-ub}) \\        
        Treewidth & $\ostar\big((2r - \varepsilon)^k\big)$ (Thm.~\ref{thm:tw_lb}) & $\ostar\big((2r+2)^k\big)$ (Thm.~\ref{thm:tw-ub}) \\
        Pathwidth & $\ostar\big((2r - \varepsilon)^k\big)$ (Thm.~\ref{thm:tw_lb}) & $\ostar\big((2r)^k\big)$ (Thm.~\ref{thm:pw-ub}) \\
        Cutwidth & $\ostar\big((3 - \varepsilon)^k\big)$ (Thm.~\ref{thm:ctw_lb}) & $\ostar\big(3^k\big)$ (Thm.~\ref{thm:ctw-ub}) \\
    \end{tabular}
    \smallskip
 \caption{Overview of our results. Here, $k$ always denotes the respective graph parameter, $r$ is the maximum degree requirement in the image of $D$ and $n$ is the order of the input graph. Clique-width lower bound~\cite{FominGLSZ19} is ETH-based, while all the lower bounds provided in this paper are based on SETH. All lower bounds and the clique-width algorithm are obtained for the edge-unweighted variants.}
\label{tbl:results}
\end{table}

\smallskip
\noindent
\textbf{Paper Organization.}\quad We provide a high-level overview of the proof techniques and technical contributions in the upcoming Section~\ref{sec:techniques}. After introducing the necessary preliminaries in the subsequent Section~\ref{sec:prelims}, we begin our exposition by presenting our algorithm for dense graph parameters in Section~\ref{sec:cw}. Next, in Section~\ref{sec:ubtechniques} we introduce the machinery that will be employed for the remaining upper bounds, including \cnc, the Isolation Lemma and our problem-specific modifications thereof.
We then proceed with the algorithms and lower bounds for treewidth (Section~\ref{sec:tw}), pathwidth (Section~\ref{sec:pw}) and cutwidth (Section~\ref{sec:ctw}).

\section{Technical Overview}\label{sec:techniques}
\subsection{The Clique-Width Algorithms}
\subparagraph*{High-level Summary.}
To obtain our algorithm for \textsc{Set of Degrees MST} on (unweighted) graphs parameterized by clique-width, our high-level aim is to employ the typical leaf-to-root dynamic programming approach used in almost all applications of structural graph parameters. There, one intuitively seeks to identify and maintain a correspondence between a suitably defined notion of a ``partial solution'' in the graph constructed so far---the processed graph---and a (typically combinatorial and compact) representation of that partial solution that is stored as an entry in our records. 
However, attempting to follow this approach directly for any of our variants over a clique-width decompositions fails, as the number of natural ``representations''  of partial solutions is not bounded by $n^{f(\cw)}$ (for any $f$). 
In order to circumvent this issue and obtain the sought-after single-exponential \XP\ algorithm, we will ``decompose'' these initial representations into collections of more concise representations for partial solutions. 

Towards this, we use a two-step representation procedure which characterizes the properties of the processed graph with a set of ``nice'' footprints that together are compatible with the same set of partial solutions in the rest of the graph. Note that the representation via footprints that we obtain is, by its nature, \emph{existential}: instead of representing each partial solution \emph{Sol} in the processed graph by a single footprint (as is standard), we use a set of nice footprints. Crucially, we provide a guarantee that for every full solution extending \emph{Sol} there exists at least one footprint that can be extended to a (potentially very different) full solution.

\subparagraph*{Setup.}
Our main contribution is a ``requirement shifting'' technique, that allows us, in the first place, to build a natural representation of a partial solution, but also more crucially, to reduce each such natural representation into a set of ``nice'' representations.
Concretely, at each step of a clique-decomposition, one builds a partial labeled graph $H$ of the final graph $G$ with the vertices of $H$ partitioned into $k$ label classes. 
 In a hypothetical solution, each vertex of $H$ will still receive a certain number of additional edges that have not been created yet. We call this number the (\emph{remaining}) \emph{requirement} of the vertex, whereas the intersection of the hypothetical solution with $H$ forms a spanning forest in that graph.

In a clique-decomposition, all vertices of the same label in $H$ will  act as twins (i.e., receive the same added neighbors) in future steps of the decomposition.
It is not difficult to show that one can accumulate 
 the requirements of all vertices of the same label in each connected component of a partial solution to a single vertex.
To see why this is valid, note that if the current partial solution is to be extended to a full solution in the final graph, then we can reroute the edges attached to vertices of the same label in a connected component to a single vertex representing all these vertices, without changing the degree requirements, the connectivity or the acyclicity of the rest of the graph. We refer to Figure~\ref{fig:patsexample}(a)-(b) for an illustration of this idea.

Note that a connected component of zero requirement cannot be connected to any other component in the rest of the graph. Therefore, such a component either forms a final solution (if it spans all vertices of the graph), or the considered partial solution cannot be extended to a full solution. Thus, we can assume that there is at most one connected component with requirement zero.

\begin{figure}[h!]
    \begin{subfigure}[t]{.3\textwidth}
        \centering
        \begin{minipage}[c][50pt]{0.99\textwidth}
            \centering
        \includegraphics[page=1]{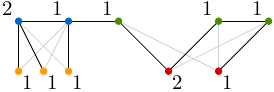}
        \end{minipage}
        \title{(a)}
    \end{subfigure}
    \hfill
    \begin{subfigure}[t]{.3\textwidth}
        \centering
        \begin{minipage}[c][50pt]{0.99\textwidth}
            \centering
        \includegraphics[page=2]{graphics/pats.pdf}
        \end{minipage}
        \title{(b)}
    \end{subfigure}   
    \hfill
    \begin{subfigure}[t]{.26\textwidth}
        \centering
        \begin{minipage}[c][50pt]{0.99\textwidth}
            \centering
        \[
        \left\langle
        \begin{pmatrix}
        3\\ 1\\
        \end{pmatrix},
        \begin{pmatrix}
        0\\ 2\\
        \end{pmatrix}
        \right\rangle
        \]
        \end{minipage}
        \title{(c)}
    \end{subfigure}
    \caption{(a) An example of a partial solution (top row), and its completion into a solution in the full graph. We depict the first label in blue, and the second in green. (b) The partial solution that results from accumulating the requirements of each label in each connected component to a single vertex. (c) We depict the pattern corresponding to the partial solution in (a).}
    \label{fig:patsexample} 
\end{figure}

Therefore, a natural representation $R$ of a partial solution $F$ in $H$ can be defined as a multiset of vectors of length $k$, where for each connected component $C$ of $F$ we add a vector to $R$ with the $i$th coordinate of this vector is defined by the total remaining degree requirement of all vertices with label $i$ in that component (see Figure~\ref{fig:patsexample}(c)). We call this representation a \defi{pattern} for $F$. Using standard arguments, one can show that this builds a correct representation of the family of partial solutions, and that can be extended correctly along the decomposition, counting the total number of solutions in the graph. This basic representation via patterns forms the foundation for our arguments.
However, the total number of such patterns is not bounded by $n^{f(k)}$ for any function $f$, as each coordinate of each vector can take values up to $n$, and a dynamic program directly building on these would need to store multisets of up to $n$ such vectors.

\subparagraph*{The Shifting Technique.}
In order to overcome this problem, we introduce a second layer of representation, that serves as the baseline for our technique.
The end goal is to reduce each pattern to a compact family of ``nice'' patterns that we use to index our dynamic programming tables. Formally, we define nice patterns as follows:
\begin{definition*}
We call a pattern $A$ \defi{nice} if $A$ contains at most one $\mathbf{0}$ vector and for each coordinate $i$, $A$ contains at most one non-unit vector $v$ such that $v_i\geq 1$.
  \end{definition*}

In order to achieve this, we show that each pattern $R$ that is not nice can be reduced into two new ``nicer'' patterns $R_1$, $R_2$ (i.e., each of $R_1$, $R_2$ has strictly fewer non-zero entries in non-unit vectors than $R$).
Crucially, $R$ behaves as an ``OR'' of $R_1$ and $R_2$. 
We remark that this existence-preserving guarantee differs from the parity-preserving representations typically employed in conjunction with the \emph{\cnc} technique---in fact, our approach yields deterministic algorithms parameterized by clique-width.

In particular, we show on one hand that both $R_1$ and $R_2$ are weaker than $R$ in the sense that each valid extension of either of $R_1$, $R_2$ is also a valid extension for $R$. On the other hand, we also prove that $R_1$ and $R_2$ together are exactly as strong as $R$.
      In order to build $R_1$ and $R_2$ from $R$, we branch over the ways in which two components of a partial solution are connected to each other in the final graph, which allows to shift (the rest of) the degree requirements between these two components. This allows us to accumulate requirements in fewer ``big'' components, resulting in nice patterns.

It is not hard to see that the total number of nice patterns is bounded by $n^{\oh(k)}$ which matches our desired running time. However, since creating a nice pattern requires us to apply the shifting operation exhaustively, a brute-force recursive application of shifting 
 could lead to an exponential blow-up in the number of patterns before they become nice. In order to overcome this problem, we apply the shifting operation in a carefully structured way which guarantees that the size of the intermediate pattern families never exceeds $n^{\oh(k)}$---see Figure~\ref{fig:smartshift} for an overview and Subsection~\ref{sub:cw-smartshiftingetc} for details.

\begin{figure}[ht!]
    \centering
\begin{tikzpicture}[scale=1]
    \pic{flowchart};
\end{tikzpicture}
\caption{
    A high-level overview of the procedure of building a nice-pattern representation from a given pattern. The procedure $\operatorname{ReduceToNice}$ takes as input a single pattern and outputs an equivalent family of nice patterns. Crucially, it fixes $k$ special vectors (one per coordinate) that will receive all the shifted requirements. It makes a single call to $\operatorname{ProcessBig}$ that makes sure that the supports of these special vectors form a partition of all labels. Both $\operatorname{ReduceToNice}$ and $\operatorname{ProcessBig}$ make calls to $\operatorname{ReduceVector}$, turning a single vector nice. Finally, $\operatorname{ReduceVector}$ makes calls to $\operatorname{ShiftVector}$ that branches over the ways the vector will connect to the rest of the solution, and shifts the rest of the requirements accordingly, replacing this vector with a set of at most $k$ ``nicer'' vectors. We show that the output of the last procedure (and hence, all of them) is equivalent to exhaustive application of the two shifting operations $\pi_1$ and $\pi_2$ that branch on a single coordinate, which correspond to the patterns $R_1$ and $R_2$ described above.}
    \label{fig:smartshift}
\end{figure}

A second challenge the algorithm needs to overcome is that when one would need to add edges inside $H$, it is not immediately obvious how this changes our patterns. We describe how this challenge can be overcome in the next paragraph which details the algorithm's operation, but to 
 facilitate the presentation of our ideas we design our dynamic programming algorithm along the closely related notion of NLC-width and NLC-decompositions. This is without loss of generality, as a clique-decomposition can be transformed into an NLC-decomposition of the same width in polynomial time~\cite{johansson1998clique}.
  The advantage of NLC-width here is that it avoids adding edges into the graph processed so far, which is important for dealing with the opaque representations we need to use for the algorithm.

\subparagraph*{The Algorithm.} We process the NLC-decomposition in a bottom-up manner, distinguishing two types of nodes: introduce vertex nodes, and join nodes. The latter performs a disjoint union of two previously constructed graphs, adds edges between these two graphs according to a set of specified label pairs, and relabels vertices based on a specified relabeling function (see Section~\ref{sec:prelims} for details). 

For an introduce vertex node, the created graph contains a single partial solution, namely the forest consisting of the single introduced vertex. Here, instead of keeping a single entry in our table for this partial solution, at this stage we already account for all possible target requirements of the created vertex and keep a separate entry for each option. 

For a join node, on a high level, we start by combining each pair of nice patterns at the two children nodes. For each such pair, our goal is to iterate over all possible ways of selecting the edges between them to include in the solution. This will result in a new family of partial solutions; after that, we apply the corresponding relabeling operations on the resulting family. We will then represent the resulting partial solutions using a family of nice patterns via the process described in Figure~\ref{fig:smartshift}, thus completing the computation of our records for the node.
 
The most challenging part of the algorithm is to handle the aforementioned edge selection step. Note that a join operation could add $\Omega(n^2)$ edges between a pair of partial solutions, meaning that a direct approach would iterate over an exponential number of new partial solutions. To overcome this challenge, we first observe that adding an edge between two ``unit-requirements'' creates a connected component of zero-requirement, therefore we restrict to choosing at most one such edge.
Using this, we can show that the exponentially-many ways of adding edges collapse into $n^{\oh(k)}$ possible resulting patterns. 
 
In particular, we distinguish three types of edges: the (at most one) edge between two unit components, edges between two big components, and edges between a big component and a unit component. 
 The algorithm first exhaustively branches over the first two types of edge. Second, we branch to determine the number of edges of the third type incident to each vertex of a big component, and the number of unit components of each label incident to these ``type-3'' edges. Note that each 
 type-3 edge ``consumes'' the requirement of this unit component, whereas the exact choice of which unit components are connected to which big components does not affect the resulting representation.
 Therefore, for each such choice, it suffices to verify the validity of the choice. We show that this final verification can be done via a reduction to the well-known \textsc{Bipartite $b$-Matching} problem, which is solvable in polynomial time~\cite{cunningham1998combinatorial}.
  This allows us to establish our first main result:

\begin{restatable}{theorem}{cwub}
 \label{thm:cw-ub}
There exists an algorithm that, given an $n$-vertex instance $(G, D)$ of unweighted \textup{\textsc{Set of Degrees MST}} and an NLC-decomposition of $G$ using $k$ labels, solves the problem in time $n^{\bigoh(k)}$.
 \end{restatable}

\subparagraph*{Exact Leaf Spanning Tree.}
Before we proceed with the remaining results for degree-constrained MST problems, we show that the shifting technique developed for Theorem~\ref{thm:cw-ub} can also be adapted to obtain an ETH-tight algorithm for the \textsc{Exact Leaf Spanning Tree} problem parameterized by clique-width. 
Again, a partial solution of this problem is a spanning forest, where we mainly care about the number of leaves of each label in a connected component of the forest, and whether at least one inner vertex of each label exists in that component. Therefore, a natural definition of a \defi{pattern} representing a partial solution consists of a vector that counts the number of isolated vertices of each label, and a multiset that contains an element for each remaining component. Each such element is a pair consisting of a vector counting the number of leaves of each label, and a set of labels that contain the labels of the inner vertices in that component.

In contrast to the previous section, instead of shifting requirements, intuitively we shift leaves (and inner vertices) between connected components of a partial solution. This requires us to allow, among others, the construction of connected components consisting of a single inner vertex. We achieve this by adding an additional ``virtual'' label $0$, that will allow us to build artificial leaves and inner vertices when needed, whose count will not be considered in the final solution. After that we can apply a similar shifting technique as before, where we guess which leaves or inner vertices connect two components in the final solution, and shift all the other leaves and inner vertices between them accordingly.

The same notion of compatibility and equivalence of partial solutions as used for \textsc{Set of Degrees MST} carries over to this setting as well. The main difficulty arises when we attempt to resolve the bottleneck of adding edges between two partial solutions. In particular, the previous solution via using $b$-matching will not apply to this setting---indeed, 
 it is possible to create an exponential number of valid non-collapsing representations by adding different sets of join edges to a pair of partial solutions. This means that iterating over all possible representations instead of sets of edges will still not suffice to achieve our targeted running time bound. Intuitively, this issue arises because for \textsc{Exact Leaf Spanning Tree}, the result of adding an edge between a unit component and a big component depends on the label of the isolated vertex, whereas for \textsc{Set of Degrees MST} this was not the case.
 
We overcome this problem by introducing a more fine-grained notion of equivalence, where we require two equivalent partial solutions to not only complete the same set of partial solutions in the rest of the graph (as was the case for \textsc{Set of Degrees MST}), but additionally that the number of added edges between each pair of labels is the same.
  We note that, while this notion of equivalence is stronger than the one defined previously, the shifting operation employed 
 there would also preserve this stronger notion of equivalence. 
Crucially, with this stronger notion of equivalence we can avoid iterating over all subsets of edges, and can instead add the edges one by one and reduce the resulting partial solutions after adding each edge. While more technically involved, this approach avoids the aforementioned exponential blow-up and yields an algorithm with the desired running time.
       
\begin{restatable}{theorem}{elst-cwub}
 \label{thm:elst-cw-ub}
There exists an algorithm that, given an $n$-vertex graph $G$, an integer $k$ and an NLC-decomposition of $G$ using $k$ labels, decides whether $G$ contains a spanning tree of exactly $k$ leaves in time $n^{\bigoh(k)}$.
 \end{restatable}

\subsection{Cut \& Count}
To obtain our algorithms for pathwidth, treewidth and cutwidth, we make use of the aforementioned \cnc technique, that was introduced to deal with connectivity constraints in parameterized algorithms~\cite{CyganNPPRW22}. This technique is a two step randomized approach: First, using the isolation lemma, one can reduce a decision version of a connectivity problem to the counting (modulo 2) version. After that, instead of counting  the solutions, the technique counts the so called ``consistent cuts'' of all relaxed (non-connected) solutions. This is equivalent to counting the number of ways to partition the vertices of a relaxed solution into two sides of a cut, such that each connected component is fully contained in one side of the cut, and by fixing the side of the cut of a single solution vertex. As a result, a relaxed solution admits an odd number of cuts if and only if it is connected. Therefore, non-connected solutions cancel out in the total count (modulo 2).
Our algorithms combine this technique with a degree-counting argument.

Unfortunately, a direct application of this technique would not be able to break the $(2r+2)^{k}$ barrier even for pathwidth, thus keeping a gap between the upper bound and our best achievable lower bound for each of these parameters.
We close this gap for pathwidth by a novel ``lazy coloring'' technique, where the algorithm decides dynamically when to ``fix'' or ``forget'' the side of a consistent cut a vertex belongs to in a partial solution, depending on the current degree of that vertex in a partial solution (instead of statically fixing the side when the vertex is introduced).

\begin{restatable}{theorem}{pwub}
 \label{thm:pw-ub}
    There exists a Monte Carlo algorithm that, given an instance $(G, w, \dsets, \costbd)$ of the \setmstp problem together with a path decomposition of $G$ of width $\pw$, solves the problem in time $\ostar\big((2r)^{\pw}\big)$.
    The algorithm produces false negatives only, and outputs the right answer with probability at least one half.
 \end{restatable}

In order to get a tight algorithm when parameterized by cutwidth, we use a given linear arrangement of width $\ctw$ to build a path decomposition of the same width, with some additional desired properties. Using a state counting argument based on the AM/GM inequality~\cite{DBLP:journals/tcs/JansenN19}, we show that the previous dynamic programming algorithm for pathwidth, based on our ``lazy coloring'' technique, only defines $3^{\ctw}$ states at each bag of this path decomposition, which yields a SETH-tight single exponential bound for cutwidth as well. 

\begin{restatable}{theorem}{ctwub}
 \label{thm:ctw-ub}
    There exists a Monte Carlo algorithm that, given an instance $(G, w, \dsets, \costbd)$ of the \setmstp problem, together with a linear arrangement of $G$ of width $\ctw$, solves this problem in time $\ostar(3^{\ctw})$. The algorithm produces false negatives only, and outputs the right answer with probability at least one half.
 \end{restatable}

Finally, while we cannot achieve the same improvement for treewidth, obtaining the current $\ostar\big((2r+2)^{k}\big)$ bound already requires non-trivial fast convolution techniques. We follow the approach of van Rooij~\cite{Rooij20} by applying multidimensional fast Fourier transformation. While a trivial non-cyclic application would add a factor $2$ to the base of the running time, we show, following an argument from \cite{Rooij20} for generalized domination problems, that this additional factor can be omitted by the utilization of the so-called ``filters'' that cancel out cyclic dependencies while only adding a polynomial factor to the running time. As a result we obtain the following.

\begin{restatable}{theorem}{twub}
 \label{thm:tw-ub}
    There exists a Monte Carlo algorithm that, given an instance $(G, w, \dsets, \costbd)$ of the \setmstp problem together with a nice tree decomposition of $G$ of width $\tw$, solves the problem in time $\ostar\big((2r + 2)^{\tw}\big)$. 
    The algorithm produces false negatives only and outputs the right answer with probability at least one half.
 \end{restatable}

\subsection{Lower Bounds}

For our tight SETH-based lower bounds for pathwidth and cutwidth, on a high level we follow the approach of Lokshtanov, Marx, and Saurabh~\cite{DBLP:journals/talg/LokshtanovMS18}---but with a twist. 
We start with an overview of this approach, which is essentially a reduction from a constraint satisfiability problem~\cite{MR4121882} to our problem, \textsc{Specified Degree MST}.
For each variable, we construct a sequence of small components, where two consecutive components are separated by a cut-vertex; see \cref{fig:pw-lb-summary} for an illustration for pathwidth.
These components encode the value of the corresponding variable, and we use the specified degree constraint to make sure that they encode the same value.
Then for each constraint of the satisfiability instance, we have a gadget to combine a component from each involved variable to determine if the constraint is satisfied.
Finally, we need to ensure that the solution is a spanning tree.

\begin{figure}[ht!]
    \centering
    \includegraphics[width=.8\textwidth]{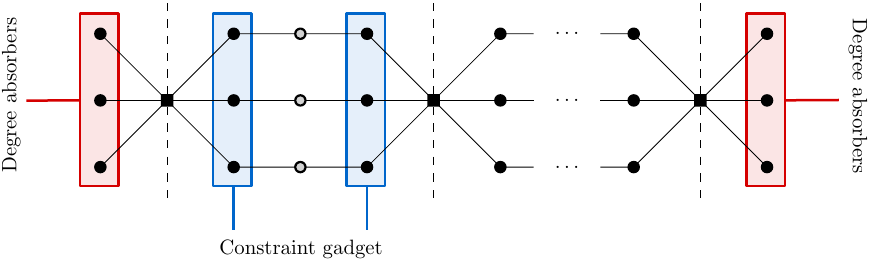}
    \caption{A schematic of the vertex gadget in the lower bound construction for pathwidth. This gadget corresponds to a variable of the satisfiability problem. The cut-vertices (drawn as squares) have the same degree requirement. The value of the variable is roughly represented by how many edges of the solution spanning tree are on the left side of each cut-vertex; if the degree requirements of the gray vertices are one, then this value is consistent among all cut-vertices of the gadget.}
    \label{fig:pw-lb-summary}
\end{figure}

While the description above follows a general template for a SETH-reduction, our reductions also incorporate a new technique.
This technique arises from the difficulty of obtaining a classical ``one-to-one'' reduction to our problem using standard gadgets. 
Intuitively, several constraints imposed by the desired lower bounds and the problem conflict with each other.
For example, due to the specified degree constraint, in a one-to-one reduction, we may want to connect the first and the last components of the variable gadget so that the total degrees are preserved regardless of the value of the variable.
However, it is not obvious how to do so while keeping both the pathwidth of the constructed graph and the maximum degree of the solution small.
To circumvent this issue, we employ Turing reductions instead of the usual one-to-one reductions used for SETH; to the best of our knowledge, this is the first time such reductions are used for structural parameters. 
In the example above, with this approach, instead of connecting the first and the last components of the variable gadget, we can instead have some \emph{degree absorbers} to accommodate the excess degree at each end of the gadget; although we do not know exactly what this excess degree is, we can show that there are only a polynomial number of possibilities across all variable gadgets.
Using this technique, we can prove the desired result for pathwidth.

\begin{restatable}{theorem}{twlb}
 \label{thm:tw_lb}
    For every $\maxreq \geq 3$ and every $\varepsilon\in \mathbb{Q}^+$, unweighted \sdmstp cannot be solved in time $\ostar((2\maxreq-\varepsilon)^{\pw})$ unless the Strong Exponential Time Hypothesis fails.
 \end{restatable}

We note that such a lower bound is highly unlikely to hold for $\maxreq = 2$, since this case is equivalent (up to a polynomial factor in the running time) to the \textsc{Hamiltonian Path} problem, which can be solved in (SETH-tight) running time $\ostar((2+\sqrt 2)^{\pw})$~\cite{CyganKN18}.

There is also another highly non-standard behavior in our reduction specific to cutwidth.
Instead of producing graphs with cutwidth $n+\delta$ where $n$ is the input complexity measure and $\delta$ is some fixed constant, we can only obtain graphs with cutwidth $\omega\cdot n+\delta$, where $\omega>1$. 
Under normal circumstances, this would not allow us to obtain a tight bound under SETH.
However, here we construct a family of reductions that allows us to set $\omega$ arbitrarily close to $1$ and show that this still suffices to achieve the following desired bound.

\begin{restatable}{theorem}{ctwlb}
 \label{thm:ctw_lb}
    Assuming SETH, the unweighted \sdmstp problem cannot be solved in time $\ostar((3-\varepsilon)^{\ctw})$ for any $\varepsilon>0$.
 \end{restatable}

\section{Preliminaries}\label{sec:prelims}
\noindent
We assume basic familiarity with graph terminology~\cite{Diestel}, the Exponential Time Hypothesis (ETH) along with its strong variant (SETH)~\cite{ImpagliazzoPZ01, ImpagliazzoP01} and the $\bigoh^*$ notation which suppresses polynomial factors of the input size. 
For a positive integer $k$, we use $[k]$ to denote the set $\{1,\dots,k\}$ and $[k]_0$ the set $\{0, 1, \dots, k\}$. 
All vectors we consider are indexed by the set $[k]$ for some value of $k$.
We denote the unit vector of the $i$-th dimension with $\mathbf{e}_i$ and the zero vector with $\mathbf{0}$. We denote multisets using angle brackets.

For a function $g:U\rightarrow V$ that is not explicitly defined as a weight function, and for a set $C\subseteq U$, we denote by $g(C)$ the set $\{g(x)  \colon x \in C \}$. If $g$ is explicitly defined as a weight function, where $V$ is some ring, we denote by $g(C)$ the sum $\sum_{x\in C} g(x)$.
 
Let $G = (V,E)$ be a graph.
For a set of edges $F \subset E$, we denote by $G[F]$ the graph $(V',F)$, where $V'$ is the set of endpoints of $F$.
For a subgraph $H$ of $G$ and a vertex $v$ of $G$, we denote by $\deg_H(v)$ the degree of $v$ in $H$.
For convenience, when $H = G[F]$ for a set of edges $F$, we use $\deg_F(v)$ for $\deg_{G[F]}(v)$.
We denote by $\cc(G)$ the set of all connected components of $G$ and by
$N_G(v)$ the neighborhood of $v$ in~$G$. 

For a ground set $U$, a mapping $f:U\rightarrow V$ and a set $S\subseteq U$, we define the restriction $f|_S: S\rightarrow V$ as the restriction of $f$ to $S$.
For some element $u$ (not necessarily in $U$), and a value $v\in V$, we define the mapping $f[u\mapsto v]: U\cup \{u\}\rightarrow V$ as the mapping
\[
f[u\mapsto v](u') = \begin{cases}
v &\colon \text{if } u' = u,\\
f(u') &\colon \text{otherwise}.
\end{cases}
\]

The \textsc{Bipartite $b$-Matching} problem takes as input a bipartite graph $(A\cup B, E)$ together with a vertex capacity function $c:A\cup B\rightarrow \mathbb{N}_0$ (and possibly an edge capacity function $c':E\rightarrow \mathbb{N}_0$), and asks whether there is an assignment $f:E\rightarrow \mathbb{N}_0$ such that for each vertex $v\in A\cup B$, it holds that $c(v)=\sum_{uv\in E}f(uv)$ (and for each $e\in E$, it holds that $f(e)\leq c'(e)$).
In this work, we are interested in the restricted case where $c'(e) = 1$ for each edge $e\in E$. 
It is well-known that this problem can be solved in polynomial time by a straightforward reduction to the \textsc{Maximum Flow} problem. We refer to \cite[Section 5.5]{cunningham1998combinatorial} for a deeper discussion, in particular to Theorem 5.35 of this book.

\smallskip
\noindent
\textbf{Degree-Constrained Spanning Tree Problems.}\quad

We formally define the different variants of the degree-constrained spanning tree problems studied in this work.

\problembox{(Weighted) Set of Degrees MST}{Graph $G = (V,E)$, edge weight $w: E \to \mathbb{R}_{\geq 0}$, function $D:V(G)\rightarrow 2^{\mathbb{N}}$, bound $\costbd \in \mathbb{R}_{\geq 0}$}{Does $G$ admit a spanning tree $T$ of cost at most $\costbd$ such that $\deg_T(v)\in D(v)$ for each $v\in V(G)$?}

\textsc{Bounded Degree MST} and \textsc{Specified Degree MST} are defined analogously, but there the images of $D$ are sets of the form $[i]$ or $\{i\}$ (for some integer $i$), respectively. We use $r$ to denote $\max_{v\in V(G),x\in D(v)}x$.
For the unweighted versions of the problems above, we omit the edge weights $w$ and the cost bound $\costbd$ from the input and hence instances merely consist of a pair $(G,D)$. 
The unweighted versions of all three considered problems are \NP-hard as they admit a straightforward reduction from \textsc{Hamiltonian Path}.
Clearly,  (weighted) \textsc{Set of Degrees MST} generalizes both (weighted) \textsc{Bounded Degree MST} and (weighted) \textsc{Specified Degree MST} by definition. There is also a trivial weight- and graph-preserving reduction from (weighted) \textsc{Specified Degree MST} to (weighted) \textsc{Bounded Degree MST} 
where instances of \textsc{Specified Degree MST} not satisfying $\sum_{v\in V(G)} d(v) = 2n-2$ can be rejected, while an instance that satisfies this equality can be turned into an equivalent instance of \textsc{Bounded Degree MST} by replacing $D(v) = \{i\}$ for each $v\in V(G)$ with $D(v) = [i]$.

\smallskip
\noindent \textbf{Structural Parameters.}\quad
We now define the graph parameters considered in this work. 
\smallskip

\noindent \emph{Treewidth and pathwidth.}\quad 
A \defi{tree decomposition} of a graph $G=(V,E)$ is a pair $(\dtree, \bags)$ where $\dtree$ is a tree and $\bags:V(\dtree)\rightarrow 2^V$ assigns the nodes of $\dtree$ to subsets of $V$ such that
         \begin{itemize}
        \item for each $v\in V$ there exists $x\in V(\dtree)$ with $v \in \bags(x)$,
        \item for each $e\in E$ there exists $x\in V(\dtree)$ with $e\subseteq \bags(x)$,
        \item for each $v\in V$ the set $\{x\in V(\dtree)\colon v \in \bags(x)\}$ induces a connected subgraph of $\dtree$.
    \end{itemize}
         A \defi{path decomposition} is a tree decomposition $(\dtree, \bags)$, where $\dtree$ is a simple path.
    When the decomposition $(\dtree, \bags)$ is clear from context, we denote $B_v = \bags(v)$, and $\nodes = V(\dtree)$.
    Further, we will sometimes denote the decomposition by $\dtree$ and assume that $\bags$ is implicitly given.

    We define the width of a tree decomposition $(\dtree, \bags)$ as  $\tw(\dtree,\bags) = \max_{x\in\nodes} |\bags(x)|-1$. The \defi{treewidth} (\defi{pathwidth}) of a graph $G$, denoted by $\tw(G)$ ($\pw(G)$) is the smallest width of a tree decomposition (path decomposition) of $G$.
    Finally, a path decomposition can be defined by a sequence of sets $B_1, \dots, B_h$ that corresponds to the path decomposition $(P_h, \mathcal{B})$, where $P_h$ is a simple path over the vertices $v_1, \dots, v_h$, and $\mathcal{B}=\{B_i\colon i\in[h]\}$, where $B_i$ is the bag corresponding to the node $v_i$ for $i\in[h]$.

In the following, we define the notion of a nice tree decomposition. Different (close) notions of nice decompositions were defined in the literature, the first of which was defined by Kloks \cite{DBLP:books/sp/Kloks94}. It is well-known that given a tree (or path) decomposition of a graph $G$, one can construct a nice tree (or path, respectively) decomposition of the same width and polynomial size in $|V|$ in polynomial time \cite{DBLP:books/sp/Kloks94}. Such a decomposition can be easily modified to meet our requirements in the same running time.
We will use the following notion of a nice tree decomposition to simplify the presentation of our algorithms; note that we assume, w.l.o.g., that leaf bags are empty and all vertices are introduced later on.

\begin{definition}[Nice tree decomposition]
    A \defi{nice tree decomposition} $(\dtree, \mathcal{B})$ of a graph $G = (V,E)$ is called a tree decomposition of $G$ where $\dtree$ is rooted at a node $r\in\nodes$, and it holds for each $x\in\nodes$ that $x$ has one of the following types:
    \begin{itemize}
        \item A leaf node, where $x$ is a leaf of $\dtree$. We have $B_x = \emptyset$ and define $G_x = (V_x, E_x)$ as an empty graph.
        \item An introduce vertex node $(v \in V)$, where $x$ has a single child $x'$ and $B_x = B_{x'}\cup \{v\}$. We define $G_x = (V_x, E_x)$ as the graph obtained from $G_{x'}$ by adding the vertex $v$ to $V_{x'}$.
        \item An introduce edge node $(e \in E)$, where $x$ has a single child $x'$ and $B_x = B_{x'}$. We define $G_x = (V_x, E_x)$ as the graph obtained from $G_{x'}$ by adding the edge $e$ to $E_{x'}$.
        \item A forget vertex node $(v \in V)$, where $x$ has a single child $x'$ and $B_{x'} = B_x\cup \{v\}$. We define the graph $G_x = G_{x'}$.
        \item A join node, where $x$ has exactly two children $x_1$ and $x_2$, and $B_{x_1} = B_{x_2} = B_x$. We define $G_x$ as the graph resulting from the disjoint union of $G_{x_1}$ and $G_{x_2}$ by identifying both copies of each vertex $v\in B_x$ (i.e., we consider each copy as a distinct vertex when computing the union, and after that we identify the copies).
    \end{itemize}
    Moreover, it holds for a nice tree decomposition that $B_r = \emptyset$ and for each $e\in E$ there exists exactly one node $x$ that introduces $e$. Note that $G_r$ is isomorphic to $G$. A \defi{nice path decomposition} is a nice tree decomposition that does not have any join nodes.
\end{definition}

\noindent \emph{NLC-Width.}\quad 
To present our algorithmic result for clique-width we will use a closely related graph measure called \defi{NLC-width}, which we define below. 
A labeled graph $G$ is a graph equipped with a vertex labeling $\texttt{lab}_G : V(G)\rightarrow \mathbb{N}$ that assigns to each vertex an integer. We call a labeled graph a $k$-graph, if its labels are from the set $[k]$.
A \emph{(labeled) subgraph} of a labeled graph $G$ is a labeled graph that is a subgraph of $G$ that preserves the labels of vertices.

An \defi{initial $k$-graph} $\bullet_i$ is a graph consisting of a single vertex labeled $i$. Given an \defi{edge mapping} $\alpha\subseteq [k]^2$ and a \defi{relabeling function} $\beta: [k]\rightarrow [k]$, we define the \defi{join operation} $\oplus_\alpha^\beta$ (where we omit the superscript if $\beta$ is the identity) as follows.\footnote{We remark that in the literature, the relabeling functions are typically applied via separate operations after joins; this is merely a cosmetic difference that streamlines the presentation of our algorithm.}
For two $k$-graphs $G_1$ and $G_2$, $G_1 \oplus_\alpha^\beta G_2$ is the $k$-graph obtained by:
\begin{enumerate}
\item performing the disjoint union of $G_1$ and $G_2$, then
\item adding an edge between every vertex $v_1\in V(G_1)$ and every vertex $v_2\in V(G_2)$ such that $(\texttt{lab}(v_1),\texttt{lab}(v_2))\in \alpha$, and finally
\item for each $v\in V(G_i)$, changing the label of $v$ to $\beta(\lab(v))$.
\end{enumerate}

Naturally, join operations can be chained together, and we call an algebraic expression using initial $k$-graphs as atoms and join operations as binary operators an NLC-decomposition of the graph produced by the expression. A graph $G$ has \defi{NLC-width} $k$ if $k$ is the minimum integer such that (some vertex labeling of) $G$ admits an NLC-decomposition involving at most $k$ labels. 
We denote by $\syntaxtree$ the tree corresponding to the NLC-decomposition and by $\nodes=V(\syntaxtree)$ the set of all nodes of the tree. For $x\in\nodes$, we denote by $G_x$ the graph obtained by the algebraic expression $x$ and by $V_x$ the set $V(G_x)$. For some set $S$ and a function $f:V(G)\rightarrow S$, we write $f_x$ for $f|_{V_x}$.

It is well known that for every graph $G$ with NLC-width \texttt{NLCw} and clique-width \texttt{cw}, $\texttt{NLCw} \leq \texttt{cw} \leq 2\cdot \texttt{NLCw}$~\cite{johansson1998clique,BodlaenderH12}; 
  moreover, this relationship is constructive and so known algorithms for computing approximately-optimal decompositions for clique-width~\cite{FominK24} 
 also yield approximately-optimal NLC-decompositions. Intuitively, NLC-width can be seen as an analogue of clique-width where edge addition and relabeling occurs in a ``single-shot fashion'' rather than gradually.

\smallskip
\noindent \emph{Cutwidth.}\quad
A linear arrangement $\ell = v_1,\dots, v_n$ of a graph $G$ is a linear ordering of the vertices of $G$. Let $V_i = \{v_1, \dots, v_i\}$, and $\overline{V}_i = V\setminus V_i$. We define the cut-graph $H_i = G[V_i, \overline{V}_i]$ as the bipartite graph induced from $G$ by the edges having exactly one endpoint in $V_i$. The width of a linear arrangement $\ell$ is defined by $w(\ell) =\max\{|E(H_i)|\colon i\in[n]\}$. The cutwidth of a graph $G$ (denoted $\ctw(G)$) is defined as the smallest width of a linear arrangement of $G$. A linear arrangement of minimum width can be computed in time $2^{\bigoh(\ctw(G)^2)}\cdot n^{\bigoh(1)}$~\cite{GiannopoulouPRT19}, and there is no single-exponential algorithm for the problem unless the ETH fails~\cite{KorhonenL23}.

\section{Dense Graph Parameters}\label{sec:cw}
\noindent
For the remainder of this section,
let $(G, D)$ be a given unweighted instance of \textsc{Set of Degrees MST},
where $G$ is an $n$-vertex graph of NLC-width $k$, and let us consider a given NLC-decomposition 
$\phi$ of $G$ using $k$ labels. 

  In this section, all vectors are integer-valued and $k$-dimensional (unless specified otherwise).
For a node $x \in \nodes$, we denote $\lab_x = \lab_{G_x}$. 
The aim of this section is to establish the following theorem:

\cwub*

\subsection{Partial Solutions and Patterns}

Let us first define the set of objects that we are working with. Firstly, as usual, we have the notion of partial solutions. In order to define partial solutions, we need the notion of fixed forests. A fixed forest is a standalone graph with a requirement function that specifies how much the degree of each vertex will  increase ``during further processing'' (this is different from the requirement function of the problem definition). 

 \begin{definition}
A \defi{fixed forest} is a tuple $(F,g)$ of a $k$-graph $F$ that is a forest and a function $g : V(F) \to \mathbb{N}_0$.
For an instance $(H, D')$, we define a \defi{partial solution} $(F,g)$ in $(H, D')$ as a fixed forest where $F$ is a spanning forest of $H$ and 
$g(v)+\deg_F(v)\in D'(v)$.
\end{definition}

We represent each fixed forest by a multiset of vectors, where we have one vector per connected component. This vector captures how many edges we need to add to vertices of each label in the given connected component. 

 \begin{definition}
         A \defi{$k$-pattern} is a multiset of at most $n$ vectors in $([n]_{0})^k$ (we omit $k$ when clear from context). Let $\mathcal{P}$ be the family of all patterns.
    For a fixed forest $(F,g)$, its \defi{corresponding pattern} $\repr(F,g)$ is the multiset that contains exactly the vectors $v(C)$ for each connected component $C$ of $F$, where $v(C)$ is defined by
    \[\big(v(C)\big)_i=\sum_{\substack{u\in V(C),\\\lab(u)=i}} g(u), \text{ for $i \in [k]$}.\]
\end{definition}

\begin{remark}
    Except in the description of \cref{def:cw_algo}, all patterns we handle in this work are $k$-patterns where $k$ is the width of the NLC-expression as defined above.
\end{remark}

The definition above maps each fixed forest to a pattern. Next, we define a mapping in the other direction, i.e., a mapping that assigns a canonical fixed forest to each pattern. 

 \begin{definition}\label{def:canonical_forest}
Given a pattern $A \in \mathcal{P}$, we define the \defi{canonical fixed forest}, denoted $(F_A, f_A)$, as the fixed forest containing one path for each vector $v\in A$ constructed as follows.
 For each $i \in [k]$, we create a vertex $u^A_{v,i}$ with label $i$ and assign $f_A(u^A_{v,i}) = v_i$.
Next, we construct the path $(u^A_{v, 1},...,u^A_{v, k})$.
Finally, we obtain the $k$-graph~$F_A$ by taking the disjoint union of all such paths.
\end{definition}

In our dynamic programming table, at each node of the $k$-expression, the records that we store will contain the corresponding patterns of all partial solutions.

 \begin{definition}\label{def:record}
Given an instance $(H, D')$ we define the \defi{record} $\mathcal{R}_{(H, D')}$ of $(H, D')$ as the set of the patterns corresponding to all partial solutions of $(H, D')$. 
For $x\in \nodes$, we define $\mathcal{R}_x=\mathcal{R}_{(G_x, D_{x})}$. 
\end{definition}

\para{Compatibility and Equivalence.}
Let us now introduce the notion of compatibility and equivalence of the objects that we are working with.
We aim to replace partial solutions with simpler ones that behave the same, i.e., they can be extended to a solution in the same way. Informally, two fixed forests are compatible if they can be combined to obtain a spanning tree satisfying both requirement functions simultaneously.

 \begin{definition}    
\label{def:cw_compatible_forests}
For two fixed forests $(F, g)$ and $(F', g')$ and an edge mapping $\alpha\subseteq [k]^2$, an \defi{$\alpha$-spanning tree}
for $(F, g)$ and $(F', g')$ is a spanning tree~$T$ of $F\oplus_{\alpha} F'$ such that $E(F), E(F')\subseteq E(T)$, and $\deg_T = (g+\deg_F)\cup (g'+\deg_{F'})$.
We say that $(F, g)$ and $(F', g')$ are \defi{$\alpha$-compatible} if there exists an $\alpha$-spanning tree
for $(F, g)$ and $(F', g')$.
\end{definition}

Let us now observe some basic properties of $\alpha$-compatibility.
Firstly, note that $\alpha$-compatibility is not symmetric, i.e., if $(F, g)$ and $(F', g')$ are $\alpha$-compatible, that does not imply that $(F', g')$ and $(F, g)$ are $\alpha$-compatible.
(Instead, $(F', g')$ and $(F, g)$ are $\alpha^R$-compatible, where $\alpha^R = \{(y, x) \, : \, (x, y) \in \alpha\}$.) It is also easy to see that if $(F,g)$ and $(F',g')$ are $\alpha$-compatible, then $(F,g)$ and $(F',g')$ are $\alpha'$-compatible for all $\alpha'\supseteq \alpha$.

As we will see in the later proofs, it is sometimes useful to work with fixed forests with disjoint sets of labels. In particular, we observe the following:

\begin{observation}\label{obs:cw_label_shifting}
    Let $(F,g)$ and $(F',g')$ be two fixed forest with labels in $[k]$ and let $\alpha\subseteq [k]^2$. Let $F_{\text{rel}}$ be the graph with labels in $\{k+1,\dots, 2k\}$ obtained by increasing the label of each vertex in $F$ by $k$ and $\alpha_{\text{rel}}=\{(i+k,j)\colon (i,j)\in\alpha\}$. We have that $(F,g)$ and $(F',g')$ are $\alpha$-compatible if and only if $(F_{\text{rel}}, g)$ and $(F',g')$ are $\alpha_{\text{rel}}$-compatible. 
\end{observation}

Towards deriving reduction rules later, we now define the notion of equivalence, denoted by $\simeq$, for a range of different objects, starting from fixed forests. Informally, fixed forests $(F, g)$ and $(F', g')$ are equivalent if for all $\alpha$,  $\alpha$-compatibility with $(F, g)$ implies $\alpha$-compatibility of $(F', g')$ and vice versa.
If the implication only goes one way, we can speak of weaker fixed forests. 
We can extend the notion of weakness and equivalence to sets of fixed forests.
  \begin{definition}
\label{def:cw_equiv_ff}
    A set $R$ of fixed forests is \defi{weaker} than a set $R'$ of fixed forests, denoted by $R \preceq R'$, if for each edge mapping $\alpha\subseteq [k]^2$ and for any fixed forest $(F^*, g^*)$ it holds that if $(F^*,g^*)$ is $\alpha$-compatible with some fixed forest $F\in R$, then it is $\alpha$-compatible with some fixed forest $F' \in R'$.
    We call $R$ and $R'$ \defi{equivalent} ($R \simeq R'$), if both $R \preceq R'$ and $R' \preceq R$ hold.

    Given two fixed forests $(F,g)$ and $(F',g')$, we say that $(F,g)$ is \defi{weaker} than (resp., \defi{equivalent} to) $(F',g')$ if the singleton set $\{(F,g)\}$ is weaker than (resp., equivalent to) the singleton set $\{(F',g')\}$.
\end{definition}

Definition~\ref{def:cw_equiv_ff} also allows us to extend the notion of equivalence to patterns and sets of patterns. Recall that each pattern $A$ is associated with a canonical fixed forest $(F_A, f_A)$.

\begin{definition}
\label{def:equiv_patt}
    Let $A, A'\in\mathcal{P}$. We say $A$ is \defi{weaker} than $A'$ if $(F_A, f_A)$ is weaker than $(F_{A'}, f_{A'})$.
   We call $A$ and $A'$ \defi{equivalent}, if their canonical fixed forests are equivalent, i.e., $(F_A, f_A)\simeq (F_{A'}, f_{A'})$. 

   Let $\mathcal{R}, \mathcal{R}'\in 2^\mathcal{P}$ and let $F_\mathcal{R}=\{(F_A, f_A):\: A\in \mathcal{R} \}$, $F_{\mathcal{R}'}=\{(F_A, f_A):\: A\in \mathcal{R}' \}$. We say that $\mathcal{R}$ is \defi{weaker} than $\mathcal{R}'$ if $F_\mathcal{R}$ is weaker than $F_{\mathcal{R}'}$. We say $\mathcal{R}$ and $\mathcal{R}'$ are \defi{equivalent} if $F_\mathcal{R}\simeq F_{\mathcal{R}'}$.
\end{definition}

It is easy to see that all notions of equivalence above are equivalence relations.
The above definitions allow us to speak about equivalent records. Intuitively, two instances having equivalent records means that they behave the same, i.e., we can replace an instance by a simpler instance that has an equivalent record. 

Corollary~\ref{cor:cw_equiv_repr} below states that in order to prove that two fixed forests are equivalent, it suffices to look at their patterns.
By definition, two patterns are equivalent if their \emph{canonical} fixed forests are equivalent.
We show that the reverse is true: two fixed forests are equivalent if their corresponding patterns are equivalent.

For convenience, we first prove this for the case when the patterns are the same in the following lemma.

\begin{longlemma}
\label{lem:cw_same_repr}
    Let $(F, g), (F',g')$ be two fixed forests. If $\repr(F,g) = \repr(F',g')$, then $(F,g)\simeq (F',g')$.
\end{longlemma}
 \begin{proof}
    Let $R=\repr(F, g)=\repr(F', g')$. Fix an edge mapping $\alpha\in [k]^2$ and a fixed forest $(\tilde{F}, \tilde{g})$. Suppose that $(\tilde{F}, \tilde{g})$ is $\alpha$-compatible with $(F, g)$, and let $T$ be an $\alpha$-spanning tree for $(\tilde{F}, \tilde{g})$ and $(F,g)$. 

    Note that for each pair of connected components $C$ in $F$ and $\tilde{C}$ in $\tilde{F}$, there is at most one edge in $T$ between them (otherwise we would have a cycle in $T$).
    Further, let $\hat{E} = E(T) \setminus (E(F) \cup E(\tilde{F}))$.
    Since $\deg_T = (g+\deg_F)\cup (\tilde{g}+\deg_{\tilde{F}})$, it holds that 
    \begin{itemize}
        \item[(*)] For every $i \in [k]$ and every connected component $C$ of $F$ corresponding to a vector $w$ in $R$, we have $\sum_{\substack{x \in C\\ \lab_F(x) = i}} \deg_{\hat{E}}(x) = w_i$. 
    \end{itemize} 

    We initialize $T'$ as an empty graph on $V(F') \cup V(\tilde{F})$. 
    For each connected components $C$ of $F$ and $\tilde{C}$ of $\tilde{F}$, if there is an edge $uv$ such that $u \in C$, $v \in \tilde{C}$, $\lab_F(u)=i$, $\lab_{\tilde{F}}(v) = j$, for some $i,j \in [k]$, then we add an edge to $T'$ as follows.
    Let $C'$ be the connected component of $F'$ that corresponds to the same (copy of a) vector $w \in R$ as $C$. 
    Find a vertex $u' \in C'$ such that $\lab_{F'}(u') = i$ and $u'$ is incident to less than $g'(u')$ edges in $T'$ up to now.
    We then add the edge $u'v$ to $T'$.

    We show that such a $u'$ exists.
    By construction, before adding $u'v$, the total number of edges in $T'$ incident to a vertex of label $i$ in $C'$ is less than $\sum_{x \in C : \lab_F(x) = i} \deg_{\hat{E}}(x)$.
    By (*), this number is less than $w_i$.
    Since $C$ and $\tilde{C}$ correspond to the same vector $w \in R$, $w_i = \sum_{x \in C' : \lab_{F'}(x) = i} g'(i)$.
    Hence, such a $u'$ must exist.

    Further, by the same argument above, after processing all pairs in $\cc(F') \times \cc(\tilde{F})$, the graph $T'$ obtained up till now satisfies that $\deg_{T'} = g' \cup \tilde{g}$.
    Further, by construction it is also a subgraph of $F' \oplus_{\alpha} \tilde{F}$ and does not contain any edge in $E(F') \cup E(\tilde{F})$.
    
    Therefore, as the final step, we add $E(F') \cup E(\tilde{F})$ to $T'$.
    We now have $\deg_{T'} = (g' + \deg_{F'}) \cup (\tilde{g} + \deg_{\tilde{F}})$.
    Further, $T'$ is a tree, since a cycle in $T'$ would correspond to a cycle in $T$ by the construction above.
    Hence $T'$ is an $\alpha$-spanning tree for $(\tilde{F},\tilde{g})$ and $(F',g')$.

    The above implies that $(F,g) \preceq (F',g')$.
    By an analogous argument, we also obtain $(F',g') \preceq (F,g)$.
    Hence, $(F,g) \simeq (F',g')$, as required.
\end{proof} 
\begin{corollary}
\label{cor:cw_equiv_repr}
    Let $(F, g), (F',g')$ be two fixed forests. If $\repr(F,g)\simeq\repr(F',g')$, then $(F,g)\simeq (F',g')$.
\end{corollary}
\begin{proof}
    By Lemma~\ref{lem:cw_same_repr}, $(F,g) \simeq (F_{\repr(F,g)}, f_{\repr(F,g)})$, and $(F',g') \simeq (F_{\repr(F',g')},f_{\repr(F',g')})$.
    By \cref{def:equiv_patt}, $(F_{\repr(F,g)}, f_{\repr(F,g)}) \simeq (F_{\repr(F',g')}, f_{\repr(F',g')})$.
    Hence, $(F,g) \simeq (F',g')$.
\end{proof} 
\para{Nice Patterns and Reduction Rules.}

In order to get the claimed running time, we have to reduce the set of records in our dynamic programming tables into a smaller set. In order to achieve this we define the family of nice patterns, whose size is bounded by $n^{\bigoh(k)}$. We describe reduction rules that can be used to turn any set of patterns into an equivalent set of nice patterns efficiently.

 \begin{definition}
We call a pattern $A$ \defi{nice} if $A$ contains at most one $\mathbf{0}$ vector and for each coordinate $i\in[k]$, $A$ contains at most one vector $v$ such that $v$ is not a unit vector and $v_i\geq 1$.
We denote by $\mathcal{P}^*\subseteq \mathcal{P}$ the family of all nice patterns.
\end{definition}

\begin{observation}
\label{obs:sizenice}
It holds that $|\mathcal{P}^*|\leq n^{\bigoh(k)}$, since the set of the non-unit non-zero vectors forms a partition of a subset of the labels.
\end{observation}

\subsection{Requirement Shifting and Pattern Representation}
\label{sub:cw-smartshiftingetc}

Now we define a new structure called ``ordered pattern'' that allows us to identify and manipulate specific vectors of a pattern more systematically.

\begin{definition}\label{def:ord_pat}
    An \defi{ordered pattern} is a pair $\ord = (O,b)$, where $O = O_1,\dots, O_m$ is a sequence of vectors
    and $b : [k] \rightarrow [m]\cup\{\perp\}$ is a mapping that assigns to each coordinate $\ell\in[k]$ an index $j\in [m]$ such that $(O_j)_{\ell} \geq 1$, or $\perp$ if no such index exists. We then call $j$ the \defi{big index} of $\ell$ with respect to $b$, and we call $b$ a \defi{big-vector mapping} of $O$.

    Given a pattern $A\in\mathcal{P}$, an \defi{ordered pattern of $A$} is an ordered pattern $\ord_A = (O_A, b_A)$, where $O_A$ is a sequence of the vectors in $A$. Given an ordered pattern $(O, b)$ over $m$ vectors, we define the \defi{underlying pattern} $p(O,b)$ as the multiset $\langle O_j \colon j\in [m] \rangle$. 
    We call a family of ordered patterns $\ord$ weaker than (or equivalent to) another family $\ord'$, if the underlying pattern $p(\ord)$ is weaker than (or equivalent to) the underlying pattern $p(\ord')$.
\end{definition}

We first define bad indices and bad coordinates as the ``obstacles'' that we need to remove from an ordered pattern to make it nice:

\begin{definition}
    Let $\ord = (O, b)$ be an ordered pattern over $m$ vectors with $O = O_1, \dots, O_m$.    
    For $i\in[m]$ we define the set of \defi{bad coordinates} of $i$ as $Z_{\ord}(i) = \emptyset$ if $O_i$ is a unit or zero vector. Otherwise, we define
    \[
    Z_{\ord}(i) = \big\{\ell\in[k]\colon b(\ell)\neq i \land (O_i)_{\ell}\geq 1\big\}.
    \]

    We define the set of \defi{bad indices} of $\ord$ as 
    $I_{\ord} \coloneq \big\{i\in[m]\colon Z_{\ord}(i) \neq \emptyset \big\}$.
\end{definition}

\begin{observation}\label{obs:cw_ord_nice}
    Let $\ord$ be an ordered pattern with at most one zero vector, such that $I_{\ord} = \emptyset$. Then $\ord$ is nice.
\end{observation}

Now we introduce two operations that allow to ``shift'' requirements between vectors in an ordered pattern. Intuitively, this corresponds to rerouting edges in a solution. Essentially, we show that the patterns resulting from these operations on a pattern correctly replace this pattern, i.e., the initial pattern behaves like an ``OR'' of the resulting patterns, which means that the initial pattern can be extended to a solution if and only if at least one of the resulting patterns can be extended to a solution.

\begin{definition}
\label{def:cw_pi}
    We define the functions \defi{$\rrule_1, \rrule_2$} that take as input an ordered pattern $\ord=(O,b)$ over $m$ vectors, a bad index $i\in I_\ord$ and a coordinate $\ell\in Z_\ord(i)$, and output ordered patterns $\rrule_1(O, b, i, \ell) \coloneqq (O^1, b^1)$ and $\rrule_2(O, b, i, \ell)\coloneqq(O^2, b^2)$, defined as follows. Let $j=b(\ell)$. For $h\in[m]$ we have
    \[
    (O^1)_h = \begin{cases}
        O_j + O_i - \idv_{\ell} & \colon h = j,\\
        \idv_{\ell} & \colon h = i,\\
        O_h & \colon \text{otherwise,}
    \end{cases}
     (O^2)_h = \begin{cases}
        O_j + (O_i)_{\ell} \cdot \idv_{\ell} & \colon h = j,\\
        O_i - (O_i)_{\ell} \cdot \idv_{\ell} & \colon h = i,\\
        O_h & \colon \text{otherwise.}
    \end{cases}
    \]
    We define $b^2 = b$, and $b^1$ results from $b$ by assigning each coordinate $\ell' \in [k]$ with $b(\ell') = i$ to $j$ instead, i.e., $b^1(\ell') = j$ if $b(\ell') = i$, and $b^1(\ell') = b(\ell')$ otherwise.
\end{definition}
Note that applying $\rrule_1, \rrule_2$ only affects vectors $O_i$ and $O_j$. In both cases, we subtract a specific vector from $O_i$ and add it to $O_j$. We call this process \emph{shifting} from index $i$ to index $j$ in $(O,b)$.
In \cref{fig:cw_example_rrule2} we provide an example describing how the canonical fixed forest of a pattern is transformed by $\rrule_1$ or $\rrule_2$.
\begin{figure}[t!]
    \centering
    \includegraphics[page=3]{graphics/cw_figures.pdf}
    
    \vspace{2em}
    \includegraphics[page=1]{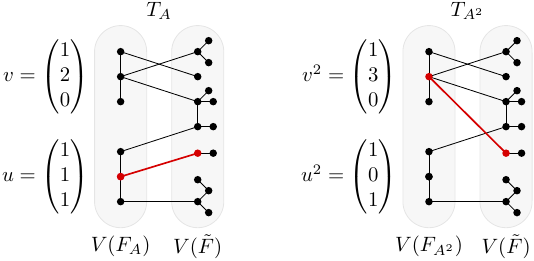}
    \caption{
    Let $(O,b)$ be the ordered pattern such that $O = v, u$ for $v, u$ as given in the figure, $b(1) = b(2) = 1$ and $b(3) = 2$. 
    Let $A = p(O,b)$, $A^1 = p(\rrule_1(O, b, 2, 2))$ and $A^2 = p(\rrule_2(O, b, 2, 2))$.
     For some edge-mapping $\alpha$, fixed forest $(\tilde{F}, \tilde{g})$, and $\alpha$-spanning tree $T_A$ for $(F_A, f_A)$ and $(\tilde{F}, \tilde{g})$, if $T_A$ contains a path that connects the tree corresponding to $v$ with the second vertex of that corresponding to $u$ (see the bold path in the top left figure), we can obtain an $\alpha$-spanning tree $T_{A^1}$ for $(F_{A^1}, f_{A^1})$ and $(\tilde{F}, \tilde{g})$ by ``moving'' the blue edges (i.e., edges incident to the tree corresponding to $u$ and not on the path); see top right figure.
    If there is no such path, then we can ``move'' the red edge incident to the second vertex of the tree corresponding to $u$ instead; see bottom figures. This gives an intuitive explanation of why~$\{A\}\preceq \{A^1, A^2\}$.}
    \label{fig:cw_example_rrule2}
\end{figure}
 We use the following lemma to prove the correctness of the reduction rules $\rrule_1$ and $\rrule_2$.
\begin{longlemma}
\label{lem:low_level_shifting}
Let $\alpha\subseteq [k]^2$ be an edge mapping, $(\tilde{F}, \tilde{g})$ be a fixed forest, $A$ be a pattern,
$x, y$ be vectors in $A$,
$T$ be an $\alpha$-spanning tree for $(\tilde{F},\tilde{g})$ and $(F_{A}, f_A)$,
and $P$ be a path from a vertex of $\{ u^A_{x,\ell} \colon \ell \in [k] \}$ to a vertex of $\{ u^A_{y,\ell} \colon \ell \in [k] \}$ in $T$.
Further, let $z'$ be a vector with $z'_\ell = 
 |(N_{T}(u^A_{x,\ell}) \cap V(\tilde{F})) \setminus V(P)|$ for $\ell \in [k]$,
 $z$ be a vector such that $z \leq z'$, 
and let $A^* = (A \setminus \langle x,y \rangle) \cup \langle x - z, y + z \rangle$.
 Then there is an $\alpha$-spanning tree $T^*$ for $(\tilde{F}, \tilde{g})$ and $(F_{A^*}, f_{A^*})$.
\end{longlemma}
\begin{proof}
    Let $\hat{E} = E(T) \setminus (E(\tilde{F}) \cup E(F_{A}))$.
    Further, let $X = \{u^{A}_{x,\ell} \colon \ell \in [k]\}$, $Y = \{u^{A}_{y,\ell} \colon \ell \in [k]\}$, $X^* = \{u^{A^*}_{x-z,\ell} \colon \ell \in [k]\}$, and $Y^* = \{u^{A^*}_{y+z,\ell} \colon \ell \in [k]\}$.

    Intuitively, we obtain $T^*$ from $T$ as follows.
    For $\ell \in [k]$, we choose an arbitrary set of $z_\ell$ edges incident to $u^A_{x,\ell}$ in $\hat{E} \setminus E(P)$ and replace the endpoint $u^A_{x,\ell}$ of these edges with $u^A_{y,\ell}$.
    After that, we rename the vertices in $F_{A}$ to get the vertices of $F_{A^*}$.

    Formally, we define a mapping $\lambda: V(\tilde{F} \oplus_{\alpha} F_{A}) \to V(\tilde{F} \oplus_{\alpha} F_{A^*})$ as follows.
    For $w \in V(\tilde{F}) \cup V(F_{A})$,
    \[
    \lambda(w)=
    \begin{cases}
        u^{A^*}_{x-z,\ell} & \text{if $w = u^{A}_{x,\ell}$ for some $\ell \in [k]$,} \\
        u^{A^*}_{y+z,\ell} & \text{if $w = u^{A}_{y,\ell}$ for some $\ell \in [k]$,} \\
        u^{A^*}_{x',\ell} & \text{if $w = u^{A}_{x',\ell}$ for some $x' \in A \setminus \langle x,y \rangle$, $\ell \in [k]$,} \\
        w & \text{if $w \in V(\tilde{F})$.}
    \end{cases}
    \]    
    
    We then construct $T^*$ from $T$ in three steps.
    In Step 1, for each $\ell \in [k]$, we arbitrarily mark $z_{\ell}$ edges in $\hat{E} \setminus E(P)$ incident to $u^A_{x,\ell}$.
    (Since $z \leq z'$, at least $z_{\ell}$ such edges exist.)
    For each marked edge $u^A_{x,\ell} w$, we add an edge $u^{A^*}_{y+z,\ell}w$ to $T^*$.
    In Step 2, for each unmarked edge $ww'$ in $\hat{E}$, we add the edge $\lambda(w)\lambda(w')$ to $T^*$.
    Finally, in Step 3, we add $E(\tilde{F}) \cup E(F_{A^*})$ to $T^*$.
    
    For $t \in [3]$, let $E^*_t$ be the set of edges added to $T^*$ in Step $t$.
    $E^*_3$ is trivially disjoint from $E^*_1$ and $E^*_2$.
    We show that $E^*_1$ and $E^*_2$ are also disjoint.
    For the sake of contradiction, suppose that this is not the case.
    Then there must be an edge $u^{A^*}_{y+z,\ell}w$ that is in both $E^*_1$ and $E^*_2$.
    This implies that $u^A_{x,\ell}w$ and $u^A_{y,\ell}w$ are edges in $T$.
    However, $u^A_{x,\ell}$ and $u^A_{y,\ell}$ then have two distinct paths in $T$ between them, one being $(u^A_{x,\ell},w,u^A_{y,\ell})$ and the other being the concatenation of the path in $X$ between $u^A_{x,\ell}$ and an endpoint of $P$, the path $P$, and the path in $Y$ between the other endpoint of $P$ and $u^A_{y,\ell}$, a contradiction.
    Since $E^*_1$, $E^*_2$, and $E^*_3$ are pairwise disjoint, and since $T$ is an $\alpha$-spanning tree for $(\tilde{F},\tilde{g})$ and $(F_A, f_A)$, it is easy to verify that $\deg_{T^*} = (\tilde{g} + \deg_{\tilde{F}}) \cup (f_{A^*} + \deg_{f_{A^*}})$.
    
    It is also easy to see that $T^*$ is a subgraph of $\tilde{F} \oplus_{\alpha} F_{A^*}$ with $E(\tilde{F}), E(F_A) \subseteq E(F_{A^*})$.
    Hence, it remains to show that $T^*$ is a spanning tree of $\tilde{F} \oplus_{\alpha} F_{A^*}$.
    Let $P^*$ be the path obtained from $P$ by replacing any vertex $w$ on the path by $\lambda(w)$.
    Let $U = V(P) \cup X \cup Y$ and $U^* = V(P^*) \cup X^* \cup Y^*$.
    Observe that $T^*[U^*]$ is exactly $P^* \cup F_{A^*}[X^* \cup Y^*]$ and is a tree.
    Hence, there is a path in $T^*[U^*]$ between any two vertices.
    Next, for two arbitrary distinct vertices $w$ and $w'$ in $\tilde{F} \oplus_{\alpha} F_A$, let $Q$ be the path between $w$ and $w'$ in $T$.
    If $Q$ contains no vertex in $U$, then replacing every vertex $w'' \in Q$ by $\lambda(w'')$, we obtain a path between $\lambda(w)$ and $\lambda(w')$ in $T^*$.
    Otherwise, let $w_1$ and $w_2$ be the first and last vertex in $Q$ that is also in $U$ (they may be identical).
    We then obtain a path between $\lambda(w)$ and $\lambda(w')$ in $T^*$ from $Q$ as follows.
    For each vertex $w''$ in $Q$ before $w_1$ or after $w_2$, we replace it by $\lambda(w'')$.
    If the edge preceding $w_1$ in $Q$ is marked, $w_1$ must be of the form $u^A_{x,\ell}$; we then replace $w_1$ in $Q$ by $u^{A^*}_{y+z,\ell}$.
    Otherwise, we replace $w_1$ by $\lambda(w_1)$.
    Let $w^*_1$ be the vertex that we replace $w_1$ with.
    Similarly, we replace $w_2$ in $Q$ with $w^*_2$, which is either $u^{A^*}_{y+z,\ell}$ if the edge succeeding $w_2$ in $Q$ is marked and $w_2$ is of the form $u^A_{x,\ell}$, or $\lambda(w_2)$ otherwise.
    Then we remove all vertices between $w_1$ and $w_2$ in $Q$ and insert the path between $w^*_1$ and $w^*_2$ in $T^*[U^*]$.
    It is easy to verify that the obtained path is indeed a path in $T^*$.
    Since $\lambda$ is a bijection, this implies that $T^*$ is connected.
    Further, by construction, $|E(T^*)| = |E(T)| = |V(\tilde{F} \oplus_{\alpha} F_{A})|-1 = |V(\tilde{F} \oplus_{\alpha} F_{A^*})|-1$.
    Hence, $T^*$ is a spanning tree, as required.
\end{proof}

\begin{figure}[tb]
    \centering
    \includegraphics[page=2]{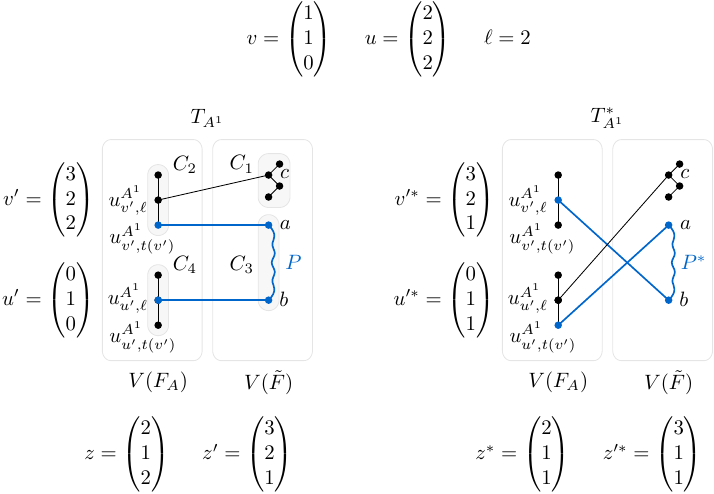}
    \caption{
    Example illustrating how $T_{A^1}^*$ is derived from $T_{A^1}$ in the proof of Lemma~\ref{lem:cw_equiv_rrule_1}.
    }
    \label{fig:cw_equiv_rrule_1}
\end{figure}

\begin{longlemma}\label{lem:cw_equiv_rrule_1}
    Let $\ord=(O,b)$ be an ordered pattern over $m$ vectors. For an index $i \in [m]$ and bad coordinate $\ell\in Z_{\ord}(i)$, let $(O^1, b^1) = \rrule_1(O, b, i, \ell)$. Then it holds that $(O^1, b^1) \preceq (O, b)$.
\end{longlemma}
\begin{proof}
    Let $j=b(\ell)$, $A = p(O, b)$, $A^1 = p(O^1, b^1)$, $u = O_i$, $v = O_j$, $u' = (O^1)_i$, $v' = (O^1)_j$.
    Note that $A^1 = (A \setminus \langle u, v \rangle) \cup \langle u', v' \rangle $ and that $(O^1, b^1) \preceq (O, b)$ is equivalent to $A^1 \preceq A$, which we will now show.
    
    Let $\alpha\subseteq [k]^2$ be an edge mapping, and let $(\tilde{F}, \tilde{g})$ be a fixed forest that is $\alpha$-compatible with the fixed forest of $A^1$, $(F_{A^1}, f_{A^1})$.
    We need to show that $(\tilde{F}, \tilde{g})$ is $\alpha$-compatible with the fixed forest of $A$, $(F_A, f_A)$, that is, there is an $\alpha$-spanning tree $T_A$ for $(\tilde{F}, \tilde{g})$ and $(F_A,f_A)$.

    By the $\alpha$-compatibility of $(\tilde{F}, \tilde{g})$ and $(F_{A^1}, f_{A^1})$, there is an $\alpha$-spanning tree for them.
    Let $P$ be a shortest path from $\{ u^{A^1}_{v',t} \colon t \in [k] \}$ to $\{ u^{A^1}_{u',t} \colon t \in [k] \}$ in $T_{A^1}$.
    Further, let $t(v')$ be the index of the first vertex of $P$ in $\{ u^{A^1}_{v',t} \colon t \in [k] \}$.
    We let $z'$ be a vector with $z'_t = 
 |(N_{T_{A^1}}(u^{A^1}_{v',t}) \cap V(\tilde{F})) \setminus V(P)|$ for $t \in [k]$.
 Observe that $z' = v' - \idv_{t(v')}$.

 We distinguish two cases.
 First, suppose that $z \coloneqq v' - v \leq z'$.
 Then, we can apply Lemma~\ref{lem:low_level_shifting} to $\alpha, (\tilde{F}, \tilde{g}), A^1, v', u', T_{A^1}, P$, and $z$. Since, $(u, v) = (u' + z, v' - z)$, this yields a spanning tree satisfying the conditions of this lemma.

 Otherwise, we proceed as follows.
 Observe that the last vertex of $P$ is $u^{A^1}_{u',\ell}$, since $u' = \idv_{\ell}$, and observe that the edge $u^{A^1}_{v',t(v')}u^{A^1}_{u',\ell}$ cannot be in $T_{A^1}$, as it is not in $\tilde{F} \oplus_\alpha F_{A}$. Hence $P$ counts at least three vertices; we denote the
    second vertex of $P$ by $a$, and the penultimate vertex of $P$ by $b$.
    Finally, let $c \in N_{T_{A^1}}(u^{A^1}_{v',\ell}) \cap V(\tilde{F})$, which exists since $v'_{\ell} \geq 1$.

    We obtain $P^*$ by replacing the first edge of $P$, $u_{v',t(v')}a$,
    with $u_{u',t(v')}a$, 
    and the last edge, $bu_{u',\ell}$, with $bu_{v',\ell}$.
    Further, we obtain $T_{A^1}^*$ from $T_{A^1}$ by performing the same two edge substitutions, and further, we replace the edge $u^{A^1}_{v',\ell}c$ with $u^{A^1}_{u',\ell}c$.
    See \cref{fig:cw_equiv_rrule_1}.

We define the vectors
\begin{align*}
    u'^* &= u' + \idv_{t(v')} \text{, and} \\
    v'^* &= v' - \idv_{t(v')}.
\end{align*}
Further, let $A^{1*} = (A^1 \setminus \langle u', v' \rangle) \cup \langle u'^*, v'^* \rangle$.
We will now argue that $T_{A^1}^*$ is an $\alpha$-spanning tree for $(\tilde{F}, \tilde{g})$ and $(F_{A^{1*}}, f_{A^{1*}})$, where we identify $u^{A^1}_{v',t} = u^{A^{1*}}_{v'^*,t}$ and $u^{A^1}_{u',t} = u^{A^{1*}}_{u'^*,t}$ for all $t \in [k]$, and $u^{A^1}_{w,t} = u^{A^{1*}}_{w,t}$ for all $w \in A^1 \setminus \langle u', v' \rangle$ and all $t \in [k]$.
First, observe that
    $T_{A^1} - \{u^{A^1}_{v',\ell}c, u^{A^1}_{v',t(v')}a, u^{A^1}_{u',\ell}b\}$ consists of four connected components $C_1,C_2,C_3,C_4$ with
    $c \in V(C_1)$,
    $u^{A^1}_{v',\ell}, u^{A^1}_{v',t(v')} \in V(C_2)$,
    $a,b \in V(C_3)$, and
    $u^{A^1}_{u',\ell} \in V(C_4)$.
Edge $u^{A^1}_{u',\ell}c$ connects components $C_1$ and $C_4$,
edge $u^{A^1}_{u',t(v')}a$ connects components $C_4$ and $C_3$, 
and
edge $u^{A^1}_{v',\ell}b$ connects components $C_3$ and $C_2$, 
yielding that $T_{A^1}^*$ is a tree spanning $V(T_{A^1})$.
Second, observe that for each edge $xy$ we replaced with $x'y'$,
we have $\lab(x) = \lab(x')$ and $\lab(y) = \lab(y')$, and thus, since all replaced edges are in $E(F_{A^1} \oplus_\alpha \tilde{F}) \setminus (E(F_{A^1}) \cup E(\tilde{F}))$,
so are all new edges.
Third, observe that the degree of $u^{A^1}_{v',t(v')}$ in $T_{A^1}^*$ is decreased by one compared to $T_{A^1}$, while the degree of $u^{A^1}_{u',t(v')}$ in $T_{A^1}^*$ is increased by one compared to $T_{A^1}$.
Thus $T_{A^1}^*$ is an $\alpha$-spanning tree for $(\tilde{F}, \tilde{g})$ and $(F_{A^{1*}}, f_{A^{1*}})$. Furthermore, observe that $P^*$ is a path connecting $u^{A^1}_{u',t(v')}$ and $u^{A^1}_{v',\ell}$ in $T_{A^1}^*$.
Towards applying Lemma~\ref{lem:low_level_shifting}, we define the vectors
\begin{align*}
    (z'^*)_{t} &= \big|\big(N_{T^*_{A^1}}(u^{A^1}_{v',t}) \cap V(\tilde{F})\big) \setminus V(P)\big| \text{ for } t \in [k] \text{, and}\\
    z^* &= v'^* - v,
\end{align*}
and observe that $z'^* = z' - \idv_{\ell}$.
It remains to show that $z^* \leq z'^*$.
By simple algebraic manipulations, the condition is equivalent to $v_{\ell} \geq 1$, which holds by precondition.
Finally, we apply Lemma~\ref{lem:low_level_shifting} to $\alpha, (\tilde{F}, \tilde{g}), A^{1*}, v'^*, u'^*, T_{A^1}^*, P^*$, and $z^*$. Since, $(u'^* + z^*, v'^* - z^*) = (u, v)$, we obtain a spanning tree satisfying the conditions of this lemma.
\end{proof}

\begin{longlemma}\label{lem:cw_equiv_rrule_2}
    Let $\ord=(O,b)$ be an ordered pattern over $m$ vectors. For an index $i \in [m]$ and a bad coordinate $\ell\in Z_{\ord}(i)$, let $(O^2, b^2) = \rrule_2(O, b, i, \ell)$. Then it holds that $(O^2, b^2) \preceq (O, b)$.
\end{longlemma}
\begin{proof}
    Let $j=b(\ell)$, $A = p(O, b)$, $A^2 = p(O^2, b^2)$, $u = O_i$, $v = O_j$, $u' = (O^2)_i$, $v' = (O^2)_j$.
    Note that $A^2 = (A \setminus \langle u, v \rangle) \cup \langle u', v' \rangle $ and that $(O^2, b^2) \preceq (O, b)$ is equivalent to $A^2 \preceq A$, which we will now show.
    
    Let $\alpha\subseteq [k]^2$ be an edge mapping, 
    and let $(\tilde{F}, \tilde{g})$ be a fixed forest that is $\alpha$-compatible with the fixed forest of $A^2$, $(F_{A^2}, f_{A^2})$.
    We need to show that there is an $\alpha$-spanning tree for $(\tilde{F}, \tilde{g})$ and $(F_{A}, f_{A})$.

    By the $\alpha$-compatibility of $(\tilde{F}, \tilde{g})$ and $(F_{A^2}, f_{A^2})$, there is an $\alpha$-spanning tree $T_{A^2}$ for $(\tilde{F}, \tilde{g})$ and $(F_{A^2}, f_{A^2})$.
    Let $P$ be a shortest path from $\{ u^{A^2}_{v',t} \colon t \in [k] \}$ to $\{ u^{A^2}_{u',t} \colon t \in [k] \}$ in $T_{A^2}$.
    Further, let $t(v')$ be the index of the first vertex of $P$ in $\{ u^{A^2}_{v',t} \colon t \in [k] \}$.

     We aim to invoke Lemma~\ref{lem:low_level_shifting}.
    Let $z'$ be a vector with $z'_{t} = 
 |(N_{T_{A^2}}(u^{A^2}_{v',t}) \cap V(\tilde{F})) \setminus V(P)|$ for $t \in [k]$.
 Further, let $z = u_{\ell} \cdot \idv_{\ell}$.
 We want to show that $z \leq z'$. Observe that this is equivalent to $z_{\ell} \leq z'_{\ell}$.
 If $t(v') \neq \ell$, we have $z'_{\ell} = v_{\ell} + u_{\ell}$, hence $z_{\ell} = u_{\ell} \leq z'_{\ell}$.
 If otherwise $t(v') = \ell$, we have $z'_{\ell} = v_{\ell} + u_{\ell} -1 \geq u_{\ell}$ by the precondition $v_{\ell} \geq 1$.
 Hence $z_{\ell} = u_{\ell} \leq z'_{\ell}$.

Finally, we apply Lemma~\ref{lem:low_level_shifting} to $\alpha, (\tilde{F}, \tilde{g}), A^2, v', u', T_{A^2}, P$, and $z$. Since, $(u' + z, v' - z) = (u, v)$, this yields a spanning tree satisfying the conditions of this lemma.
\end{proof}

\begin{longlemma}\label{lem:cw_equiv_rrule_12}
    Let $\ord=(O,b)$ be an ordered pattern over $m$ vectors. For an index $i \in [m]$ and a bad coordinate $\ell\in Z_{\ord}(i)$, let $(O^1, b^1) = \rrule_1(O, b, i, \ell)$ and $(O^2, b^2)=\rrule_2(O,b,i,\ell)$. Then it holds that $\{ (O, b) \} \preceq \{ (O^1, b^1), (O^2, b^2) \}$.
\end{longlemma}
\begin{proof}
    Let $j=b(\ell)$, $A = p(O, b)$, $A^1 = p(O^1, b^1)$, $A^2 = p(O^2, b^2)$,
    $u = O_{i}$, $v = O_{j}$,
    $u'_1 = (O^1)_{i}$, $v'_1 = (O^1)_{j}$, and
    $u'_2 = (O^2)_{i}$, $v'_2 = (O^2)_{j}$.
    Note that
    $A^1 = (A \setminus \langle u, v \rangle) \cup \langle u'_1, v'_1 \rangle$,
    $A^2 = (A \setminus \langle u, v \rangle) \cup \langle u'_2, v'_2 \rangle$,
    and that $\{ (O, b) \} \preceq \{ (O^1, b^1), (O^2, b^2) \}$ is equivalent to $\{ A \} \preceq \{ A^1, A^2 \}$, which we will now show.

    Let $\alpha\subseteq [k]^2$ be an edge mapping, and let $(\tilde{F}, \tilde{g})$ be a fixed forest that is $\alpha$-compatible with the fixed forest of $A$, $(F_{A}, f_{A})$.
    We need to show that there is an $\alpha$-spanning tree for $(\tilde{F}, \tilde{g})$ and $(F_{A^1}, f_{A^1})$, or there is an $\alpha$-spanning tree for $(\tilde{F}, \tilde{g})$ and $(F_{A^2}, f_{A^2})$.
    
    By the $\alpha$-compatibility of $(\tilde{F}, \tilde{g})$ and $(F_{A}, f_{A})$, there is an $\alpha$-spanning tree $T_{A}$ for $(\tilde{F}, \tilde{g})$ and $(F_{A}, f_{A})$.
    Let $P$ be a shortest path from $\{ u^{A}_{u,t} \colon t \in [k] \}$ to $\{ u^{A}_{v,t} \colon t \in [k] \}$ in $T_{A}$.
    Further, let $t(u)$ be the index of the first vertex of $P$ in $\{ u^{A}_{u,t} \colon t \in [k] \}$, and let $z'$ be a vector with $z'_t = 
 |(N_{T_{A}}(u^{A}_{u,t}) \cap V(\tilde{F})) \setminus V(P)|$ for $t \in [k]$.

    First, suppose $t(u) = \ell$.
    Let $z = u - u_\ell \cdot \idv_\ell$.
    For $t \neq \ell$, we have $z_t = z'_t$, and for $t = \ell$, we have $z_t = 0$. Hence $z \leq z'$.
    We apply Lemma~\ref{lem:low_level_shifting} to $\alpha, (\tilde{F}, \tilde{g}), A^1, u, v, T_A, P$, and $z$. Since, $(u - z, v + z) = (u'_1, v'_1)$, this yields an $\alpha$-spanning tree for $(\tilde{F}, \tilde{g})$ and $(F_{A^1}, f_{A^1})$.
    
    Conversely, suppose $t(u) \neq \ell$.
    Let $z = u_\ell \cdot \idv_\ell$.
    For $t \neq \ell$, we have $z_t = 0$, and for $t = \ell$, we have $z_{t} = z'_t$. Hence $z \leq z'$.
    We apply Lemma~\ref{lem:low_level_shifting} to $\alpha, (\tilde{F}, \tilde{g}), A^2, u, v, T_A, P$, and $z$. Since, $(u - z, v + z) = (u'_2, v'_2)$, this yields an $\alpha$-spanning tree for $(\tilde{F}, \tilde{g})$ and $(F_{A^2}, f_{A^2})$.
\end{proof} 
Combining Lemmas~\ref{lem:cw_equiv_rrule_1}, \ref{lem:cw_equiv_rrule_2} and~\ref{lem:cw_equiv_rrule_12} immediately yields the following corollary:
\begin{corollary}
\label{cor:cw_equiv_rrule12}
     Let $\ord=(O,b)$ be an ordered pattern over $m$ vectors, index $p \in I_\ord$ and coordinate $\ell\in Z_{\ord}(\ell)$ and 
    let $(O^1, b^1) = \rrule_1(O, b, p, \ell)$, $(O^2, b^2)=\rrule_2(O,b,p,\ell)$. Then it holds that $\{ (O, b) \} \simeq \{ (O^1, b^1), (O^2, b^2) \}$.
\end{corollary}

\para{Subroutines.}
Based on the functions $\pi_1$ and $\pi_2$, we formulate a subroutine that transforms a single pattern into an equivalent family of nice patterns. 

We now define the first subroutine, $\operatorname{ShiftVector}$, which takes as input an ordered pattern $(O, b)$, and a bad index $i\in I_{O,b}$, and outputs a set of ordered patterns, that result from an exhaustive application of the operations $\rrule_1$ and $\rrule_2$ to all bad coordinates of $i$.

Intuitively, we obtain the $j$th ordered pattern in the output from $(O, b)$ as follows. Let $z_{j}$ be the $j$th smallest bad coordinate of index $i$. We shift from index $i$ everything except one unit, i.e., the vector at index $i$ becomes $\idv_{z_{j}}$. The units from the first $j-1$ bad coordinates are shifted to the corresponding big indices, and all other units (i.e., units from coordinates that were not bad and all except one unit from $z_{j}$) are shifted to $b(z_{j})$. If $i$ was a big vector for some coordinates, $b(z_{j})$ becomes the big vector for those coordinates instead. Finally, the last ordered pattern in the output is obtained by shifting all bad coordinates from $i$ to the respective big vectors.

\begin{procedure}
\label{def:cw_shift_vector}
    The operation \defi{$\operatorname{ShiftVector}$} takes as input a pair $(\ord, i)$, for an ordered pattern $\ord=(O, b)$ with $O= O_1,\dots, O_m$, and $i\in I_{\ord}$. 
    Let $Z_{\ord}(i) = \{z_1, \dots, z_{\ell}\}$.
    The operation outputs a set of pairs $\{(\ord^1, p_1),\dots,(\ord^{\ell+1}, p_{\ell+1})\}$, where $p_{j}=b(z_{j})$ for $j\in [\ell]$, $p_{\ell+1}=\perp$, and $\ord^j = (O^j, b^j)$ for $j\in[\ell+1]$ are ordered vectors defined as follows:
    Let $v=O_{i}$.
    
    For $j \in [\ell]$ and $d\in[m]$ we define the vector $w = (O^j)_d$ from $w'=O_d$ as follows. We distinguish four different cases:
    \begin{itemize}
        \item If $d = b(z_{j})$, we define 
    \[
    w_z = \begin{cases}
        w'_z + v_z & \colon z \in \{z_1, \dots, z_{j-1}\} \land b(z_{j}) = b(z),\\
        w'_z & \colon z \in \{z_1, \dots, z_{j-1}\} \land b(z_{j}) \neq b(z),\\
        w'_z + v_z - 1 & \colon z = z_{j},\\
        w'_z + v_z & \colon \text{otherwise.}
    \end{cases}
    \]
    \item If $d=b(z)$ and $d\neq b(z_{j})$ for some $z\in\{z_1,\dots,z_{j-1}\}$, we define $w_z=w'_z+v_z$ and $w_{z'}=w'_{z'}$ for $z'\neq z$. 
    \item If $d = i$: we define $w = \idv_{z_{j}}$.

         \item If $d=b(z)$ and $d\neq b(z_{j})$ for some $z\in\{z_1,\dots,z_{j-1}\}$, we define $w_z=w'_z+v_z$ and $w_{z'}=w'_{z'}$ for $z'\neq z$. 
    \item Otherwise, we define $w=w'$.
\end{itemize}
    We define $b^j$ as the mapping such that $b^j(z) = b(z_{j})$ for each coordinate $z$ with $b(z) = i$, and $b^j(z) = b(z)$ for all other coordinates $z$.
    
    Finally, we define $b^{\ell+1} = b$. For $d\in[m]$, we define the vector $w = (O^{\ell+1})_d$ from $w' = O_d$ as follows:
    \[
    w_z =
    \begin{cases}
        0 &\colon d = i \land z \in Z_{O,b}(i),\\
        w'_z + v_z & \colon d = b(z) \land z \in Z_{O,b}(i),\\
        w'_z &\colon \text{otherwise.}
    \end{cases}
    \]
    \xqed
\end{procedure}

By definition, \cref{def:cw_shift_vector} outputs ordered patterns where the index $i$ is no longer bad, and the only (possibly) newly created bad index is given by the output pair. The following observation follows:

\begin{observation}
    Let $(\ord', p)$ be some pair in the output of $\operatorname{ShiftVector}(\ord, i)$ with $\ord'= (O', b')$. Then the following points hold:
    \begin{itemize}
        \item either $I_{\ord'} = I_{\ord}\setminus \{i\}$, or $I_{\ord'} = (I_{\ord}\setminus \{i\}) \cup p$, and
        \item $b^j([k])\subseteq b([k])$.
    \end{itemize}
    In particular, it holds for the unique pair $(\ord, p)$ with $p=\perp$ that $I_{\ord'} = I_{\ord}\setminus \{i\}$.
\end{observation}

\begin{lemma}\label{lem:cw_shift_vector_equiv}
    Given an ordered pattern $\ord = (O, b)$ and an index $i\in I_{\ord}$. Let $P$ be the set resulting from $\operatorname{ShiftVectorL}(\ord, i)$ and let $S=\{\ord\colon (\ord, p)\in P\}$. Then $S$ is equivalent to $\ord$.
\end{lemma}

\begin{proof}
    Let $Z_\ord(i)=\{z_1,\dots, z_t\}$, and let $\ord^1, \dots, \ord^{t+1}$ be the set of ordered patterns as defined in \cref{def:cw_shift_vector}.
    We define the sequence of ordered patterns $(Y^0, y^0),\dots,(Y^{\ell}, y^{\ell})$, where  $(Y^0, y^0) = (O,b)$, and $(Y^j, y^j) = \rrule_2(Y^{j-1}, y^{j-1}, i, z_j)$ for all $j\in[\ell]$. Then it holds for each $j\in[\ell]$ that $\ord^j = \rrule_1(Y^{j-1}, y^{j-1}, i, z_j)$. We show by induction over $j$ that $\{\ord^1, \dots, \ord^j\}\cup \{(Y^j, y^j)\}$ is equivalent to $\{(O, b)\}$. The lemma then follows from the observation that $\ord^{\ell+1} = (Y^{\ell}, y^{\ell})$.

    For $j=0$ the claim holds trivially since $(Y^0, y^0) = (O, b)$.
    Assume now that the claim holds for $j-1$, i.e., that $\{\ord^1, \dots, \ord^{j-1}\}\cup \{(Y^{j-1}, y^{j-1})\}$ is equivalent to $\{(O, b)\}$. Since it holds that $\ord^{j} = \rrule_1(Y^{j-1}, y^{j-1}, i, z_j)$ and $(Y^j, y^j) = \rrule_2(Y^{j-1}, y^{j-1}, i, z_j)$, it holds again by \cref{cor:cw_equiv_rrule12} that $\{\ord^1, \dots, \ord^{j}\}\cup \{(Y^j, y^j)\}$ is equivalent to $\{(O, b)\}$.
\end{proof}

Based on this operation, we define the recursive procedure $\operatorname{ReduceVector}$ that we use as a subroutine in our final process that turns any given pattern into an equivalent family of nice patterns. Informally, the $\operatorname{ReduceVector}$ operation selects an index $i$ and outputs a family of patterns where $i$ is no longer a bad index, and no other indices become bad.

\begin{procedure}
    \label{def:cw_reduce_vector}
    The recursive operation \defi{$\operatorname{ReduceVector}$} takes as input an ordered pattern $\ord$ over $m$ vectors and an index $i\in[m]$, and outputs a family of ordered patterns as follows:

    If $O_i$ is not a bad vector with respect to $b$, then we output the family $\{\ord\}$.
    Otherwise, let $\{(\ord^1, p_1),\dots,(\ord^{\ell+1}, p_{\ell+1})\}= \operatorname{ShiftVector}(\ord,i)$, and let $(\ord^{\ell+1}, p_{\ell+1})$ be the unique pair with $p_{\ell+1}=\perp$.
    Then for $j\in[\ell]$, we define $C_j = \operatorname{ReduceVector}(\ord^j,p_j)$.
    We define the output of the operation as $C = \{\ord^{\ell+1}\} \cup \bigcup_{j\in[\ell]} C_j$.
\end{procedure}

Now we prove some properties of this operation that we need later to prove the correctness of our procedures. We start with the following definition:

\begin{definition}
    Let $(O,b)$ be an ordered pattern over $m$ vectors and $i\in[m]$. We define the tree \defi{$\mathcal{T}_{O, b, i}$} as the recursion tree of the operation $\operatorname{ReduceVector}(O, b, i)$ (we omit the ordered pattern $O,b$ when clear from context). We also label each node of $\mathcal{T}_{O,b,i}$ corresponding to a call $(O', b', i')$ with the index $i'$.
\end{definition}

\begin{observation}\label{obs:cw_reduce_vector}
    For a recursion tree $\mathcal{T}_{O,b,i}$ and a node $\operatorname{ReduceVector}(O',b',i')$ corresponding to a child of the root of $\mathcal{T}_{O,b,i}$, all the following points hold:
     
    \begin{enumerate}
        \item $(O')_i = \idv_{d}$ for some $d\in Z_{O,b}(i)$, and hence, $i\notin I_{O',b'}$,\label{it:obsredvec_unit}
        \item $(b')^{-1}(i) = \emptyset$,\label{it:obsredvec_no_big}
        \item $i' = b(d)$ for some $d \in Z_{O,b}(i)$, and \label{it:obsredvec_shift_to_big}
        \item $b'([k]) \subseteq b([k])$.\label{it:obsredvec_big_increases}
    \end{enumerate}
\end{observation}

\begin{lemma}\label{lem:cw_reduce_vector_distinct_labels}
    There cannot exist two distinct nodes on the same root-leaf path of the recursion tree $\mathcal{T}_{O,b,i}$ sharing the same label $j\in[m]$.
\end{lemma}

\begin{proof}
    By induction over the depth of the nodes in the recursion tree, and by items (2) and (4) of \cref{obs:cw_reduce_vector}, for a node corresponding to a call $\operatorname{ReduceVector}(O',b',i')$, and for a descendant node corresponding to a call $\operatorname{ReduceVector}(O'',b'',i'')$, it holds that $(b'')^{-1}(i')= \emptyset$. The lemma follows then by item (3) of \cref{obs:cw_reduce_vector}.
\end{proof}

\begin{lemma}\label{lem:cw_reduce_vector_size}
    Let $\mathcal{T}_{O,b,i}$ be the recursion tree of a call $\operatorname{ReduceVector}(O, b, i)$.
    Then $\mathcal{T}_{O,b,i}$ has at most $k^k$ nodes.
    In particular, the procedure runs in time $k^{\oh(k)}$,
    and it holds that $|\operatorname{ReduceVector}(O, b, i)|\leq k^{\bigoh(k)}$.
\end{lemma}

\begin{proof}
    A call of $\operatorname{ReduceVector}(O, b, i)$ makes at most $|Z_{O,b}(i)|\leq k$ recursive calls, and hence, each node of $\mathcal{T}_{O,b,i}$ has at most $k$ children. Moreover, it holds by \cref{lem:cw_reduce_vector_distinct_labels} that the labels assigned to the nodes of each root-leaf path of the recursion tree $\mathcal{T}_{O,b,i}$ are pairwise distinct. Hence, $\mathcal{T}_{O,b,i}$ has height at most $k$. It follows that $\mathcal{T}_{O,b,i}$ has at most $k^{k}$ nodes. The bound on the running time follows then, since each node of the recursion tree, aside from the recursive calls, only makes a call of $\operatorname{ShiftVector}$, which runs in polynomial time. 
    
    Now let $T_h$ be the maximum size of a family returned by a call $\operatorname{ReduceVector}(O, b, i)$ with a recursion tree of height $h$. Then it holds that $T_0 = 1$, and for $h\geq 1$, it holds that $T_h\leq k\cdot T_{h-1} + 1\leq (k+1)T_{h-1}$. This implies that $T_k\leq (k+1)^k\leq k^{\bigoh(k)}$.
\end{proof}

\begin{lemma}\label{lem:cw_reduce_vector_equiv}
    Let $(O,b)$ be an ordered pattern over $m$ vectors and let $i\in[m]$. Then the operation $\operatorname{ReduceVector}(O, b, i)$ produces a family of ordered patterns equivalent to $\{(O,b)\}$, such that $I_{O',b'}\subseteq I_{O,b}\setminus\{i\}$ and $b'([k])\subseteq b([k])$ for each pattern $(O',b')$ in the resulting family.
\end{lemma}

\begin{proof}
    Let $(O^0,b^0,i_0),\dots,(O^r,b^r,i_r)$ be a sequence of calls of the $\operatorname{ReduceVector}$ operation corresponding to some root-leaf path of the recursion tree $\mathcal{T}_{O,b,i}$. In particular, $i_0,\dots,i_r$ are the labels assigned to the nodes of this path. Intuitively, in each iteration the coordinates of $i_j$ are shifted to either their big indices or to $i_{j+1}$, so $i_{j+1}$ is the only index that can turn bad in $O^{j+1}$ with respect to $b^{j+1}$. Formally, we have that $I_{O^{j+1},b^{j+1}}\subseteq (I_{O^j,b^j}\setminus\{i_j\})\cup \{i_{j+1}\}$.
    It follows by induction, that $I_{O^r,b^r}\subseteq I_{O,b}\setminus\{i\}$ holds for the last pattern $(O^r,b^r)$ in the sequence.

    Moreover, by \cref{obs:cw_reduce_vector}, $b^{j+1}([k])\subseteq b^j([k])$, and hence, $b^{r}([k])\subseteq b([k])$ holds also by induction. 
    The equivalence of the families also holds by induction using \cref{lem:cw_shift_vector_equiv}, since we only replace a pattern with the family resulting from an application of $\operatorname{ShiftVector}$ on this pattern.
\end{proof}

 Finally, we define the $\operatorname{ReduceToNice}$ operation that, given a pattern $A$, outputs a family of nice patterns equivalent to $A$. The first step of this operation is a subroutine called $\operatorname{ProcessBig}$ that takes care of the big vectors, i.e., ensures that the big vectors have no bad coordinates.

\begin{procedure}\label{def:cw_process_big}
    We define the subroutine \defi{$\operatorname{ProcessBig}$} that takes an ordered pattern $(O,b)$ as input, and outputs a family of patterns as follows:
    Let $H_0 = \{(O,b)\}$. We compute $H_j$ from $H_{j-1}$ for $j\in[k]$ as follows:
    For each pattern $(O', b')\in H_{j-1}$, we define the set $C_{O',b'}$, where $C_{O',b'} = \{(O',b')\}$ if $I_{O',b'}\cap b'([k]) = \emptyset$, and $C_{O',b'} = \operatorname{ReduceVector}(O', b', i)$ for the smallest index $i\in I_{O',b'}\cap b'([k])$ otherwise. Then we define $H_j = \bigcup_{(O',b')\in H_{j-1}} C_{O',b'}$.
    
    We define $H_j = \bigcup_{\ord'\in H_{j-1}} C_{\ord'}$. Let $s$ be the smallest integer such that $H_j = H_{j+1}$. The operation outputs the family $H_s$.
\end{procedure}

\begin{lemma}\label{lem:cw_process_big}
    It holds for each ordered pattern $(O,b)$, that $\operatorname{ProcessBig}(O,b)$ runs in time $k^{\bigoh(k)}$ and outputs a family $H$ of ordered patterns equivalent to $(O,b)$, such that for each pattern $(O',b')\in H$, it holds that $I_{O',b'}\cap b'([k]) = \emptyset$, $b'([k])\subseteq b([k])$ and $O_j = O'_j$ for all $j\notin b([k])$. Moreover, the process finishes after at most $k$ iterations.
\end{lemma}

\begin{proof}
    The last claim follows from the fact, that for each ordered pattern $(O',b') \in H_j$ with $(O',b')\in \operatorname{ReduceVector}(O'',b'',i)$ for some ordered pattern $(O'',b'')\in H_{j-1}$, it holds that $i\in b''([k])\cap I_{O'',b''}$. It holds by \cref{lem:cw_reduce_vector_equiv} that $I_{O',b'} \subseteq I_{O'',b''}\setminus \{i\}$. Hence, it holds that $I_{O',b'}\cap b'([k])$ is strictly smaller than $I_{O'',b''}$. Therefore, it holds by item (3) of \cref{obs:cw_reduce_vector} that largest size of such a set over all patterns in $H_j$ decreases after each iteration. Hence, it must hold after $k$ iterations for all ordered patterns $(O',b')\in H_j$ that $I_{O',b'}\cap b'([k]) = \emptyset$.

    The claim $b'([k])\subseteq b([k])$ holds by \cref{obs:cw_reduce_vector}. The claim $O_j = O_j'$ for all $j\notin b([k])$ also holds by \cref{obs:cw_reduce_vector},
    since we only call $\operatorname{ReduceVector}$ on indices in $b([k])$, which only shifts indices to $b([k])$. It holds that $H_j$ is equivalent to $H_{j-1}$ by \cref{lem:cw_reduce_vector_equiv}, since we only replace a pattern with the family resulting from an application of $\operatorname{ReduceVector}$ on this pattern. It follows by induction over $j$ that $H$ is equivalent to $(O,b)$.

    We claim that $|H_j|\leq k^{\bigoh(k)}$ for all $j\geq 0$. The lemma then follows by \cref{lem:cw_reduce_vector_size}. Let $(O',b')\in H_j$, then it holds for each index $i\in I_{O',b'}$ that 
    $O'_i = O_i$. Namely, each time we shift to an index $i$ (possibly making it a bad index), we directly call $\operatorname{ReduceVector}$ on $i$, and hence, by \cref{lem:cw_reduce_vector_equiv}, $i$ cannot be a bad index in the patterns resulting form such a call. Hence, each pattern of $H_j$ can be uniquely identified by the indices of $b([k])$ that were ``reduced'' (i.e., indices $i$ such that $i\not\in I_{O',b'}$), and by their respective values. Let $S$ be the set of these indices that do not correspond to unit or zero vectors. Since all the indices of $S$ are not bad indices, the non-zero coordinates of their corresponding vectors must form a partition of a subset of $[k]$. Hence, we have at most $3^k\cdot k^k \cdot (k+1)^k = k^{\bigoh(k)}$ possible patterns in $H_j$, where the first term counts for each index $i\in b([k])$ whether $O'_i = O_i$, $O'_i$ is a unit or zero vector, or $i \notin I_{O',b'}$ and $O_i$ is not a unit or zero vector. The second term counts the number of partitions of the non-zero coordinates, and the third term counts the different unit and zero vectors for each index corresponding to a unit or zero vector $O'_i$.
\end{proof}

Now we are ready to describe the procedure $\operatorname{ReduceToNice}$, that takes as input a pattern and outputs a family of nice patterns equivalent to the input:

\begin{procedure}\label{def:cw_reduce_to_nice}
    The procedure \defi{$\operatorname{ReduceToNice}$} takes as input a pattern $A\in\mathcal{P}$, and outputs a family of patterns.
    The procedure first fixes an arbitrary ordered pattern $\ord_A\coloneqq(O,b)$ of $A$. It computes the family $S_0=\operatorname{ProcessBig}(O,b)$.
         
    Let $\{i_1,\dots,i_s\} \coloneqq I_{O,b}\setminus b([k])$.
    For $j\in[s]$, we define $ S_j$ as follows:
    Let $ S'_j = \bigcup_{(O',b')\in  S_{j-1}} \operatorname{ReduceVector}(O', b', i_j)$.
    Then the family $ S_j$ results from $ S'_j$ as follows: We iterate over all patterns $(O',b')\in  S'_j$, and we add $(O', b')$ to $ S_j$ if the following two conditions hold: first if $(O',b')$ contains at most one zero vector, and second if $ S_j$ does not already contain another pattern $(O'',b'')$ that shares the same underlying pattern with $(O',b')$, i.e., $p(O',b') = p(O'',b'')$. In other words, for each pattern $p\in\mathcal{P}$, we only add one ``representative'' ordered pattern of $p$ to $ S_j$ if such a pattern exists in $ S'_j$.
    The operation outputs the family
    $\big\{p(\ord)\colon \ord\in S_s\big\}$.
    \xqed
\end{procedure}

In the following we show that this subroutine runs in time $n^{\bigoh(k)}$, and correctly produces a family of nice patterns equivalent to the input pattern.

\begin{lemma}\label{lem:cw_reduce_to_nice}
    It holds for all $j\in[s]_0$ and $(O',b')\in  S'_j$ that $I_{O',b'} = \{i_{j+1},\dots, i_s\}$ and $b'([k]) \subseteq b([k])$. Further, it also holds that $O'_{\ell} = O_{\ell}$ for all $\ell\in\{i_{j+1},\dots,i_s\}$.
\end{lemma}
\begin{proof}
    We prove the claim by induction over $j\in[s]_0$.
    For $j=0$, the claim holds by \cref{lem:cw_process_big}, since $I_{O',b'} = I_{O,b}\setminus b([k])$ for all $(O',b')\in  S_0$. Moreover, $b'([k]) \subseteq b([k])$ holds by Lemma~\ref{lem:cw_reduce_vector_equiv}.
     
    Now assume the claim holds for $j-1$. Let $(O',b')$ be a pattern in $ S'_{j}$, and let $(O^*, b^*)$ be a pattern in $ S_{j-1}$ such that $(O',b')\in \operatorname{ReduceVector}(O^*, b^*, i_j)$.
    Then it holds by \cref{lem:cw_reduce_vector_equiv} that 
    $I_{O',b'}\subseteq I_{O^*,b^*}\setminus\{i_j\}$ and $b'([k])\subseteq b^*([k])$. Moreover, it holds by the induction hypothesis that $I_{O^*,b^*} = \{i_{j},\dots,i_s\}$, and $b^*([k])\subseteq b([k])$. It follows that $I_{O',b'} \subseteq \{i_{j+1},\dots,i_s\}$, and $b'([k])\subseteq b([k])$.

    Finally, the inclusion $\{i_{j+1},\dots,i_s\}\subseteq I_{O',b'}$ follows from the fact that $\operatorname{ShiftVector}$ only shifts to big indices, and that it holds by induction hypothesis that $b^*([k])\subseteq b([k])$ is distinct from all indices $i_1,\dots, i_s$. It also holds by the same observation, that $O'_{i_{\ell}} = O^*_{i_{\ell}} = O_{i_{\ell}}$ for all $\ell\in\{j+1,\dots,s\}$.
\end{proof}

\begin{corollary}\label{cor:cw_reduce_to_nice_output}
    For each $A\in\mathcal{P}$, the operation $\operatorname{ReduceToNice}(A)$ produces a family of nice patterns equivalent to $A$.
\end{corollary}
\begin{proof}
    It holds by \cref{lem:cw_reduce_to_nice} that $I_{O',b'} = \emptyset$ for all $(O',b')\in  S_s$. Since we remove all patterns containing at least two zero vectors, it holds by \cref{obs:cw_ord_nice} that the output family is nice.
    
    It holds by \cref{lem:cw_process_big} that $S_0$ is equivalent to $\{\ord_A\}$.
    We show that $S_j$ is equivalent to $S_{j-1}$ for all $j\in[s]$, and hence, it holds by induction over $j\in[s]_0$ that $ S_j$ is equivalent to $\{\ord_A\}$.
    The family $ S'_i$ is equivalent to $ S_{i-1}$ for all $i\in[s]$ by \cref{lem:cw_reduce_vector_equiv}. The family $ S_i$ is then equivalent to $ S'_i$ by the fact that we only remove patterns that are either equivalent to other patterns in the family, or contain at least two zero vectors. The latter patterns are not compatible with any partial solutions, and hence, are equivalent to the empty family, and they can be removed safely.
\end{proof}

\begin{lemma}\label{lem:cw_reduce_to_nice_size}
    It holds for each $j\in[s]_0$ that $| S_j| \leq |\mathcal{P^*}| \leq n^{\bigoh(k)}$.
\end{lemma}
\begin{proof}
	    For $j=0$ this holds by \cref{lem:cw_process_big}. Now assume $j>0$.
    By \cref{lem:cw_reduce_to_nice}, all patterns in $ S_j$ have the same set of bad indices and identical vectors at bad indices, i.e., it holds for all $(O',b')\in  S_j$, $I_{O',b'} = \{i_{j+1}, \dots, i_s \}$, and for $i\in I_{O',b'}$, $O'_i = O_i$.
    Let $ S_j^*$ be the family resulting from $ S_j$ by removing all vectors $\{O_i\colon i\in \{i_{j+1}, \dots, i_s \}\}$. Then $ S_j^*$ contains nice patterns only, each containing at most $n$ vectors. Moreover, for any two patterns $(O',b'),(O'',b'')\in  S_j^*$, it holds that $p(O',b') \neq p(O'',b'')$, since otherwise, the original ordered patterns in $ S_j$ would share the same underlying pattern, a contradiction to our choice of of $ S_j$. Hence, the size of $ S_j^*$ (and hence, of $ S_j$) is bounded by the number of nice patterns with at most $n$ vectors (i.e., $|\mathcal{P^*}|$), which is at most $n^{\bigoh(k)}$, by \cref{obs:sizenice}.
\end{proof}

\begin{corollary}\label{cor:cw_reduce_to_nice_time}
    The operation $\operatorname{ReduceToNice}(A)$ runs in time $n^{\bigoh(k)}$ for each pattern $A\in\mathcal{P}$.
\end{corollary}
\begin{proof}
    It holds by \cref{lem:cw_process_big} that $\operatorname{ProcessBig}(O, b)$ runs in time $k^{\bigoh(k)}$.
    For each $j\in [s]$, we apply $\operatorname{ReduceVector}$ on each element of the family $ S_{j-1}$. This runs in time $| S_{j-1}|\cdot k^{\bigoh(k)}\poly(n) \leq n^{\bigoh(k)}$ by Lemmas~\ref{lem:cw_reduce_vector_size} and \ref{lem:cw_reduce_to_nice_size}, since each call of the $\operatorname{ShiftVector}$ runs in polynomial time.
    Moreover, it holds by these two lemmas that $ S'_j$ has size $n^{\bigoh(k)}$. Hence, $ S_j$ can be computed from $ S'_j$ in time $(n^{\bigoh(k)})^2\poly(n) = n^{\bigoh(k)}$. 
    Since $s\leq n$, the total running time is $n^{\bigoh(k)}$.
\end{proof}

\subsection{The Algorithm}

Before defining the algorithm, we describe an auxiliary operator that corresponds to relabeling vertices in the NLC-expression, i.e., it describes how a pattern corresponding to a partial solution changes after a relabel operation.
 \begin{definition}
    For a relabeling function $\beta : [k] \to [k]$ and      for $v\in \mathbb{N}^{k}$, we define $\rho_{\beta}(v) = v'$, where $v'_i = \sum\limits_{j \in \beta^{-1}(i)} v_j$.
    Given a pattern $A$, we define
    $\rho_{\beta}(A) = \langle \rho_{\beta}(v)\colon v\in A\rangle$.
\end{definition}

Moreover, we define the notion of ``unit components'' and ``big components'' that will be used in the algorithm:

\begin{definition}
    Given a nice pattern $A$, we call a connected component of the forest $F_A$ a \defi{unit component}, if it corresponds to a unit vector of $A$.
Similarly, a component of $F_A$ is a \defi{zero component}, if it corresponds to a zero vector. We call a connected component of $F_A$ a \defi{big component}, if it is neither a unit nor a zero component.
\end{definition}

Note that while this definition closely relates to the definitions of unit and big indices in an ordered pattern, the main difference is that here we forbid big components to correspond to unit vectors, where allowing such was necessary in the definitions of big indices.

Now, we describe our algorithm. At a high level, at each join node of the NLC-decomposition, we iterate over all pairs of patterns representing partial solutions at the two children of the current node. For each such pair of patterns, we iterate over all possible ways of combining them by adding edges between the partial solutions corresponding to these two patterns. Finally, we execute the relabeling procedure on each resulting pattern, and represent the outcome with its nice equivalent using the $\operatorname{ReduceToNice}$ operation.

The main difficulty lies in guessing the edges between a pair of partial solutions at the children nodes. By~\cref{obs:cw_label_shifting}, without loss of generality we can shift the labels of one partial solution by $k$, and subsequently unify the two partial solutions. Thus we obtain a graph with labels in $[2k]$, which has at most $2k$ big components, at most $n$ unit components and at most two zero components. 

We will add edges between these components in three stages: first, we add the edges whose both endpoints lie in big components. We follow by adding edges whose both endpoints lie in unit components. Finally, we add edges that connect a unit component to a big component. The first two types of edges can be added in a straightforward way. Namely, we can add at most $2k-1$ edges between big components (adding more edges would yield a cycle), and at most one edge between two unit components (adding an edge between two unit components disconnects them from the rest of the graph, so the only case in which this leads to a solution is when the current forest consists of only two unit components). Therefore, we can simply iterate over all such choices.

However, the third type of edges, those between unit and big components, has to be treated with more care. The naive approach of iterating over all such possibilities leads to a $k^n$ algorithm at best. Instead, we start by observing that after adding these edges, some unit components will remain unchanged, while the rest be ``absorbed'' by big components, without adding additional requirement to the big components they attach to. This holds, since by definition, they only have unit requirements, which are satisfied by the added edge. 
In other words, the total requirement of a big component decreases by the number of unit components attached to each label, regardless of the types of the unit components attached. Therefore, we can guess two things independently:
Firstly, for each of the $2k$ types of unit components, we guess how many components of each type will be absorbed. Secondly, for each label, by definition, there exists at most one vertex of this label in some big component having non-zero requirement. We then guess the number of unit vectors that will attach to each such vertex. To verify the validity of our guesses, we need to find a set of edges between unit and big components the realizes them. We achieve this by a call of a \textsc{Bipartite $b$-Matching} oracle. Although different valid solutions to the \textsc{Bipartite $b$-Matching} problem lead to different resulting forests, it holds by the argument above, that all these  forests admit the same pattern. Therefore, a single solution of the \textsc{Bipartite $b$-Matching} instances suffices to construct the resulting pattern.  

 \begin{algorithm}
\label{def:cw_algo}
    For each $x\in\nodes$, we define the families $T_x \subseteq \mathcal{P}^*$ recursively over $\syntaxtree$ as follows. For an initial node $\bullet_i$, let $v$ be the only vertex of $\bullet_i$; we define $T_x = \{\langle t\cdot \idv_i \rangle:\: t\in D(v) \}$.
    For a join node $x\in\nodes$ with $G_x = G_{x_1} \oplus_{\alpha}^{\beta} G_{x_2}$, we define $P_x = \{(A_1,A_2)\colon A_1\in T_{x_1}, A_2\in T_{x_2}\}$. 
    The algorithm computes the table $T_x$ in three steps, where it first computes two intermediate tables $T^{(1)}_x$ and $T^{(2)}_x$ as follows:
    
    First, we define the table $T^{(1)}_x$.
    For each pair $(A_1, A_2)\in P_x$, let $F_1 := F_{A_1}$, and $F_2$ be the forest obtained from $F_{A_2}$ by relabeling each vertex of label $i$ to label $k+i$ for all $i\in[k]$, and let $\alpha' = \{(i, k+j)\colon (i,j)\in \alpha\}$. Let $\hat{G}$ be the disjoint union of $F_1$ 
    and $F_2$, and let
    $\hat{E} = E(F_1\oplus_{\alpha'} F_2) \setminus E(\hat{G})$,
    i.e.\ $\hat{E}$ is the set of edges added to $\hat{G}$ by $\alpha'$. Finally, let $f:[2k]\rightarrow \mathbb{N}_0$ be defined as $f(i) = f_{A_1}(i)$ for $i\in[k]$ and $f(i) = f_{A_2}(i-k)$ for $i\in\{k+1,\dots, 2k\}$.
        
    Let $B$ be the set of vertices in big components with degree requirement at least one, and $U$ be the set consisting of the unique vertex of requirement one in each unit component. For each subset $E'\subseteq \hat{E}$ containing at most $2k-1$ edges between vertices of $B$, and at most one edge between vertices of $U$ and no other edges, let $G'$ be the graph resulting from $\hat{G}$ by adding the edges of $E'$. If it holds that $G'$ is acyclic, and it holds for all vertices $v\in \hat{G}$ that $\deg_{E'}(v) \leq f(v)$, we do the following. 

     Define $f' = f - \deg_{E'}$. 
    Note that for each $i\in[2k]$, the set $B$ contains at most one vertex of label $i$. 
    Let $b,u\in \mathbb{N}_0^{2k}$ be the vectors defined as follows: We define $b_i = f'(u)$, for the unique vertex $u \in B$ labeled $i$ if such a vertex exists, or $b_i = 0$ otherwise. We define $u_i = |\{u\in U\colon \lab_{G'}(u) = i\}|$, i.e.\ $u_i$ is the number of unit components, whose single vertex of requirement one is labeled $i$. For each choice of vectors $b'\leq b$ and $u'\leq u$, where $\leq$ is a componentwise inequality, let $H$ be the bipartite graph defined by $H=(B\cup [2k], E_H)$, where $E_H = \{(v,i)\in B\times [2k]\colon (\lab_{G'}(v), i)\in\alpha'\}$. We define the capacity function $c:V(H)\to \mathbb{N}_0$ as $c(v) = b'(v)$ for $v\in B$ and $c(i) = u'(i)$ for $i\in [2k]$. The algorithms verifies whether $(H,c)$ is a YES-instance of the \textsc{Bipartite $b$-Matching} problem.
    
    If this is the case, let $g$ be a solution of the \textsc{Bipartite $b$-Matching} instance $(H,c)$. We construct the graph $G''$ from $G'$ as follows: we assign $g(v, i)$ unit vertices of label $i$ to the vertex $v$, in such a way that each unit vertex is assigned to at most one vertex of $B$. Let $M$ be the set of edges consisting of a single edge between each such unit vertex and its assigned vertex $v$. We define $G''$ as the graph resulting from $G'$ by adding $M$ to it, and we define the mapping $f'': V(G'')\to \mathbb{N}_0$ as $f''(v) = f'(v)-\deg_M(v)$. Next, relabel the vertices of $G''$ by changing the labels of vertices in $V(F_2)$ from $i$ to $i-k$, thus obtaining a graph $G''_{\text{rel}}$ with labels in $[k]$.

    If $\repr(G''_{\text{rel}}, f'')$ is a pattern with at most one zero vector, then we add this pattern to $T^{(1)}_x$.  
    Next, we define     
        $T^{(2)}_x = \{ \rho_{\beta}(A) \colon A \in T^{(1)}_x\}$.
    Finally, we set $T_x =\bigcup_{A\in T^{(2)}_x} \operatorname{ReduceToNice}(A)$. The algorithm outputs YES if $\langle \mathbf{0}\rangle\in T_r$ at the root node $r$, or outputs NO otherwise.
    \xqed
\end{algorithm}

Now we show that the computed table entries are equivalent to the records at each node:

\begin{lemma}
\label{lem:cw_alg_corr}
    It holds for each $x\in \nodes$ that $T_x$ is equivalent to the record $\mathcal{R}_x$, i.e., $T_x\simeq \mathcal{R}_x$.
\end{lemma}

\begin{proof}
    We prove this lemma by induction over $\syntaxtree$. 
    Let $x$ be a leaf node corresponding to an operation $\bullet_i$, i.e.\ $G_x$ consists of a single vertex $v$ labeled $i$.
    Hence, the set of all partial solutions at $x$ is $\big\{(G_x, \{v\mapsto t\})\colon t\in D(v)\big\}$. It follows that $\mathcal{R}_x=\{\langle t\cdot \idv_i\rangle\colon t\in D(v) \}=T_x$.

    For an inner node $x\in\nodes$ with children $x_1$ and $x_2$, it holds by the induction hypothesis that $T_{x_1}\simeq \mathcal{R}_{x_1}$ and $T_{x_2}\simeq\mathcal{R}_{x_2}$. Since $T_x$ results from $T^{(2)}_x$ by applying $\operatorname{ReduceToNice}$ on each of its patterns, and since the union operation preserves equivalence, it suffices by \cref{cor:cw_reduce_to_nice_output} to show that $T_x^{(2)}\simeq \mathcal{R}_x$. 
    Let $G_x'=G_{x_1}\oplus_\alpha G_{x_2}$ and let $\mathcal{R}'_x$ be the record of $(G_x', D_x)$. Let $\alpha^R=\{(i, j)\in [k]^2\colon (j, i)\in \alpha  \}$.

    \begin{longclaim}\label{clm:alg_equiv}
        We have $\mathcal{R}'_x\simeq T_x^{(1)}$.
    \end{longclaim}
    \begin{claimproof}[Proof of \cref{clm:alg_equiv}]
    Let $R\in \mathcal{R}'_x$ and let $(F, g)$ be a partial solution of $(G_x', D_x)$ such that $R=\repr (F, g)$. 
    Let $(\tilde{F}, \tilde{g})$ be a fixed forest that is $\tilde{\alpha}$-compatible with $R$ (i.e. with the canonical fixed forest of $R$). By~\cref{obs:cw_label_shifting}, without loss of generality we may assume that $\tilde{F}$ has labels in $\{2k+1,\dots, 3k\}$ and correspondingly $\tilde{\alpha}\subseteq \{2k+1,\dots,3k\}\times [k]$.
    By \cref{lem:cw_same_repr}, $(\tilde{F}, \tilde{g})$ is $\tilde{\alpha}$-compatible with $(F, g)$. Let $T$ be the $\tilde{\alpha}$-spanning forest for $(\tilde{F}, \tilde{g})$ and $(F, g)$.
    We aim to construct an entry in $T_x^{(1)}$ that $(\tilde{F}, \tilde{g})$ is $\tilde{\alpha}$-compatible with. 
        
   Let $F_1=F\cap G_{x_1}$ and $F_2=F\cap G_{x_2}$. we define the mappings $g_1:V_{x_1}\rightarrow \mathbb{N}$ and $g_2:V_{x_2}\rightarrow \mathbb{N}$  with $g_1=g + \deg_F-\deg_{F_1}$, $g_2= g + \deg_F-\deg_{F_2}$, where we apply the sum pointwise, and we restrict the domains to $V_{x_1}$ and $V_{x_2}$, respectively.
   Note that $(F_1, g_1)$ is a partial solution of $(G_{x_1}, D_{x_1})$ and $(F_2, g_2)$ is a partial solution of $(G_{x_2}, D_{x_2})$, so $\repr(F_1, g_1)\in \mathcal{R}_{x_1}$ and $\repr(F_2, g_2)\in \mathcal{R}_{x_2}$. 
   Intuitively, our goal is to identify an entry of $P_x$ that corresponds to $(F_1, g_1)$ and $(F_2, g_2)$, and to provide a valid selection of edges between the corresponding forests resulting in a compatible pattern.
   
   Let $H_1$ be the labeled forest containing vertices $V(F_2)\cup V(\tilde{F})$, and edges $\{uv\in E(T)\colon u\not\in V(F_1), \: v\not\in V(F_1) \}$ (i.e., the forest obtained by deleting $F_1$ from $T$).
   The graph $\hat{H} := H_1\oplus_{\alpha^R\cup\tilde{\alpha}} F_1$ has vertices $V(F_2)\cup V(\tilde{F})\cup V(F_1)$. Since the vertices of $F_2$ only use labels in $[k]$, the vertices of $\tilde{F}$ only use labels in $\{k+1, \dots, 2k\}$, and since $\alpha^R$ and $\tilde{\alpha}$ are disjoint, the edge set of $\hat{H}$ is obtained by taking the edges between $F_2$ and $F_1$ corresponding to $\alpha^R$ and the edges between $\tilde{F}$ and $F_1$ corresponding to $\tilde{\alpha}$. Thus $T$ is a spanning tree in $\hat{H}$.

   Therefore, $(H_1, (\deg_T-\deg_{F_1})\cup(\deg_T-\deg_{H_1}))$ is $(\alpha^R\cup\tilde{\alpha})$-compatible with $(F_1, g_1)$.
   Since $\repr(F_1, g_1)\in \mathcal{R}_{x_1}$ and $\mathcal{R}_{x_1}\simeq T_{x_1}$ (by induction hypothesis), there is an element $A_1\in T_{x_1}$ such that $(H_1,(\deg_T-\deg_{F_1})\cup(\deg_T-\deg_{H_1}))$ is $(\alpha^R\cup\tilde{\alpha})$-compatible with $(F_{A_1}, f_{A_1})$. Let $T_1$ be the corresponding $(\alpha^R\cup\tilde{\alpha})$-spanning tree.

   Let $H_2$ be the labeled forest defined by the vertex set $V(F_{A_1})\cup V(\tilde{F})$ and the edges $\{uv\in E(T_1)\colon u\not\in V(F_2), v\not\in V(F_2) \}$ (i.e., the forest obtained by deleting $F_2$ from $T_1$). Analogously, $(H_2, (\deg_{T_1}-\deg_{F_2})\cup (\deg_{T_1}-\deg_{H_2}))$ is $(\alpha\cup\tilde{\alpha})$-compatible with $(F_2, g_2)$,  
 and there is an element $A_2\in T_{x_2}$ such that $(H_2, (\deg_{T_1}-\deg_{F_2})\cup (\deg_{T_1}-\deg_{H_2}))$ is $(\alpha\cup\tilde{\alpha})$-compatible with $(F_{A_2}, f_{A_2})$. Let $T_2$ be the corresponding $(\alpha\cup \tilde{\alpha})$-spanning tree.

Now we obtained a pair $(A_1, A_2)\in P_x$, and it remains to describe the set of edges that need to be added.
   Let $E_{\text{add}}=\{uv\in E(T_2)\colon u\in V(F_{A_1}), v\in V(F_{A_2}) \}$.
       Note that the labels of endpoints of edges in $E_{\text{add}}$ belong to $\alpha\cup\tilde{\alpha}$. However, the vertices of $F_{A_1}$ and $F_{A_2}$ do not have labels in $\{2k+1, \dots, 3k\}$, so the endpoints of all edges in $E_{\text{add}}$ have labels in $\alpha$, i.e., $E_{\text{add}}\subseteq E(F_{A_1}\oplus_\alpha F_{A_2})\setminus (E(F_{A_1})\cup E(F_{A_2}))$. By construction, $T_2$ witnesses that 
   \begin{equation}\tag{*}\label{eq:cw_correctness}
   (\tilde{F}, \tilde{g}) \text{ is } \tilde{\alpha}\text{-compatible with } 
    (F_{A_1}\cup F_{A_2}\cup E_{\text{add}}, (f_{A_1}\cup f_{A_2})-\deg_{E_{\text{add}}}).
   \end{equation}

   Let $F_1'=F_{A_1}$ and $F_2'$ be the forest obtained from $F_{A_2}$ by relabelling each vertex of label $i$ to $k+i$. Since $F_{A_2}$ and $F_2'$ have the same underlying vertex set, we can regard $E_{\text{add}}$ as edges between $F_1'$ and $F_2'$. 
   Set $f_1=f_{A_1}$ and $f_2(i+k)=f_{A_2}(i)$ for $i\in [k]$.
   
   Let $\hat{G}$ be the disjoint union of $F_1'$ and $F_2'$ and $B$ the set of vertices with non-zero requirement of the union of all big components of $\hat{G}$ and $U$ the set of vertices of non-zero requirement of the union of all unit components.
   Since $E_{\text{add}}$ is a subgraph of a tree, it is acyclic. In particular, since $|B|\leq 2k$, $E_{\text{add}}$ contains at most $2k-1$ edges between vertices of $B$. Additionally, $E_{\text{add}}$ contains at most one edge between unit components of $\hat{G}$ (otherwise the connectivity constraints would not be satisfied). It remains to show that the edges in $E_{\text{add}}$ between big and unit components correspond to a solution of a \textsc{Bipartite $b$-Matching} instance.

   Let $E_{UB}\subseteq E_{\text{add}}$ be the set of edges between unit and big components of $\hat{G}$. Consider the \textsc{Bipartite $b$-Matching} instance with vertices $B\cup [2k]$, edges $
   E_M=\{(v,i)\in B\times [2k]\colon (\lab_{\hat{G}}(v),i)\in \alpha_{\text{rel}}\}$ and capacities $c(v)=\deg_{E_{UB}}(v)$ for $v\in B$ and $c(i)=|\{u\in U\colon \lab_{\hat{G}}(u)=i, \deg_{E_{UB}}(u)=1 \}|$. Note that the algorithm will, in some iteration, solve this instance of \textsc{Bipartite $b$-Matching}.
   
   Using $E_{UB}$, we will construct a solution to this instance as follows. Let $m(v,i)=|\{u\in U\colon  \lab_{\hat{G}}(u)=i, uv\in E_{UB}\}|$ for all $v\in B, i\in [2k]$. Now for each $v\in B$, the capacity constraint is satisfied as $\sum_{iv\in E_{M}} m(v,i)=|\{u\in U\colon uv\in E_{UB} \}|=\deg_{E_{UB}}(v)$. Similarly, for $i\in [2k]$, $\sum_{iv\in E_{M}} m(v,i)=|\{uv\in E_{UB}\colon \lab_{\hat{G}}(u)=i \}|=c(i)$ as each vertex in $U$ has degree zero or one in $E_{UB}$. 

   Let $m'$ be the solution to this \textsc{Bipartite $b$-Matching} instance obtained by the algorithm and $M$ the corresponding set of edges constructed by the algorithm. Consider the pattern added to $T_x^{(1)}$ in this step. 
   For each vertex $v\in B$, its original requirement, $(f_1\cup f_2)(v)$, will decrease by $|\{vv'\in E_{\text{add}}\colon v'\in B\}|+\deg_M(v)=|\{vv'\in E_{\text{add}}\colon v'\in B\}|+c(v)$. For each $i\in [2k]$, the pattern will have $c(i)$ zero vectors, and the remaining $|\{u\in U\colon \lab_{\hat{G}}(u)=i, \deg_{E_{UB}}=0\}|$ unit vectors will remain unchanged. Thus the pattern constructed in this iteration is equal to the pattern of $(\hat{G}\cup E_{\text{add}}, (f_1\cup f_2)-\deg_{E_{\text{add}}})$. By~\cref{eq:cw_correctness} and \cref{obs:cw_label_shifting}, $(\tilde{F}, \tilde{g})$ and $(\hat{G}\cup E_{\text{add}}, (f_1\cup f_2)-\deg_{E_{\text{add}}})$ are $\tilde{\alpha}'$-compatible, where $\tilde{\alpha}'=\tilde{\alpha}\cup \{(i,j+k)\colon (i,j)\in \tilde{\alpha}\}$. Therefore, after relabelling, we have obtained a pattern in $T_x^{(1)}$ that is $\tilde{\alpha}$-compatible with $(\tilde{F},\tilde{g})$.

    The proof in the other direction can be obtained by reversing the direction of implications.
    \end{claimproof}
   \begin{longclaim}\label{clm:alg}
       We have $\mathcal{R}_x\simeq T_x^{(2)}$.
   \end{longclaim}
    \begin{claimproof}[Proof of \cref{clm:alg}]
       Let $R\in \mathcal{R}_x$ be a pattern and let $(H, h)$ be a partial solution with $\repr(H, h)=R$. Let $\alpha\subseteq[k]^2$ and let $(F, g)$ be a fixed forest that is $\alpha$-compatible with $(H, h)$ and let $T$ be the corresponding $\alpha$-spanning tree. We aim to construct an element of $T_x^{(2)}$ such that $(F, g)$ is $\alpha$-compatible with its canonical fixed forest.
       
       Let $(H_{\varphi}, h)$ be the partial solution whose pattern belongs to $\mathcal{R}'_x$, such that $H$ is obtained from $H_{\varphi}$ by applying the relabeling mapping $\beta$ to $H_\varphi$.
       Let $\alpha'=\{(i,j):\: (i, \beta(j))\in \alpha \}$. Note that $F\oplus_{\alpha'} H_\varphi$ has the same edge set as $F\oplus_\alpha H$, so
       $(F, g)$ is $\alpha'$-compatible with $(H_\varphi, h)$ (with $\alpha'$-spanning tree $T$). Since $\repr(H_\varphi, h)\in \mathcal{R}'_x$, and since it holds by \cref{clm:alg_equiv} that $\mathcal{R}'_x\simeq T_x^{(1)}$, there must exist a pattern $A\in T_{x}^{(1)}$ such that $(F, g)$ is $\alpha'$-compatible with $(F_A, f_A)$. Let $T_1$ be the corresponding $\alpha'$-spanning tree. Let $B=\rho_\beta(A)$. 
       Note that $B\in T_{x}^{(2)}$. It remains to prove that $(F, g)$ is $\alpha$-compatible with $B$. Note that $F_B$ is the same labeled forest as $F_A$, since $\rho_\beta$ only affects the requirement function.
       Let us construct an $\alpha$-spanning tree $T_2$ in $F\oplus_\alpha F_A$ as follows. For each edge $uv\in T_1$, where $u\in V(F)$ and $v\in V(F_A)$, we do the following. Let $\lab_{F_A}(v)=i$ and let $C$ be the connected component of $v$ in $F_A$. We add to $T_2$ the edge between $u$ and the unique vertex in $C$ with label $\beta(i)$. Note that, since the edge $uv$ had labels $(\lab_F(u), i)\in \alpha'$, the newly constructed edge of $T_2$ has labels $(\lab_{F}(u), \beta(i))\in \alpha$.
       Finally, we add to $T_2$ all edges inside $F$ and $F_A$. It is easy to see that $T_2$ is a spanning tree in $F\oplus_\alpha F_A$. For each vertex $u\in V(F)$, $\deg_{T_2}(u)=\deg_{T_1}(u)$, and for each vertex $v\in V(F_A)$, $\deg_{T_2}(v)$ is equal to the sum of $\deg_{T_1}(v')$, where $v'$ is a vertex in the same connected component as $v$ and $\lab_{F_A}(v')\in \beta^{-1}(\lab_{F_A}(v))$. Thus $(F, g)$ is $\alpha$-compatible with $B$. 
       
       The other direction can be obtained by reversing the direction of implications.   
   \end{claimproof}
   Lastly, \cref{clm:alg} and \cref{cor:cw_reduce_to_nice_output} imply that $T_x\simeq \mathcal{R}_x$, as desired.
\end{proof} 
The following corollary tells us that the input is a YES-instance if and only if the table entry at the root contains the pattern $\langle \mathbf{0} \rangle$. Informally, the pattern $\langle\mathbf{0}\rangle$ corresponds to a partial solution with one connected component, where all degree constraints are satisfied, i.e., a solution. 
\begin{corollary}
\label{cor:cw_zero_acc}
    The instance $(G,D)$ is a YES-instance if and only if $T_r$ contains $\langle \mathbf{0}\rangle$.
\end{corollary}
\begin{proof}
    It is easy to see that $\langle \mathbf{0}\rangle$ is the only pattern that is $\emptyset$-compatible with the empty graph.
    
    Suppose $\langle\mathbf{0}\rangle\in T_r$. Since $T_r\simeq R_r$ and $\mathbf{0}$ is $\emptyset$-compatible with the empty graph, there is an element in $\mathcal{R}_r$ that is $\emptyset$-compatible with the empty graph, i.e., $\langle\mathbf{0}\rangle\in \mathcal{R}_r$. The partial solution that corresponds to $\langle \mathbf{0}\rangle$ is a forest in $G$ with one connected component that satisfies all degree requirements, i.e., it is a solution to our input instance.

    For the other direction, let $T$ be a spanning tree in $G$ that satisfies all degree requirements. Then $\repr(T, \mathbf{0})=\langle \mathbf{0}\rangle\in \mathcal{R}_r$. By the same argument as above, $\langle\mathbf{0}\rangle\in T_r$.
\end{proof} 
Finally, we bound the running time of our algorithm:
\begin{lemma}\label{lem:cw_time}
    Algorithm~\ref{def:cw_algo} runs in time $n^{\bigoh(k)}$.
\end{lemma}

\begin{proof}
For a join node $x\in \nodes$, let $x_1, x_2\in \nodes$ be its children. 
We claim that we can compute $T_x$ in time $n^{\bigoh(k)}$ given  $T_{x_1}$ and $T_{x_2}$.

Firstly, by construction and \cref{cor:cw_reduce_to_nice_output}, $T_{x_1}$ and $T_{x_2}$ are sets of nice patterns. 
By Observation~\ref{obs:sizenice}, this implies that $|T_{x_1}|, |T_{x_2}| \leq n^{\bigoh(k)}$.
Therefore, $|P_x| \leq n^{\bigoh(k)}$.

Secondly, we bound the complexity of computing $T^{(1)}_x$.
Since $|B|\leq 2k$ and there are at most $n$ unit components, the set $E'$ can be chosen in $n^2\cdot {4k^2 \choose 2k-1}=n^2\cdot k^{\bigoh(k)}$. Since $u$ and $b$ are $2k$-dimensional vectors, there are $n^{\bigoh(k)}$ ways to choose $u'$ and $b'$. The \textsc{Bipartite $b$-Matching} problem is polynomial time solvable. Thus computing $T^{(1)}_x$ takes $n^{\bigoh(k)}$ time.
    
Thirdly, since the relabeling operation does not increase the number of vectors, we have $|T^{(2)}_x| \leq |T^{(1)}_x|$.
Finally, to obtain $T_x$, we perform $|T^{(2)}_x|$ calls of $\operatorname{ReduceToNice}$, each of which takes time $n^{\bigoh(k)}$, by \cref{cor:cw_reduce_to_nice_time}.
Hence, in total, this step takes time $n^{\bigoh(k)}$.
The lemma then follows.
\end{proof} 
Now we are ready to prove the main theorem for clique-width.

\begin{proof}[Proof of Theorem~\ref{thm:cw-ub}]
    Given an instance $(G,D)$ of \setmstp together with an NLC-decomposition of width $k$, we run Algorithm~\ref{def:cw_algo}. By Corollary~\ref{cor:cw_zero_acc}, the input instance is a YES-instance if and only if the table entry at the root contains the pattern $\langle \mathbf{0}\rangle$. The total running time is $n^{\bigoh(k)}$, by Lemma~\ref{lem:cw_time}.
\end{proof}

\section{Upper Bound Techniques}\label{sec:ubtechniques}
\noindent
Except for our clique-width algorithm, all our upper bounds rely on the \cnc technique that counts the number of weighted solutions (modulo two) in the given running time. We make use of the isolation lemma to reduce the decision version of this problem to the counting (modulo two) version with high probability.

\para{Cut and Count.}
Let $(G=(V,E), \weightf, \dsets, \costbd)$ be a weighted \setmstp instance, with $n = |V|$ and $m = |E|$. 
Further, let $d(v) = \max \dsets(v)$ for each vertex $v\in V$, and let $\maxreq = \max_{v\in V} d(v)$ be the maximum requirement over all vertices.
Next, let $W = \max_{e\in E} \weightf(e)$, $\overline W = [(n-1)\cdot W]_0$, and $\overline M = [m]_0$.
We will later fix a value $Z\in \mathbb{N}$ that is bounded polynomially in $m$, and a second weight function $w':E\rightarrow [Z]_0$. 
We set $\overline{Z} = [(n-1)\cdot Z]_0$.
Similar to \cite{CyganNPPRW22}, we define the family of all solutions, relaxed solutions, and consistent cuts.

    A \defi{solution} is a set of edges $F\subseteq E$, such that $G[F]$ is a tree, and $\deg_F(v) \in \dsets(v)$ for each $v\in V$. Let $\mathcal{S}$ be the family of all solutions. We define families of weighted solutions, where for $\omega_1 \in \overline W$ and $\omega_2 \in \overline Z$ we define:
    $\mathcal{S}[\omega_1, \omega_2] = \big\{F\in \mathcal{S}\colon w(F) = \omega_1, w'(F) = \omega_2 \big\}$.
    Our goal is to compute the parity of $\mathcal{S}[\omega_1,\omega_2]$ for all values $\omega_1, \omega_2$. We achieve this by counting consistent cuts (defined below) using the \cnc technique.

    We define a \defi{relaxed solution} as a set of edges $F\subseteq E$ where $|F|=n-1$, and for each vertex $v$ it holds that $\deg_F(v)\in \dsets(v)$. 
    Let $\mathcal{R}$ be the family of all relaxed solutions.
    Similarly, we define families of weighted relaxed solutions: 
    $\mathcal{R}[\omega_1, \omega_2] = \big\{F\in \mathcal{R}\colon w(F) = \omega_1, w'(F) = \omega_2\}$.    

\begin{definition}
\label{def:usingconsistentcuts}
    Given a set of edges $F\in\mathcal{R}$, a \defi{consistent cut} of $G[F]$ is a pair $(F,c)$ where $c:V\rightarrow \{\LL,\RR\}$ is a mapping such that for each connected component $C\in \cc(G[F])$ it holds that $|c(C)| = 1$. For a relaxed solution $F\in \mathcal{R}$, let $\mathcal{C}(F)$ be the family of all consistent cuts of $F$. Let $\mathcal{C}$ be the family of all consistent cuts of all relaxed solutions. We define families of consistent cuts of weighted relaxed solutions $\mathcal{C}[\omega_1,\omega_2]$ as follows:
    \[
    \mathcal{C}[\omega_1, \omega_2] = \big\{(F,c) \in 
    \mathcal{R}[\omega_1,\omega_2] \times \{\LL,\RR\}^V\colon (F, c) \text{ is a consistent cut of } G[F]\big\}.
    \]
    Furthermore, let $\ncuts[\omega_1, \omega_2] = \big|\mathcal{C}[\omega_1, \omega_2]\big|$.
\end{definition}

We use a slightly modified version of the \cnc technique that was first used by Hegerfeld and Kratsch \cite{HegerfeldK23}. They prove the following result:

\begin{lemma}[\cite{HegerfeldK23}]
\label{tw-ub-lem:cnc}
    It holds that $C[\omega_1, \omega_2] \equiv_4 2 \cdot \big|\mathcal{S}[\omega_1, \omega_2]\big|$.
\end{lemma}

\begin{proof}
    It holds for each $F\in \mathcal{R}$ that if $G[F]$ is connected, then $G[F]$ is acyclic as well, since $|F| = n-1$. Hence, it holds that $F\in \mathcal{S}$ if and only if $F\in \mathcal{R}$ and $G[F]$ is connected.

    Let $F\in\mathcal{R}$. Since all vertices of each connected component must be assigned one of the two values $\LL$ and $\RR$, and since we can choose this value independently for each connected component, it holds that $|\mathcal{C}(F)| = 2^{|\cc(G[F])|}$. Therefore, it follows that
    \[
    C[\omega_1, \omega_2] = \sumstack{F\in \mathcal{R}[\omega_1, \omega_2]\\G[F] \text{ is connected}} 2 + \sumstack{F\in \mathcal{R}[\omega_1, \omega_2]\\G[F] \text{ is disconnected}} 2^{|\cc(G[F])|} \equiv_4 2 \cdot \big|\mathcal{S}[\omega_1,\omega_2]\big|,
    \]
    where the congruence holds, since the second summand is always congruent to zero modulo four, as $G[F]$ contains at least two connected components whenever $G[F]$ is disconnected.
\end{proof} 
In other words, it holds that $C[\omega_1, \omega_2]$ is always even, and that $\big|\mathcal{S}[\omega_1, \omega_2]\big|$ is odd, if and only if $C[\omega_1, \omega_2] \equiv_4 2$ (i.e., if $C[\omega_1, \omega_2]/2$ is odd). Hence, one can decide the parity of $|\mathcal{S}[\omega_1, \omega_2]|$ by computing $C[\omega_1, \omega_2]$ modulo 4. We use the isolation lemma to show that this suffices to solve the decision version with high probability.

\para{Isolation Lemma.} We say that a function $\omega:U\rightarrow \mathbb{Z}$ \defi{isolates} a set family $\mathcal{F}\subseteq 2^U$ if there exists a unique $S'\in\mathcal{F}$ with $\omega(S')=\min_{S\in \mathcal{F}}\omega(S)$.
Mulmuley et al.~\cite{MulmuleyVV87} showed that one can isolate any set family with high probability by choosing the weight function uniformly at random in a large enough space.
\begin{longlemma}[\cite{MulmuleyVV87}]\label{tw-ub-lem:iso-basic}
Let $\mathcal{F}\subseteq 2^U$ be a set family over a universe $U$ with $|\mathcal{F}|>0$, and let $N>|U|$ be an integer. For each $u\in U$, choose a weight $\omega(u)\in \{1,2,\dots,N\}$ uniformly and independently at random. Then it holds that $P[\omega \text{ isolates } \mathcal{F}]\geq 1-|U|/N$.
\end{longlemma}
 Based on their result, we prove the following lemma:

\begin{lemma}
\label{tw-ub-cor:isolation}
    Let $Z = 2 \cdot m$. We choose $w':E\rightarrow [Z]_0$, where we choose $w'(e)$ independently and uniformly at random for each edge $e\in E$. Let $\omega_1 \in \overline W$ be the smallest weight of a solution $F\in \mathcal{S}$. Then with probability at least $1/2$, there exists a value $\omega_2 \in \overline Z$ such that $\mathcal{S}[\omega_1, \omega_2]$ contains a single solution only.
\end{lemma}

\begin{proof}
    Let $U = E$, and $\mathcal{F} = \{F\in\mathcal{S}\colon w(F) = \omega_1\}$. By \Cref{tw-ub-lem:iso-basic}, it holds that $P[\omega \text{ isolates } \mathcal{F}]\geq 1-|E|/Z = 1/2$.
\end{proof} 
In our remaining algorithmic results, we show how to compute the values $C[\omega_1, \omega_2]$ (modulo $4$) in the given running time.

\section{Treewidth}\label{sec:tw}
\noindent
Since our lower bound for treewidth follows directly from the one for pathwidth, our aim in this section is to obtain the following algorithmic upper bound for \textsc{Set of Degrees MST}:

\twub*

Towards establishing Theorem~\ref{thm:tw-ub}, let $(\syntaxtree, \mathcal{B})$ be a nice tree decomposition of $G$ of width $\tw$ rooted at $r$.
Intuitively, our algorithm uses records to store the following information about each vertex $v$ in the bag: on which side of the consistent cut it lies (\textbf{L} or \textbf{R}; see Definition~\ref{def:usingconsistentcuts}) and the degree of $v$ that has been used up so far. Our aim is to compute the size of the weighted families $\mathcal{C}[\omega_1, \omega_2]$ modulo $4$ in order to apply Lemma~\ref{tw-ub-lem:cnc}.
At its core, we can formalize it as follows:

\begin{algorithm}  
\label{alg:tw}
For a vertex set $S$, we call a mapping $c:S\rightarrow \{\LL,\RR\}$ a \defi{coloring}.
For $x\in\nodes$, we define the family of indices at $x$ as
$\indices_x = ([r]_0)^{B_x}\times \{\LL,\RR\}^{B_x}$. 

Let $x\in\nodes$, and let $x'$ be the child of $x$ if $x$ has a single child, or $x_1$ and $x_2$ be the children of $x$ if it has two.
For $(a,\omega_1,\omega_2)\in\overline M \times \overline W \times \overline Z$,
we define the tables $T_x[a,\omega_1,\omega_2]\in (\zfr)^{\indices_x}$, where for $(f,c)\in \indices_x$ we define the value of $T_x[a, \omega_1,\omega_2][f,c]$ as follows:
    \begin{itemize}
         \item Leaf node introducing vertex $v$:          
        Let $\phi$ be the pair $(\emptyset_v, \emptyset_c)$, where $\emptyset_v$ is a vector over an empty set, and $\emptyset_c$ is the empty coloring.
        We define $T_x[0,0,0][\phi] = 1$, and $T_x[a,\omega_1, \omega_2][f,c] = 0$ for all other values $a,\omega_1,\omega_2,f,c$.

        \item Introduce vertex $v$: Let $f'$ and $c'$ be the restrictions of $f$ and $c$ to $B_{x'}$. We define
        \[
            T_x[a, \omega_1, \omega_2][f,c] = 
            \begin{cases}
                T_{x'}[a, \omega_1, \omega_2][f', c']& \text{if } f(v) = 0,\\
                0 & \text{otherwise}.                
            \end{cases}
        \]

        \item Introduce edge $e=\{u,v\}$: Let $f'$ be the function resulting from $f$ by decreasing $f(u)$ and $f(v)$ by one. Let $\omega_1' = \omega_1-w(e)$ and $\omega_2' = \omega_2-w'(e)$. We define
        \[
            T_x[a, \omega_1, \omega_2][f,c] = 
            \begin{cases}
                \!\begin{aligned}
                    T_{x'}&[a - 1, \omega_1', \omega_2'][f',c]\\
                    &+T_{x'}[a, \omega_1, \omega_2][f,c]
                \end{aligned}&  \text{if } c(u) = c(v) \land f(u),f(v)>0,\\
                T_{x'}[a, \omega_1, \omega_2][f,c] &  \text{otherwise}.
            \end{cases}
        \]
        
        \item Forget vertex $v$: Let
        \[
            T_x[a, \omega_1, \omega_2][f,c]= 
            \sumstack{o\in \dsets(v)\\X\in\{\LL,\RR\}}
            T_{x'}[a, \omega_1, \omega_2][f[v\mapsto o], c [v\mapsto X]].
        \]

        \item Join node: Let
        \[
            T_x[a, \omega_1, \omega_2][f,c] = 
            \sumstack{a^1+a^2 = a\\\omega_1^1 + \omega_1^2 = \omega_1\\ \omega_2^1 + \omega_2^2 = \omega_2}
            \sum\limits_{f^1+f^2=f}
            T_{x_1}[a^1,\omega_1^1, \omega_2^1][f^1,c] \cdot T_{x_2}[a^2,\omega_1^2, \omega_2^2][f^2,c].  \xqed
        \]
    \end{itemize}    
 \end{algorithm}

We formalize the intuition about the records via the following definition, and show that the two notions match in Lemma~\ref{lem:twdynprogtables}.

\begin{definition}
For $x\in \nodes$, $(a, \omega_1, \omega_2)\in \overline M \times \overline W \times \overline Z$, and $(f,c) \in \indices_x$, we define the sets $A_x[a,\omega_1, \omega_2][f,c]$ as follows:
\begin{align*}
    A_x[a, \omega_1, \omega_2][f,c] = \{(F, z)\in& 2^{E_x}\times \{ \LL,\RR\}^{V_x}\colon\\
    &|F| = a\ \land w(F) = \omega_1\ \land w'(F) = \omega_2\ \land \\
    &z|_{B_x} = c|_{B_x}\ \land\\
    &\deg_F(v) \in \dsets(v) \text{ for all } v \in V_x \setminus B_x\ \land\\
    &\deg_F(v) = f(v) \text{ for all } v \in B_x\ \land\\
    &z \text{ is a consistent coloring of } V_x\}.
\end{align*}    
\end{definition}

\begin{longlemma}  
\label{lem:twdynprogtables}
Let $x\in\nodes$. It holds for all $(a, \omega_1, \omega_2)\in \overline M \times \overline W \times \overline Z$, $(f,c)\in \indices_{x}$ that
\[
T_x[a, \omega_1, \omega_2][f,c] \equiv_4 \big|A_x[a, \omega_1, \omega_2][f,c]\big|.
\]
\end{longlemma}

\begin{proof}
    We prove the claim by induction over $\dtree$. The base case is a leaf node $x$, where the claim holds by definition, since $G_x$ is an empty graph. Let $x$ be an internal node, and assume that the claim holds for all children of $x$. We show that the claim holds for $x$ by distinguishing the different kinds of nodes $x$. Let $x'$ be the child of $x$ if $x$ has a single child, and $x_1$ and $x_2$ be the children of $x$ otherwise.
    \begin{itemize}
        \item Introduce vertex $v$: It holds that $\deg_x(v) = 0$, and hence, for $f(v) \neq 0$ it holds that $A_x[a,\omega_1, \omega_2][f,c]= \emptyset$. Otherwise, let $f'$ be the restriction of $f$ to $B_{x'}$, and $c'$ be the restriction of $c$ to $B_{x'}$. Then one can define a bijection between $A_x[a,\omega_1, \omega_2][f,c]$ and $A_{x'}[a,\omega_1, \omega_2][f',c']$ by assigning to a pair $(F, z)$ the pair $(F, z')$, where $z'$ is the restriction of $z$ to $B_x$. Hence, it holds that
        \[
        T_x[a,\omega_1, \omega_2][f,c] 
        \equiv_4 \big|A_x[a,\omega_1, \omega_2][f,c]\big|
        = \big|A_{x'}[a,\omega_1, \omega_2][f',c']\big|
        \equiv_4 T_{x'}[a,\omega_1, \omega_2][f',c'].
        \]
        
        \item Introduce edge $e=\{u,v\}$: If $c(u)$ and $c(v)$ are not consistent or $f(u) = 0$ or $f(v) = 0$, then it holds for all $(F, z) \in A_x[a, \omega_1, \omega_2][f,c]$ that $e\notin F$. It follows in this case that $A_x[a,\omega_1, \omega_2][f,c] = A_{x'}[\omega_1, \omega_2][f,c]$. Assuming this is not the case, then it holds for each pair $(F, z)\in A_x[a,\omega_1, \omega_2][f,c]$ that either $e\notin F$, and hence $(F, z)\in A_{x'}[a,\omega_1, \omega_2][f,c]$, or $e\in F$, and $(F\setminus \{e\}, z)\in A_{x'}[a-1,\omega_1 - w(e), \omega_2 - w'(e)][f', c]$, where $f'$ is obtained from $f$ by decreasing $f(u)$ and $f(v)$ by one.
        Hence, it holds in this case that
        \begin{align*}
        T_x[a,\omega_1, \omega_2][f,c]
        &\equiv_4 \big|A_x[a,\omega_1, \omega_2][f,c]\big| \\
        &= \big|A_{x'}[a,\omega_1, \omega_2][f,c]\big| \dot\cup \big|A_{x'}[a-1,\omega_1 - w(e), \omega_2 - w'(e)][f', c]\big|\\
        &\equiv_4 T_{x'}[a,\omega_1, \omega_2][f,c] + T_{x'}[a-1,\omega_1 - w(e), \omega_2 - w'(e)][f', c].
        \end{align*}

        \item Forget vertex $v$: The family $A_x[a,\omega_1, \omega_2][f,c]$ equals the disjoint union of all sets $A_{x'}[a,\omega_1, \omega_2][f[v\mapsto o], c[v\mapsto X]]$ for all $o\in \dsets(v)$ and $X\in \{\LL,\RR\}$. This holds, since for each pair $(F, z) \in A_x[a,\omega_1, \omega_2][f,c]$, it must hold that $\deg_F(v) \in \dsets(v)$. Hence, for $o = \deg_F(v)$, and $X = z(v)$, it holds that $(F, z)\in A_{x'}[a,\omega_1, \omega_2][f[v\mapsto o], c[v\mapsto X]]$. It follows that
        \begin{align*}
	    T_x[a,\omega_1, \omega_2][f,c]
        &\equiv_4 \big|A_x[a,\omega_1, \omega_2][f,c]\big|\\
        &= \sumstack{o\in \dsets(v)\\X\in \{\LL,\RR\}} \big|A_{x'}[a,\omega_1, \omega_2][f[v\mapsto o], c[v\mapsto X]]\big| \\
        &\equiv_4 \sumstack{o\in \dsets(v)\\X\in\{\LL,\RR\}}
        T_{x'}[a,\omega_1, \omega_2][f[v\mapsto o], c[v\mapsto X]].
        \end{align*}
        
        \item Join node: Let $a^1,a^2,\omega_1^1,\omega_1^2,\omega_2^1,\omega_2^2 \in\mathbb{N}$ such that $a^1+a^2=a$, $\omega_1^1 + \omega_1^2 = \omega_1$ and $\omega_2^1 + \omega_2^2 = \omega_2$.
        Let $f_1,f_2$ be two vectors, such that $f_1 + f_2 = f$. Then for each pair $(F_1, z_1) \in A_{x_1}[a^1,\omega_1^1, \omega_2^1][f_1, c]$ and $(F_2, z_2) \in A_{x_2}[a^2\omega_1^2, \omega_2^2][f_2, c]$ it holds that
        \[
        (F_1\dot\cup F_2, z_1 \cup z_2)\in A_x[a,\omega_1, \omega_2][f,c].
        \]
        Moreover, for each pair $(F, z)\in A_x[a,\omega_1, \omega_2][f,c]$ there exists values $a^1,a^2,w_1^1,w_1^2,w_2^1+w_2^2,F_1, z_1, F_2, z_2$ with 
        $(F_1, z_1) \in A_{x_1}[a^1,\omega_1^1, \omega_2^1][f_1, c]$,
        $(F_2, z_2) \in A_{x_2}[a^2,\omega_1^2, \omega_2^2][f_2, c]$, such that $F = F_1\dot\cup F_2$ and $z = z_1 \cup z_2$. It follows that
        \begin{align*}
	  T_x[a,\omega_1, \omega_2][f,c]
        &\equiv_4 \big|A_x[a,\omega_1, \omega_2][f,c]\big| \\
        &= \sumstack{a^1+a^2=a\\w_1^1 + w_1^2 = \omega_1\\w_2^1 + w_2^2 = \omega_2}\sum\limits_{f_1+f_2 = f} \big|A_{x_1}[a^1,\omega_1^1, \omega_2^1][f_1, c]\big| \cdot \big|A_{x_2}[a^2,\omega_1^2, \omega_2^2][f_2, c]\big|\\
        &\equiv_4 \sumstack{a^1+a^2=a\\w_1^1 + w_1^2 = \omega_1\\w_2^1 + w_2^2 = \omega_2} \sumstack{f_1 + f_2 = f} T_{x_1}[a^1,\omega_1^1, \omega_2^1][f_1, c] \cdot T_{x_2}[a^2,\omega_1^2, \omega_2^2][f_2, c].
        \end{align*}
    \end{itemize}
\end{proof} 

As an immediate consequence of Lemma~\ref{lem:twdynprogtables}, we obtain:

\begin{lemma}\label{tw-ub-cor:tables-count-cuts}
    For each $(\omega_1,\omega_2) \in \overline W\times \overline Z$, it holds that $T_{\droot}[n-1,\omega_1, \omega_2][\phi] \equiv_4 C[\omega_1, \omega_2]$.
\end{lemma}

\para{Efficient Join Operations.}
In order to obtain the desired running time, we first show how one can process a join node efficiently using fast convolution methods. We base our approach on the fast convolution technique introduced by van Rooij \cite{Rooij20} that applies a multidimensional fast Fourier transformation to speed up the processing of join nodes. Instead of the usual trick of eliminating cyclic dependencies by doubling the underlying field, van Rooij applies filters over the total count of indices to cancel cyclic dependencies. We follow the same approach, in order to get the claimed running time.
We cite the following result from \cite{Rooij20} with wording revised to our setting.

\begin{longdefinition}[Multidimensional Fast Fourier Transform]
    Let $h_1, \dots, h_k \in \mathbb{N}$, and let $N = \mathbb{Z}_{h_1} \times \dots \times \mathbb{Z}_{h_k}$.
    For $f,f'\in N$, we write $f+f'$ for the componentwise addition in $\mathbb{Z}_{h_1} \times \dots \times \mathbb{Z}_{h_n}$, and we write $f\oplus f'$ for the componentwise addition in $\mathbb{Z}^k$.

    Let $T_1, T_2 \in (\zfr)^{N}$ be two tensors. We define the join operation $T_1\otimes T_2 \in (\zfr)^{N}$, where for $x\in N$ it holds that:
    \[
    T_1 \otimes T_2 [x]= \sumstack{a,b\in N\\a+b = x} T_1[a] \cdot T_2[b].
    \]
\end{longdefinition}

\begin{longlemma}[{\cite[Proposition 4]{Rooij20}}]\label{tw-ub-lem:mltdimfft}
    Given two tensors $T_1, T_2 \in (\zfr)^{N}$, one can compute the tensor $T_1\otimes T_2$ in time $\ostar(|N| \cdot \log |N|)$.
\end{longlemma}

Based on this lemma, the authors of \cite{Rooij20} show that one can compute acyclic convolution in the same running time (up to a polynomial factor), by filtering the indices by their total sum. This allows to cancel out cyclic dependencies without doubling the base of the exponent. We use the same technique to compute the join operation in our algorithm.

\begin{lemma}\label{tw-ub-lem:acyclic-conv}
    Given two tensors $T_1, T_2 \in (\zfr)^{N}$, one can compute
    in time $\ostar(|N| \cdot \log |N|)$ the tensor $T\in(\zfr)^{N}$ defined by
        \begin{equation}\label{eq:acyclic-conv}
        T[f] = \sum\limits_{f_1\oplus f_2 = f} T_1[f_1] \cdot T_2[f_2].
    \end{equation}
    \end{lemma}

\begin{proof}
    Let $K = \sum_{i=1}^k h_i$. For $\kappa \in [K]_0$, we define the tables $T_1^{\kappa}, T_2^{\kappa} \in \zfr^{N}$, where for $x\in[2]$ and $f\in N$ it holds that 
    \[
    T_x^{\kappa}[f] = \begin{cases}
        T_x[f] & \text{if } \sum_{j=1}^k f_j = \kappa,\\
        0 & \text{otherwise.}
    \end{cases}
    \]
    We define the table $T'\in \zfr^{N}$, where for $f\in N$ and $\kappa = \sum_{j=1}^k f_j$  it holds that
    \[
    T'[f]= \sum_{j=0}^{\kappa} \sum\limits_{f_1+ f_2 = f} T_1^{j}[f_1] \cdot T_2^{\kappa-j}[f_2] = \sum_{j=0}^{\kappa} (T^{j}_1\otimes T^{\kappa-j}_2)(f).
    \]
    Clearly, $T'[f]$ can be computed by $\kappa+1$ applications of the Multidimensional Fast Fourier Transform, and hence, by \cref{tw-ub-lem:mltdimfft}, it can be computed in time $\ostar(|N| \cdot \log |N|)$.

    We claim that $T' = T$. To see this, let $T_1[f_1]\cdot T_2[f_2]$ be a summand of \cref{eq:acyclic-conv}. Then it must hold for $k_1 = \sum_{j=1}^k (f_1)_j$, $k_2 = \sum_{j=1}^k (f_2)_j$ and $\kappa =  \sum_{j=1}^k (f)_j$ that $k_1 + k_2 = \kappa$. Hence, it holds that $T'[f]$ contains the summand $T_1^{k_1}[f_1] \cdot T_2^{k_2}[f_2]$. For the other direction, let $f_1,f_2,k_1,k_2$ be such that $T_1^{k_1}[f_1]\neq 0$, $T_2^{k_2}[f_2]\neq 0$ and $T_1^{k_1}[f_1] \cdot T_2^{k_2}[f_2]$ appears in the sum of $T'[f]$. Then it holds that $k_1 + k_2 = \kappa$. Moreover, if there exists some index $j\in[k]$ such that $(f_1)_j + (f_2)_j > f_j$, then it must hold by the total sum, that there exists some index $j'\in[k]$ such that $(f_1)_{j'} + (f_2)_{j'} < f_{j'}$, and hence  $(f_1)_{j'} + (f_2)_{j'} \not\equiv_{\mathbb{Z}_{j'}} f_{j'}$, which contradicts the assumption, that this summand appears in the sum of $T'[f]$. Hence, it must hold for each $j\in [k]$ that $(f_1)_j + (f_2)_j = f_j$, and hence $f_1\oplus f_2 = f$. It follows that $T'[f]$ contains the summand $T_1[f_1]\cdot T_2[f_2]$. This shows that $T' = T$.
\end{proof}

Now we show how to compute the tables at a join node in the claimed running time:

\begin{lemma}\label{tw-ub-lem:fast-join-node}
    Let $x$ be a join node, and let $a\in \overline M$, $\omega_1\in \overline W$, and $\omega_2\in \overline Z$. Then the table $T_x[a, \omega_1, \omega_2]$ can be computed
    in time $\ostar\big((2r+2)^{\tw}\big)$.
\end{lemma}

\begin{proof}
    The algorithm iterates over all values $a^1,a^2,\omega_1^1,\omega_1^2,\omega_2^1, \omega_2^2$ such that $a^1+a^2=a$, $\omega_1^1+\omega_1^2 = \omega_1$ and $\omega_2^1+\omega_2^2 = \omega_2$. For each such tuple, and for each $c\in\{\LL,\RR\}^{B_x}$, let $T_1^c[f] = T_{x_1}[a^1,\omega_1^1,\omega_2^1][f,c]$ and $T_2^c[f] = T_{x_2}[a^2,\omega_1^2,\omega_2^2][f,c]$. The algorithm computes the table $T^c = T_1^c\otimes T_2^c$. Let $T'\in (\zfr)^{\indices_x}$ with $T'[f,c] = T^c[f]$. The algorithm adds the table $T'$ componentwise to $T_x$.

    By definition, it holds that $T_{x_1}[a^1,\omega_1^1,\omega_2^1][f_1,c_1] \cdot T_{x_2}[a^2,\omega_1^2,\omega_2^2][f_2,c_2]$ appears in the sum of $T_x[a,\omega_1,\omega_2][f,c]$ if and only if it holds that $a^1+a^2=a$, $\omega_1^1+\omega_1^2 = \omega_1$ and $\omega_2^1+\omega_2^2 = \omega_2$, $c_1=c_2=c$ and $f_1\oplus f_2 = f$. Hence, the algorithm computes the correct table $T_x$.

    The algorithm iterates over a polynomial number of tuples $a^1,a^2,\omega_1^1,\omega_1^2,\omega_2^1, \omega_2^2$, and over $2^{\tw + 1}$ mappings $c$, and for each the algorithm computes the table $T^c$ in time $\ostar\big((r+1)^{\tw}\big)$ by Lemma~\ref{tw-ub-lem:acyclic-conv}. This concludes the proof. Hence, the algorithm runs in time 
    $\ostar\big(2^{\tw+1}(r+1)^{\tw}\big) = \ostar\big((2r+2)^{\tw}\big)$.
\end{proof} 
With Lemma~\ref{tw-ub-lem:fast-join-node} in hand, we can obtain an upper bound on dynamically computing the tables at each node of $\syntaxtree$ and subsequently use this to establish the main result of this section:

\begin{lemma} \label{tw-ub-lem:running-time}
    All tables $T_x[a, \omega_1, \omega_2]$ can be computed in time $\ostar((2r+2)^{\tw} \cdot W^{\oh(1)} \cdot Z^{\oh(1)})$.
\end{lemma}

\begin{proof}
    The algorithm computes all tables $T_x$ by dynamic programming in a bottom up manner over the decomposition tree. For each node $x$, the algorithm iterates over all values $(a, \omega_1, \omega_2) \in \overline{M} \times \overline{W} \times \overline{Z}$, and for each such tuple the algorithm computes $T_x[a, \omega_1, \omega_2]$ as follows: let $x'$ be the child of $x$ if $x$ has a single child or $x_1$ and $x_2$ be the children of $x$ if it has two children.
    \begin{itemize}
        \item For a leaf node $x$, the algorithm initializes $T_x[0,0,0][\phi] = 1$ and $T_x[a, \omega_1, \omega_2][f,c] = 0$ for all other values $a,\omega_1,\omega_2,f,c$.
        \item Introduce vertex node $(v)$, the algorithm iterates over all indices $(f',c')\in\indices_{x'}$. For each the algorithm sets $T_x[a, \omega_1, \omega_2][f,c] = T_{x'}[a, \omega_1, \omega_2][f',c']$ where $f' = f[v\mapsto 0]$ for both extensions $c'$ of $c$ by $v\mapsto \LL$ and $v\mapsto \RR$. The algorithm sets $T_x[a, \omega_1, \omega_2][f,c] = 0$ for all other values $f,c$.
        \item Introduce edge $e=\{u,v\}$: the algorithm iterates over all indices $(f,c)\in\indices_{x'}$. For each the algorithm sets $T_x[a, \omega_1, \omega_2][f,c] = T_{x'}[a, \omega_1, \omega_2][f,c]$. After that, for each coloring $c$ such that $c(u) = c(v)$ and $f(u),f(v)<d$, the algorithm adds $T_{x'}[a-1, \omega_1-w(e), \omega_2-w'(e)][f,c]$ to $T_{x}[a, \omega_1, \omega_2][f',c]$ where $f' = f[u\mapsto f(u)+1, v\mapsto f(v)+1]$.
        \item Forget vertex $v$: the algorithm iterates over all indices $(f,c)\in\indices_{x}$. For each the algorithm sets $T_x[a, \omega_1, \omega_2][f,c] = 0$. After that the algorithm iterates over all indices $(f',c')\in\indices_{x'}$. For each index $f'$ with $f'(v)\in \dsets(v)$, the algorithm adds the value of $T_{x'}[a, \omega_1, \omega_2][f',c']$ to $T_x[f,c]$ where $f$ is the restriction of $f'$ to $B_x$ and $c$ is the restriction of $c'$ to $B_x$.
        \item Join node: This follows by \Cref{tw-ub-lem:fast-join-node} by iterating over all values $a, \omega_1, \omega_2$.
    \end{itemize}
    Clearly, the number of values $a^1, a^2, \omega_1^1, \omega_1^2, \omega_2^1, \omega_2^2$ the algorithm iterates over for each node $x$ is polynomial in $n$, $W$ and $Z$. For a leaf node, an introduce vertex, introduce edge or forget vertex node, the algorithm iterates over each index of $\indices_x$ and $\indices_{x'}$ at most once, and for each index the algorithm performs a constant number of operations. Hence, all these nodes can be processed in time $\ostar((2r+2)^{\tw})$.
    
    For a join node, each table $T_{a^1,a^2,\omega_1^1,\omega_1^2,\omega_2^1,\omega_2^2}$ can be computed in time $\ostar((2r+2)^{\tw})$ by \cref{tw-ub-lem:fast-join-node}. Since the total number of tuples $(a^1,a^2,\omega_1^1,\omega_1^2,\omega_2^1,\omega_2^2)$ is polynomial in $n$, $W$ and $Z$, the algorithm computes all tables $T_x$ in time $\ostar((2r+2)^{\tw} \cdot W^{\oh(1)} \cdot Z^{\oh(1)})$.
\end{proof} 
\begin{proof}[Proof of \cref{thm:tw-ub}]
    The algorithm fixes $Z = 2m = n^{\oh(1)}$, and computes all tables $T_x$ for all nodes $x$ in time $\ostar((2r+2)^{\tw} \cdot W^{\oh(1)})$ by \cref{tw-ub-lem:running-time}.
    The algorithm accepts, if there exists a value $\omega_2\in\overline Z$ such that $T_{\droot}[n-1,\omega_1,\omega_2][\phi] = 2$ and rejects otherwise.

    It holds by \cref{tw-ub-cor:isolation} that with probability at least $1/2$ there exists a value $\omega_2\in\overline Z$ such that $\mathcal{S}[\omega_1,\omega_2]$ contains a single solution only, if a solution exists, and none otherwise. It holds by \cref{tw-ub-cor:tables-count-cuts} that $T_{\droot}[n-1,\omega_1,\omega_2][\phi] \equiv_4 C[\omega_1,\omega_2]$, and by \cref{tw-ub-lem:cnc} that $C[\omega_1,\omega_2] \equiv_4 2$ if $\big|\mathcal{S}[\omega_1,\omega_2]\big| = 1$. Hence, if a solution exists the algorithm accepts with probability at least $1/2$, and rejects otherwise.
\end{proof}

\section{Pathwidth}\label{sec:pw}
\noindent
In this section, we provide an upper bound that uses pathwidth to supersede the one obtained in Theorem~\ref{thm:tw-ub} and also a corresponding lower bound.

\subsection{The Upper Bound}
Our aim here is to prove:
\pwub*

Let $(\syntaxtree, \mathcal{B})$ be a nice path decomposition of width $\pw$, and $r$ the root of $\syntaxtree$.
For a subgraph $H$ of $G$, let $q_H(v) = \min\{\deg_H(v), D(v)\}$ for each vertex $v\in V(H)$.
For $x\in \nodes$, let $q_x = q_{G_x}$.

\para{Dynamic programming tables.}
To obtain our tight runtime bound, we need to extend the notion of coloring as follows:

\begin{definition}
For a vertex set $S$, we call a mapping $c\colon S\rightarrow\{\LL,\RR,\noc\}$ a \defi{partial coloring}. We define the consistency relation $\sim$ over $\{\LL,\RR,\noc\}$ with $x\sim y$ if and only if $\{x,y\}\neq\{\LL,\RR\}$. For a graph $H$, we call a partial coloring $c: V(H)\rightarrow \{\LL,\RR,\noc\}$ \defi{consistent}, if it holds for each edge $\{u,v\}\in E(H)$ that $c(u)\sim c(v)$.
\end{definition}

The dynamic programming table at a node $x$ is indexed by the family $\indices_x$, defined below.

\begin{definition}
Let $\mathcal{D}_x \subseteq \mathbb{Z}^{B_x}$ be the family of all vectors that are upper bounded (componentwise) by $q_x$, i.e.,
\[
\mathcal{D}_x = \{ u \in \mathbb{Z}^{B_x} \colon u_v \leq q_x(v) \text{ for all } v \in B_x \},
\]
and let $\mathcal{C}_x = \{\LL, \RR, \noc\}^{B_x}$. We define
\[
\indices_x = \big\{(f,c) \in \mathcal{D}_x\times\mathcal{C}_x \colon c(v) = \noc \iff f_v \in \{0, d(v)\} \big\}.
\]
\end{definition}

The algorithm we use for pathwidth then proceeds essentially analogously as Algorithm~\ref{alg:tw}, with the following distinction: if a vertex $v$ in the bag has degree $0$ then we do not need to fix its side in the consistent cut, while if it has degree $r$ then we need not store it. Implementing this change yields dynamic programming tables $\hat T$ that require two fewer states per vertex, and this also directly translates into the running time since we do not deal with join nodes.

Formally, the tables are defined recursively over $\mathcal{T}$ as follows:

\begin{longalgorithm}\label{pw-ub-def:tables}
    Let $x \in \nodes$, and let $x'$ be the child of $x$ (if exists). For $(a,\omega_1,\omega_2)\in\overline M \times \overline W \times \overline Z$,
    we define the tables $\hat{T}_x[a,\omega_1,\omega_2]\in (\zfr)^{\indices_x}$, where for $(f,c)\in \indices_x$ we define the value of $\hat{T}_x[a, \omega_1,\omega_2][f,c]$ depending on the kind of node $x$ as follows:
         \begin{itemize}
        \item For a leaf node $x$, 
        let $\phi$ be the pair $(\emptyset_v, \emptyset_c)$, where $\emptyset_v$ is a the vector over an empty set, and $\emptyset_c$ is the empty partial coloring.
        It holds that $(f,c)=\phi$. We set $\hat{T}_x[0,0,0][f,c] = 1$, and $\hat{T}_x[a,\omega_1, \omega_2][f,c] = 0$ for all other values $a, \omega_1, \omega_2$.

        \item Introduce vertex $v$: let $f', c'$ be the restrictions of $f$ and $c$ to $B_{x'}$ respectively. Since $q_x(v)\leq\deg_x(v) = 0$, it must hold that $f(v) = 0$ and $c(v)= \noc$. We set
        \[
            \hat{T}_x[a, \omega_1, \omega_2][f,c] = \hat{T}_{x'}[a, \omega_1, \omega_2][f', c'].
        \]
        
        \item Forget vertex $v$:
        We define
        \[
            \hat{T}_x[a, \omega_1, \omega_2][f,c]= \sumstack{o\in \dsets(v)\\X\in\{\LL,\RR,\noc\}\\X=\noc \iff o = d(v)}  \hat{T}_{x'}[a, \omega_1, \omega_2]\big[f[v\mapsto o], c[v\mapsto X]\big],
        \]

        \item Introduce edge $e=\{u,v\}$:
        If $f(u) = 0$ or $f(v) = 0$, or $c(u)$ and $c(v)$ are not consistent, we set $\hat{T}_x[a, \omega_1, \omega_2][f,c] = \hat{T}_{x'}[a, \omega_1, \omega_2][f,c]$. Assuming this is not the case, let $a'=a-1$, $\omega'_1 = \omega_1-w(e)$ and $\omega'_2 = \omega_2-w'(e)$. Let $f' = f[u\mapsto f_u - 1, v\mapsto f_v -1]$. If $f_u = d_u$ and $f_v = d_v$, we define $Q_u = Q_v = \{\LL, \RR\}$. Otherwise, it holds that the set $Q = \{\LL, \RR\} \cap \{f_u, f_v\}$ is a singleton. Let $Q_u = Q$ if $f(u) \geq 2$ or $Q_u = \{\uparrow\}$ otherwise, and $Q_v = Q$ if $f_v \geq 2$ or $Q_v = \{\uparrow\}$ otherwise. We set
        \begin{align*}
            \hat{T}_x[a, \omega_1, \omega_2][f,c] = &\hat{T}_{x'}[a, \omega_1, \omega_2][f,c]\\
            &+ \sumstack{X_u \in Q_u, X_v\in Q_v\\X_u\sim X_v} \hat{T}_{x'}\big[a', \omega'_1, \omega'_2\big]\big[f', c[u\mapsto X_u, v\mapsto X_v]\big]. \xqed
        \end{align*}
    \end{itemize}
\end{longalgorithm}

\begin{longdefinition}
Let $x\in \nodes$. For a mapping $f\colon B_x\rightarrow \dvecs_x$, we define the set $Q_f = \{v\in B_x\colon f(v) < d(v)\}$. For
$(a, \omega_1, \omega_2)\in \overline M \times \overline W \times \overline Z$, and $(f,c) \in \indices_x$ we define the sets $A_x[\omega_1, \omega_2][f,c]$ as follows:
\begin{align*}
    A_x[a, \omega_1, \omega_2][f,c] = \{(F, z)\in& 2^{E_x}\times \{ \LL,\RR,\noc \}^{V_x}\colon\\
    &|F| = a, w(F) = \omega_1, w'(F) = \omega_2\ \land\\
    &\deg_F(v) \in \dsets(v) \text{ for all } v \in V_x \setminus B_x\ \land\\
    &\deg_F(v) = f(v) \text{ for all } v \in B_x\ \land\\
    &z|_{Q_x} = c|_{Q_f}\ \land\\
    &z(v)\neq \noc \text{ for all } v\in V_x \setminus Q_f\ \land\\
    &z \text{ is a consistent partial coloring of } G_x\}.
\end{align*}
Intuitively, the tables $A_x$ store the set of all consistent partial colorings $z$ of all partial solutions $F$ in $G_x$ that have the footprint $(f,c)$. In particular, $z$ is an extension of $c|_{Q_f}$ that assigns $\LL$ or $\RR$ to each vertex of $V_x\setminus Q_f$. 
\end{longdefinition}

\begin{longlemma}    
\label{lem:pwcomputation}
Let $x\in\nodes$. It holds for all $(a, \omega_1, \omega_2)\in \overline M \times \overline W \times \overline Z$, $(f,c)\in \indices_{x}$ that
\[
\hat{T}_x[a, \omega_1, \omega_2][f,c] \equiv_4 \big|A_x[a, \omega_1, \omega_2][f,c]\big|.
\]
\end{longlemma}

\begin{proof}
    We prove the claim by induction over $\dtree$. The base case is a leaf node $x$, where the claim holds by definition, since $G_x$ is an empty graph. Let $x$ be an internal node, and assume that the claim holds for $x'$ the child of $x$. We show that the claim holds for $x$ by distinguishing the different kinds of nodes $x$.
    \begin{itemize}
        \item Introduce vertex $v$: It holds that $\deg_x(v) = 0$, and hence, $f(v) = 0$ and $c(v) = \noc$. Let $f', c'$ be the restrictions of $f$ and $c$ to $B_{x'}$ respectively. Then there exists a bijection between $A_x[a, \omega_1, \omega_2][f,c]$ and $A_{x'}[a, \omega_1, \omega_2][f',c']$, by assigning to a pair $(F, z)\in A_x[a, \omega_1, \omega_2][f',c']$ the pair $(F, z')\in A_{x'}[a, \omega_1, \omega_2][f,c]$, where $z'$ is the restriction of $z$ to $B_{x'}$. Hence, it holds that 
        \[
            \hat{T}_x[a,\omega_1, \omega_2][f,c] 
            \equiv_4 \big|A_x[a,\omega_1, \omega_2][f,c]\big|
            = \big|A_{x'}[a,\omega_1, \omega_2][f',c']\big|
            \equiv_4 \hat{T}_{x'}[a,\omega_1, \omega_2][f',c'].
        \]
        
        \item Forget vertex $v$: It holds that $A_x[a,\omega_1, \omega_2][f,c]$ equals the disjoint union of all sets $A_{x'}\big[a,\omega_1, \omega_2\big]\big[f[v\mapsto o], c[v\mapsto X]\big]$ for all $o\in \dsets(v)$ and $X\in \{\LL,\RR,\noc\}$ where $X = \noc$ if and only if $o = d(v)$. This holds since for each pair $(F, z) \in A_x[a,\omega_1, \omega_2][f,c]$ it must hold that $\deg_F(v) \in \dsets(v)$, and hence, for $o = \deg_F(v)$, and $X = z(v)$ if $o < d(v)$ or $X = \uparrow$ otherwise, that $(F, z)\in A_{x'}\big[a,\omega_1, \omega_2\big]\big[f[v\mapsto o], c[v\mapsto X]\big]$. Hence, it holds that
        \begin{align*}
            \hat{T}_x[a,\omega_1, \omega_2][f,c] 
            &\equiv_4 \big|A_x[a,\omega_1, \omega_2][f,c]\big| \\
            &= \sumstack{o\in \dsets(v)\\X\in \{\LL,\RR,\noc\}\\o=d(v)\iff X=\noc} \big|A_{x'}[a,\omega_1, \omega_2][f[v\mapsto o], c[v\mapsto X]]\big| \\
            &\equiv_4 \sumstack{o\in \dsets(v)\\X\in\{\LL,\RR,\noc\}\\o=d(v)\iff X=\noc}
            \hat{T}_{x'}[a,\omega_1, \omega_2][f[v\mapsto o], c[v\mapsto X]].
        \end{align*}

        \item Introduce edge $e=\{u,v\}$: If $f(u) = 0$ or $f(v) = 0$, or $c(u)$ and $c(v)$ are not consistent, then it holds for all $(F, z) \in A_x[a, \omega_1, \omega_2][f,c]$ that $e\notin F$. It follows in this case that $A_x[a, \omega_1, \omega_2][f,c] = A_{x'}[a, \omega_1, \omega_2][f,c]$. Assuming this is not the case, then it holds for each pair $(F, z)\in A_x[a, \omega_1, \omega_2][f,c]$ that either $e\notin F$, and hence $(F, z)\in A_{x'}[a, \omega_1, \omega_2][f,c]$, or $e\in F$. In this case, let $f' = f[u\mapsto f(u)-1, v\mapsto f(v) - 1]$, $a'=a-1$, $\omega'_1 = \omega_1 - w(e)$ and $\omega'_2 = \omega_2 - w'(e)$. Then $(F\setminus \{e\}, z)$ belongs to some set $A_{x'}\big[a', \omega'_1, \omega'_2\big]\big[f', c'\big]$ for $c' = c[u\mapsto X_u, v\mapsto X_v]$ for some values $X_u, X_v$ where $X_u \sim X_v$.
        Note that the sets $Q_u$ and $Q_v$ from \cref{pw-ub-def:tables} are the sets of values $X_u$ and $X_v$ that result in $c(u)$ and $c(v)$ by increasing the degree of $u$ and $v$ by one.
        Hence, it holds that 
        \begin{align*}
            &\hat{T}_x[a,\omega_1, \omega_2][f,c] \equiv_4 \big|A_x[a,\omega_1, \omega_2][f,c]\big|\\
            &= \left|A_{x'}[a,\omega_1, \omega_2][f,c]
            \dot\cup \dot{\bigcup\limits_{\substack{X_u \in Q_u, X_v\in Q_v\\X_u\sim X_v}}} A_{x'}\big[a', \omega'_1, \omega'_2\big]\big[f', c[u\mapsto X_u, v\mapsto X_v]\big]\right|\\
            &\equiv_4\hat{T}_{x'}[a, \omega_1, \omega_2][f,c]
            + \sumstack{X_u \in Q_u, X_v\in Q_v\\X_u\sim X_v} \hat{T}_{x'}\big[a', \omega'_1, \omega'_2\big]\big[f', c[u\mapsto X_u, v\mapsto X_v]\big].
        \end{align*}
    \end{itemize}
\end{proof} 

As an immediate consequence of Lemma~\ref{lem:pwcomputation}, we obtain:

\begin{lemma}\label{pw-ub-cor:tables-count-cuts}
    It holds for each value $w\in\mathbb{N}$ that $\hat{T}_{\droot}[a,\omega_1,\omega_2][\phi] \equiv_4 C[\omega_1,\omega_2]$.
\end{lemma}

\para{Running Time.}
Our next task is to establish the running time of the algorithm. 
First, we obtain the following straightforward bound on the size of $\indices_x$.

\begin{lemma} \label{pw-ub-lem:indices-size}
    For all $x\in \nodes$, it holds that $|\indices_x| \leq (2\maxreq)^{\pw+1}$.
\end{lemma}

\begin{proof}
    It holds that 
    \begin{align*}
        |\indices_x|
        &\leq \prod\limits_{v\in B_x} (2(d(v)-1) + 2)\\
        &\leq \prod\limits_{v\in B_x} (2(d(v)-1) + 2)\\
        &\leq \prod_{v\in B_x} 2d(v) \leq (2\maxreq)^{|B_x|} \leq (2\maxreq)^{\pw+1},
    \end{align*}
    where the first inequality holds, since there exists at most $b(v) + 1 \leq d(v)+1$ choices of degrees for each vertex $v\in B_x$, where for at most $d(v)-1$ of them we can assign either color $\LL$ or $\RR$, and for the rest we only assign $\noc$.
\end{proof} 
We note that the formulation of the next statement differs from the formulation of Lemma~\ref{tw-ub-lem:running-time}; this is because Lemma~\ref{pw-ub-lem:running-time-relto-indices} will also be useful when establishing the cutwidth upper bound in Section~\ref{sec:ctw}.

\begin{lemma} \label{pw-ub-lem:running-time-relto-indices}
    Let $s$ be the maximum size of $\indices_x$ over all nodes $x\in\nodes$. Then all tables $\hat{T}_x[a, \omega_1, \omega_2]$ can be computed in time $\ostar(s \cdot W^{\oh(1)} \cdot Z^{\oh(1)})$.
\end{lemma}

\begin{proof}
    The algorithm computes all tables $\hat{T}_x$ by dynamic programming in a bottom up manner over the decomposition tree. For each node $x$, the algorithm iterates over all values $(a, \omega_1, \omega_2) \in \overline{M} \times \overline{W} \times \overline{Z}$, and for each such tuple the algorithm computes $\hat{T}_x[a, \omega_1, \omega_2]$ as follows: let $x'$ be the child of $x$ if exists.
    \begin{itemize}
        \item For a leaf node $x$, the algorithm initializes $\hat{T}_x[0,0,0][\phi] = 1$, and $\hat{T}_x[a, \omega_1, \omega_2][\phi] = 0$ for all other values $a,\omega_1,\omega_2$.
        \item Introduce vertex node $(v)$, the algorithm iterates over all indices $(f',c')\in\indices_{x'}$. For each the algorithm sets $\hat{T}_x[a, \omega_1, \omega_2][f,c] = \hat{T}_{x'}[a, \omega_1, \omega_2][f',c']$ where $f' = f[v\mapsto 0]$ and $c' = c[v\mapsto \noc]$.
        \item Introduce edge $e=\{u,v\}$: the algorithm iterates over all indices $(f,c)\in\indices_{x'}$. For each the algorithm sets $\hat{T}_x[a, \omega_1, \omega_2][f,c] = \hat{T}_{x'}[a, \omega_1, \omega_2][f,c]$. After that,
        for each pair $(f, c)\in\indices_{x'}$ such that $c(u)\sim c(v)$ and $f(u),f(v)<d$, let $f' = f[u\mapsto f(u)+1, v\mapsto f(v)+1]$. Then for each $c' = c[u\mapsto X_u, v\mapsto X_v]$ for $X_u, X_v\in\{\LL,\RR,\noc\}$ such that $X_u\sim X_v$, $X_u\sim c(u)$ and $X_v\sim c(v)$ and $(f',c')\in\indices_x$, the algorithm adds $\hat{T}_x'[a-1, \omega_1-w(e), \omega_2-w'(e)][f,c]$ to $\hat{T}_{x}[f',c']$.
        \item Forget vertex $v$: the algorithm iterates over all indices $(f,c)\in\indices_{x}$. For each the algorithm sets $\hat{T}_x[a, \omega_1, \omega_2][f,c] = 0$. After that the algorithm iterates over all indices $(f',c')\in\indices_{x'}$. For each index $(f',c')$ with $f'(v)\in \dsets(v)$, the algorithm adds the value of $\hat{T}_{x'}[a, \omega_1, \omega_2][f',c']$ to $\hat{T}_x[f,c]$ where $f=f'|_{B_x}$ and $c=c'|_{B_x}$.
    \end{itemize}
    Clearly, the number of values $a^1, a^2, \omega_1^1, \omega_1^2, \omega_2^1, \omega_2^2$ the algorithm iterates over for each node $x$ is polynomial in $n$, $W$ and $Z$. For each node, the algorithm iterates over each index of $\indices_x$ and $\indices_{x'}$ a constant number of times, and for each index the algorithm performs a constant number of operations. Hence, all these nodes can be processed in time $\poly(n, W, Z)\cdot c\cdot s = \ostar(s \cdot W^{\oh(1)} \cdot Z^{\oh(1)})$, where $c$ is some constant.
\end{proof} 
Lemma~\ref{pw-ub-lem:running-time-relto-indices} together with \cref{pw-ub-lem:indices-size} allows us to obtain the desired running time bound for Theorem~\ref{thm:pw-ub}, and the proof of that theorem then follows analogously to the one of Theorem~\ref{thm:tw-ub}.

\begin{longcorollary} \label{pw-ub-cor:running-time}
    All tables $\hat{T}_x[a,\omega_1,\omega_2]$ can be computed in time 
    $\ostar\big((2r)^{\pw+1} \cdot W^{\oh(1)} \cdot Z^{\oh(1)}\big)$.
\end{longcorollary}

\begin{proof}
    This follows directly from \Cref{pw-ub-lem:running-time-relto-indices} since $s\leq (2r)^{\pw+1}$ holds by \Cref{pw-ub-lem:indices-size}.
\end{proof} 
\begin{proof}[Proof of \cref{thm:pw-ub}]
    The algorithm fixes $Z = 2m = n^{\oh(1)}$, and computes all tables $\hat{T}_x$ for all nodes $x$ in time $\ostar((2d)^{\pw} \cdot W^{\oh(1)})$ by \cref{pw-ub-cor:running-time}.
    The algorithm accepts, if there exists a value $\omega_2\in\overline Z$ such that $\hat{T}_{\droot}[n-1,\omega_1,\omega_2][\phi] = 2$ and rejects otherwise.

    It holds by \cref{tw-ub-cor:isolation} that with probability at least $1/2$ there exists a value $\omega_2\in\overline Z$ such that $\mathcal{S}[n-1, \omega_1,\omega_2]$ contains a single solution only, if a solution exists, and none otherwise. It holds by \cref{pw-ub-cor:tables-count-cuts} that $\hat{T}_{\droot}[n-1,\omega_1,\omega_2][\phi] \equiv_4 C[\omega_1,\omega_2]$, and by \cref{tw-ub-lem:cnc} that $C[\omega_1,\omega_2] \equiv_4 2$ if $\big|\mathcal{S}[\omega_1,\omega_2]\big| = 1$. Hence, if a solution exists the algorithm accepts with probability at least $1/2$, and rejects otherwise.
\end{proof}

\subsection{Lower Bound}
We now move on to establish our first tight lower bound. We recall that $\maxreq$ is defined as the maximum degree requirement of any vertex in a given instance. 

\twlb*

For this proof, we need two ingredients.
The first ingredient is the lower bound result for the $q$-CSP-$\cspdom$ problem.
In this problem, we are given a set $X$ of $n$ variables that take values in $[\cspdom]$ and a set of $q$-constraints on $X$.
A \defi{$q$-constraint} is defined by a tuple of $q$ variables of $X$ 
and a subset of $[\cspdom]^q$ of satisfying assignments. An assignment $\alpha \colon X \to [\cspdom]$ satisfies a $q$-constraint if the tuple of values assigned to the variables in the constraint belongs to the subset of satisfying assignments defined by the constraint.
The task is then to decide if there is an assignment of the variables that satisfies all the $q$-constraints.
Lampis~\cite{MR4121882} showed the following hardness result.

\begin{theorem}[\cite{MR4121882}]
\label{thm:csp}
	For any $\cspdom \geq 2$, $\varepsilon > 0$, assuming SETH, there exists $q$ such that $n$-variable $q$-CSP-$\cspdom$ cannot be solved in time $\ostar((\cspdom-\varepsilon)^n)$.
\end{theorem}

The second ingredient is a gadget adopted from~\cite{CyganKN18}.
This gadget relies on the following auxiliary definition.

\begin{longdefinition}
A \defi{label gadget} is a pair $(v, \lambda_v)$ where $\lambda_v$ is a labeling of the edges incident to $v$.
For an unweighted \sdmstp instance $(G,\degreq)$ and a vertex $v$ with $\degreq(v) = 2$, a solution spanning tree $T$ is \defi{consistent} with a label gadget $(v, \lambda_v)$ if $\lambda_v(e) = \lambda_v(e')$ where $e,e'$ are the two edges of $C$ incident to $v$. 
\end{longdefinition}

Then we can phrase the gadget as follows.

\begin{longlemma}[\cite{CyganKN18}]
\label{lem:consistent_labels}
	Let $(G,\degreq)$ be an unweighted \sdmstp instance.
	Let $v$ be a vertex of $G$ with degree requirement two, incident edges $X$, and a labeling $\lambda_v : X \to [\ell]$ for some $\ell$.
	Then we can obtain a new unweighted \sdmstp instance $(G',\degreq')$ from the original instance $(G,\degreq)$ by replacing $v$ by a gadget of $\bigoh(\ell)$ vertices with degree requirement two.
	Further, $(G,\degreq)$ has a solution spanning tree consistent with $(v,\lambda_v)$, if and only if the new instance $(G',\degreq')$ has a solution.
\end{longlemma}

Given $\maxreq \geq 3$ and $\varepsilon > 0$,
let $q$ be the smallest integer such that $n$-variable $q$-CSP-$(2\maxreq)$ does not admit an  $\ostar((2\maxreq-\varepsilon)^n)$-time algorithm. 
By \cref{thm:csp}, $q$ exists, and it only depends on $2\maxreq$ and $\varepsilon$, assuming SETH.
We will argue that if there exists an algorithm that solves \sdmstp in time $\ostar((2\maxreq-\varepsilon)^{\pw})$ where $\pw = n + g(\maxreq,\varepsilon)$ for some computable function $g$, then there exists an algorithm that solves $q$-CSP-$(2\maxreq)$ in $\ostar((2\maxreq-\varepsilon)^n)$, a contradiction.
 Let $\phi$ be a $q$-CSP-$(2\maxreq)$ instance with $n$ variables $x_1,\dots, x_n$
and $m$ $q$-constraints $C_1,\dots,C_n$.
In particular, we provide a Turing reduction, where we create $\Theta(n\maxreq(n-\maxreq))$ instances of \sdmstp, each of which is parameterized by two parameters $\mu \in [n]_0$ and $\nu \in [n(\maxreq-1)]_0$, and we show that $\phi$ is a YES-instance if and only if one of the \sdmstp instances is a YES-instance.

\para{Variable gadget.} For $i \in [n]$, we construct a variable gadget for $x_i$ as follows; see \cref{fig:tw_lb}(a) for an illustration.
We create $m+1$ vertices $v^{i,0}, \dots, v^{i,m}$.
For $j \in [m]$, we add $\maxreq$ paths of length three between $v^{i,j-1}$ and $v^{i,j}$.
Let each path be $(v^{i,j-1}, u^{i,j}_{j'}, \tilde{u}^{i,j}_{j'}, \bar{u}^{i,j}_{j'},  v^{i,j})$ for $j' \in [\maxreq-1]_0$.
We call $u^{i,j}_{j'}$ a \defi{$j'$-th $+$-neighbors} of $v^{i,j-1}$ and call $\bar{u}^{i,j}_{j'}$ a \defi{$j'$-th $-$-neighbors} of $v^{i,j}$.
Additionally, we call $u^{i,j}_{0}$ and $\bar{u}^{i,j}_{0}$ \defi{root-connected-neighbors} of $v^{i,j-1}$ and $v^{i,j}$, respectively.

Denote by $\tilde{V}$ the set of vertices $\tilde{u}^{i,j}_{j'}$ for all possible $i,j,j'$. 

\para{Constraint gadget.} 
Next, we construct the constraint gadget; refer to \cref{fig:tw_lb}(b) for an illustration.
We define a \defi{state} as a tuple $(c,d)$ where $c \in \{+,-\}$ and $d \in [\maxreq-1]_0$.
Let $\sigma: \{+,-\} \times [\maxreq-1]_0 \to [2\maxreq]$ be the function that maps a state to an integer as follows.
If $c = -$, then $\sigma(c,d) = d + 1$; otherwise, $\sigma(c,d) = \maxreq + d + 1$.
It is easy to see that this mapping is a bijection between the set of all states and $[2\maxreq]$.

Suppose the $q$-constraint $C_j$ has satisfying assignments $\{A_1, \dots, A_{p_j}\} \subseteq [2\maxreq]^q$.
All the following new vertices have $j$ as part of the superscript.
However, for ease of notation, we omit this superscript.
Let $I_j$ be the set of indices $i$ such that $x_i$ is a variable of $C_j$; note that $|I_j| = q$.
For $t \in [p_j]$, suppose $A_t$ is the assignment of $x_i = a_{i,t}$ for $i \in I_j$ and some value $a_{i,t} \in [2\maxreq]$. 
Then we create $q \cdot \maxreq$ new vertices $w^{t}_{i,j'}$ for $i \in I_j$ and $j' \in [\maxreq-1]_0$, and we connect them with a path such that the subscripts appear in the lexicographical order.
We call this path $P^j_t$.
For $i \in I_j$, let $(c,d)$ be the state $\sigma^{-1}(a_{i,t})$.
If $c=+$, we add the edge $w^{t}_{i,0} u^{i,j}_{0}$; otherwise, we add the edge $w^{t}_{i,0} \bar{u}^{i,j}_{0}$. 
Next, for $j' \in [d]$, we add the edge $w^{t}_{i,j'} u^{i,j}_{j'}$, and for $j' \in [r-1] \setminus [d]$, we add the edge $w^{t}_{i,j'} \bar{u}^{i,j}_{j'}$.

Next, for $i \in I_j$ and $j' \in [\maxreq-1]_0$, we create a new vertex $s_{i,j'}$ and connect it with the vertices $w^{t}_{i,j'}$ for all $t \in [p_j]$.
We add $p_j$ new vertices $z_1, \dots, z_{p_j}$.
For $t \in [p_j]$, we denote the two endpoints of $P^j_t$ by $\overleftarrow{\pi}_t$ and $\overrightarrow{\pi}_t$.
For $t \in [p_j-1]$, we add the edges $z_t \overleftarrow{\pi}_t$ and $z_t \overleftarrow{\pi}_{t+1}$.
For $t \in \{2, \dots, p_j\}$, we add the edges $z_t \overrightarrow{\pi}_{t-1}$ and $z_t \overrightarrow{\pi}_t$.

We assign the following labels to the incident edges of $w^t_{i,j'}$ for all applicable $t,i,j'$, effectively defining label gadgets for these vertices. 
Specifically, we assign label $\alpha$ to the edges incident to $s_{i,j'}$ for $i \in I_j$ and $j' \in [\maxreq-1]_0$ and edges connecting a vertex in a variable gadget and a vertex in a clause gadget.
We assign the label $\beta$ to the edges incident to the vertices $z_t$ for $t \in [p_j]$ and the edges along the path $P^j_t$ for $t \in [p_j]$.

\para{Dummy vertices.} For $j \in [m]$, $i \in [n] \setminus I_j$, and $j' \in [\maxreq-1]_0$, we create a dummy vertex $d^{i,j}_{j'}$ and connect it with $u^{i,j}_{j'}$ and $\bar{u}^{i,j}_{j'}$.
In words, a dummy vertex is connected with the $j'$-th $+$-neighbor of $v^{i,j-1}$ and the $j'$-th $-$-neighbor of $v^{i,j}$, if $x_i$ is not a variable of the constraint $C_j$.

\begin{figure}[t!]
    \centering
    \begin{subfigure}[t]{\textwidth}
        \centering
        \begin{tikzpicture}[x=1cm, y=.8cm]
            \pic{tw-lb-vg};      
        \end{tikzpicture} 
        \caption{Variable gadget}
    \end{subfigure}     
    \begin{subfigure}[t]{\textwidth}
        \centering
        \begin{tikzpicture}[x=1.3cm, y=.8cm]
            \pic{tw-lb-cg};
        \end{tikzpicture}
        \caption{Constraint gadget}
    \end{subfigure}

    \begin{subfigure}[t]{\textwidth}
        \centering
        \begin{tikzpicture}[x=1cm, y=1cm]
            \pic{tw-lb-rg};
        \end{tikzpicture}
        \caption{Vertices in the set $R$}
    \end{subfigure}
    \caption{Gadgets for the proof of \cref{thm:tw_lb}. The white vertices are the vertices that are adjacent to a vertex in the set $R$.}
    \label{fig:tw_lb}
\end{figure}

\para{Degree absorber.}
Let $\mu \in [n]_0$ and $\nu \in [n(\maxreq-1)]_0$ be two parameters of the \sdmstp instance.
We create the vertices $\bar{a}_1, \dots, \bar{a}_{\mu}$, $a_1, \dots, a_{n-\mu}$
and the vertices $\bar{b}_1, \dots, \bar{b}_{\nu}$, $b_1, \dots, b_{n(r-1)-\nu}$. 
We add the edges $\bar{a}_t v^{i,0}$, $a_t v^{i,m}$, $\bar{b}_t v^{i,0}$, and $b_t v^{i,m}$, for all possible $i,t$.
For convenience, we call $\bar{a}_t$ a \defi{root-connected $-$-neighbor} of $v^{i,0}$ and $a_{t}$ a \defi{root-connected $+$-neighbor} of $v^{i,m}$.
Further, we call $\bar{b}_t$ a \defi{$-$-neighbor} of $v^{i,0}$ and $b_t$ a \defi{$+$-neighbor} of $v^{i,m}$.

\para{Root.} We have a special vertex called the \defi{root} $\rho$, which ensures the connectivity of all gadgets in the construction.
We add four paths $(\rho, \rho^U_1, \dots, \rho^U_m)$, $(\rho, \rho^A_1, \dots, \rho^A_{n-\mu})$, $(\rho, \rho^{\bar{A}}_1, \dots, \rho^{\bar{A}}_{\mu})$ and $(\rho, \rho^Z_1, \dots, \rho^Z_m)$.

For $j \in [m]$, we add a path $(\rho^U_j, \rho^U_{j,1}, \dots, \rho^U_{j,n})$.
For $i \in [n]$, we add the edge $\rho^U_{j,i} \tilde{u}^{i,j}_0$ for $j \in [m]$.
For $j \in [m]$, we add the edge $\rho^Z_j z^j_1$.
For $t \in [\mu]$, we add the edge $\rho^{\bar{A}}_t \bar{a}_{t}$.
For $t \in [n-\mu]$, we add the edge $\rho^A_t a_{t}$.

Let $R$ be the set of all vertices $\rho$, $\rho^U_j$, $\rho^U_{j,i}$, $\rho^A_i$, and $\rho^{\bar{A}}_i$ for all possible $i$ and $j$.
See \cref{fig:tw_lb}(c) for an illustration.

\para{Degree requirements.} Lastly, we specify the degree requirements.
For $i \in [n], j \in [m]$, the requirement for $v^{i,j}$ is $\maxreq$.
The requirement for each vertex in $R$ and $\tilde{V}$ is equal to its degree.
The requirement for $s^j_{i,j'}$, $d^{i,j}_{j'}$, $b_i$, and $\bar{b}_i$ is one for all possible $i,j,j'$.
The requirement for all other vertices is two. 

Let $\pwinst = (G_{\mu,\nu},\degreq)$ be the unweighted \sdmstp instance obtained from the construction above.
Let $\pwinst' = (G'_{\mu,\nu},\degreq')$ be the unweighted \sdmstp instance obtained from $\pwinst$ by applying \cref{lem:consistent_labels} to replace all label gadgets in $G_{\mu,\nu}$.
Since $\maxreq \geq 3$, the maximum degree requirement in $\degreq$ and $\degreq'$ is $\maxreq$.

\begin{longlemma}
\label{lem:tw_lb_pw}
	The path width of $G'_{\mu,\nu}$ is $n + g(\maxreq,\varepsilon)$ for some computable function~$g$. 
\end{longlemma}
\begin{proof}
	We specify a path decomposition of $G'_{\mu,\nu}$.	
	We start with a bag that contains the root $\rho$, $\rho^{\bar{A}}_1$, $\bar{a}_1$, and $v^{i,0}$ for $i \in [n]$.
	Then we repeatedly introduce $\bar{a}_{i+1}$, $\rho^{\bar{A}}_{i+1}$ and then forget $\rho^{\bar{A}}_i$ and $\bar{a}_{i}$ for $i \in [\mu-1]$.
	After that, we forget $\rho^v_{\mu}$ and $\bar{a}_{\mu}$.
	We then repeatedly introduce and forget $\bar{b}_i$ for $i \in [\nu]$.

	Next, suppose for some $j \in [m-1]_0$, the current bag contains $\rho$, $v^{i,j}$ for all $i\in [n]$, and only if $i \geq 1$, $\rho^U_{j}$, $\rho^Z_{j}$.
	We specify a sequence of adjacent bags until we obtain the bag containing only $\rho$, $\rho^U_{j+1}$, $\rho^Z_{j+1}$, and, $v^{i,j+1}$ for all $i\in [n]$.
	We first introduce $\rho^U_{j+1}$, $\rho^Z_{j+1}$, and all the vertices in the constraint gadget for $C_{j+1}$ (including the vertices that replace the label gadgets).
	By construction and \cref{lem:consistent_labels}, the number of all these vertices is upper bounded by $g'(\maxreq, \varepsilon)$ for some computable function~$g'$.
	Then for $i \in [n]$, we introduce $v^{i,j+1}$, all vertices in the variable gadgets between $v^{i,j}$ and $v^{i,j+1}$, and the dummy vertices incident to these vertices.
	The number of all these vertices is $\bigoh(\maxreq)$.
	After that, we forget $v^{i,j}$ and all the recently added vertices in the variable gadgets (except for $v^{i,j+1}$) and the dummy vertices.
	When we finish with the previous steps for all $i \in [n]$, we forget $\rho^U_{j}$, $\rho^Z_{j}$, and all vertices in the constraint gadget for $C_{j+1}$.
	The current bag now contains only $\rho$, $\rho^U_{j+1}$, $\rho^Z_{j+1}$, and, $v^{i,j+1}$ for all $i\in [n]$, as claimed.
	
	Finally, we introduce $\rho^{A}_1$ and $a_1$ and then repeatedly introduce $\rho^{A}_{i+1}$ and $a_{i+1}$ and then forget $\rho^{A}_i$ and $a_i$ for $i \in [n-\mu-1]$.
	Further, we repeatedly introduce and forget $b_i$ for $i \in [n(\maxreq-1)-\nu]$.
	It is easy to see that we have finished the path decomposition of $G'_{\mu,\nu}$, and all bags are bounded by $n +  g(\maxreq,\varepsilon)$ for some computable function~$g$.
\end{proof} 
\begin{longlemma}
\label{lem:tw_lb_equiv_1}
	If $\phi$ has a satisfying assignment, then we can construct a solution spanning tree~$T$ for $\pwinst$ for some $\mu \in [n]_0, \nu \in [n(\maxreq-1)]_0$ such that $T$ is consistent to all the label gadgets in~$G_{\mu,\nu}$.
\end{longlemma}
\begin{proof}
	Suppose  $(x^*_i)_{i \in [n]}$ is the satisfying assignment of $\phi$.
	For $i \in [n]$, let $(c_i,d_i) \in \{+,-\} \times [\maxreq-1]_0$ be the state $\sigma^{-1}(x^*_i)$.
	We choose $\mu = |\{ i \in [n] \colon c_i = + \}|$ and $\nu = \sum^n_{i=1} d_i$.

	We construct a solution spanning tree $T$ for $\pwinst$ starting from an empty graph.
	We add all edges incident to vertices in $\tilde{V}$ and $R$.

	Choose an arbitrary perfect matching of the complete bipartite graph with $\{\bar{a}_1, \dots, \bar{a}_{\mu}\}$ as one part and $\{v^{i,0} \colon c_i = +\}$ as the other part.
	We add this matching to $T$.
	For $i \in [n]$ and $j \geq 1$, if $c_i=+$, we add to $T$ the edge connecting $v^{i,j}$ with its (only) root-connected $-$-neighbor.
	
	Similarly, we add to $T$ an arbitrary perfect matching of the complete bipartite graph with $\{a_1, \dots, a_{n-\mu}\}$ as one part and $\{v^{i,m} \colon c_i = -\}$ as the other part.
	For $i \in [n]$ and $j < m$, if $c_i = -$, we add to $T$ the edge connecting $v^{i,j}$ with its (only) root-connected $+$-neighbor.

	For each vertex~$v$ in $\bar{B} \coloneqq \{\bar{b}_i \colon i \in [\nu]\}$, we add to $T$ an edge incident to~$v$.
	We choose such an edge in an arbitrary fashion with the only condition that for each $i \in [n]$, $v^{i,0}$ is adjacent to exactly $d_i$ vertices in $\bar{B}$.
	For $i\in [n]$ and $j > 0$, we add to $T$ the edge connecting $v^{i,j}$ to its $j'$-th $-$-neighbor for all $j' \in [d_i]$.

	Similarly, for each vertex~$v$ in $B \coloneqq \{b_i \colon i \in [n(\maxreq-1)-\nu]\}$, we add to $T$ an edge incident to $v$, such that in total, for each $i \in [n]$, $v^{i,m}$ is adjacent to exactly $\maxreq-1-d_i$ vertices in $B$.
	For $i \in [n]$ and $j < m$, we add to $T$ the edge connecting $v^{i,j}$ to its $j'$-th $+$-neighbor  for $j' \in [\maxreq-1] \setminus [d]$.
	
	For $j \in [m]$, suppose $A_1, \dots, A_{p_j}$ are the satisfying assignment of $C_j$, and $A_{p_j^*}$ is the assignment corresponding to $(x^*_i)_{i \in [n]}$.
	Then for each vertex in $P^j_{p_j^*}$, we add the edges labeled $\alpha$ to $T$.
	For the other paths $P^j_t$ for $t \in [p_j] \setminus \{p_j^*\}$ in the constraint gadgets, we add all edges of the paths to $T$.
	For $t \in [p_j^*-1]$, we add $z_t\overleftarrow{\pi}_t$ and $z_t\overleftarrow{\pi}_{t-1}$ (the latter edge is only if $t > 1$).
	For $t \in [p_j] \setminus [p_j^*]$, we add $z_t\overleftarrow{\pi}_{t+1}$ and $z_t\overleftarrow{\pi}_{t}$ (the former edge is only if $t < p_j$).
	
	By construction, it is easy to check that the degree requirement is satisfied.
	It remains to show that $T$ is a spanning tree.
	
	Firstly, all the vertices in $R$ induce a tree in $T$.
	Secondly, in the constraint gadget for a $q$-constraint $C_j$, there is a path that visits all $z^j_i$ and all the paths $P^j_t$ except for the path $P^j_{p^*_j}$.
	Note that all vertices on this path do not have any other incident edges in $T$ outside the path except for $z^j_1$ which is connected to $\rho^Z_j$.
	Thirdly, for each $i$ and $j'$, there is a path from $s^j_{i,j'}$ that visits  $w^{p^*_j}_{i,j'}, u^{i,j}_{j'}, \tilde{u}^{i,j}_{j'}, \bar{u}^{i,j}_{j'}$, where the last three vertices are either visited in that order or in the reversed order.
	If they are visited in this order, the path then continues to $v^{i,j}$; otherwise, it goes to $v^{i,j-1}$.
	Note that these vertices, except for $v^{i,j}$, $v^{i,j-1}$, and $\tilde{u}^{i,j}_{0}$ have no other incident edges in $T$ outside this path.
	Lastly, each vertex $v^{i,j}$ is incident to exactly one of its $+$- and $-$-root-connected neighbor.
	From each root-connected neighbor, there is a unique path to the root~$\rho$ via some $\tilde{u}^{i,j}_0$, $\rho^{\bar{A}}_t$, or $\rho^{A}_t$.
	Lastly, each vertex in $B$ and $\bar{B}$ is connected to exactly one vertex in $\{v^{i,0}, v^{i,m} \colon i \in [n] \}$.

	All the above imply that each vertex in $G_{\mu,\nu}$ has exactly one unique path to the root $\rho$ in $T$. 
	Hence, $T$ is a spanning tree of $G_{\mu,\nu}$, as required.
\end{proof} 
\begin{longlemma}
\label{lem:tw_lb_equiv_2}
	If for some $\mu \in [n]_0$ and $\nu \in [n(\maxreq-1)]_0$, $\pwinst$ has a solution spanning tree~$T$ consistent to all the label gadgets in $G_{\mu,\nu}$, then $\phi$ has a satisfying assignment.
\end{longlemma}
\begin{proof}
	Firstly, all edges incident to vertices in $\tilde{V}$ and $R$ have to be in $T$, due to the degree requirements.
	Secondly, consider the vertices on a path $P^j_t$ in a constraint gadget. 
	Either all incident edges to these vertices in $T$ have label $\beta$ (i.e., the whole path is in $T$) or all of them have label $\alpha$ (i.e., the whole path is not in $T$).
	Since the vertices $s^j_{i,j'}$ have degree requirement one, exactly one path $P^j_{p_j^*}$ is not in $T$ for some $p_j^* \in [p_j]$.
	We call this path the $\alpha$-path of $C_j$.
	
	Thirdly, given an $\alpha$-path $P^j_{p_j^*}$ of $C_j$, looking at the neighbors of $z^j_{p_j^*}$, we see that the edges with label $\beta$ available left for this vertex is exactly equal to its degree requirement.
	Hence, these edges have to be in $T$.
	Iteratively argue for vertices $z^j_t$ for $t < p_j^*$ and $t > p_j^*$, we see that there is exactly one way to meet all the degree requirements of these vertices.
	Specifically, for $t \leq p_j^*$, the edges $z^j_t\overleftarrow{\pi}^j_t$ and $z^j_t\overrightarrow{\pi}^j_{t-1}$ are in $T$ (the latter edge is only for $t > 1$).
	For $t > p_j^*$, the edges $z^j_t\overleftarrow{\pi}^j_{t+1}$ and $z^j_t\overrightarrow{\pi}^j_{t}$ are in $T$ (the former edge is only for $t < p_j$).
	
	Fourthly, once we have fixed the incident edges of the vertices on all the paths $P^j_t$ in $T$, since all edges incident to vertices in $\tilde{V}$ are in $T$, there is also one fixed choice for the incident edges of all the vertices $u^{i,j}_{j'}$, $\bar{u}^{i,j}_{j'}$ for $j \in [m]$, $i \in I_j$, and $j' \in [\maxreq-1]_0$.
	Further, since the vertices on the path $P_{p^*}$ are the only vertices in the constraint gadget for $C_j$ to have edges with label $\alpha$ in $T$, and since the dummy vertices have degree requirement one, it is easy to see that (i) exactly one of $u^{i,j}_0 v^{i,j-1}$ and $\bar{u}^{i,j}_0 v^{i,j}$ is in $T$ and that (ii) for each $j' \in [\maxreq-1]$, exactly one of $u^{i,j}_{j'} v^{i,j-1}$ and $\bar{u}^{i,j}_{j'} v^{i,j}$ is in $T$.
	These imply that the number of $+$-neighbors of $v^{i,j-1}$ and $-$-neighbors of $v^{i,j}$ is exactly $\maxreq$.
	Combining this with (i), we can encode a unique state $(c_i,d_i)$ for each variable gadget, where $d_i$ is the number of $-$ neighbors each of $v^{i,j}$ is incident to in $T$; $c_i = -$ if all vertices $v^{i,j}$ are incident to a root-connected $+$-neighbors in $T$; and $c_i = +$ otherwise.
	Using the mapping $\sigma$, we can translate this into an assignment of the variables of $\phi$.
	The $\alpha$-path $P^j_{p_j^*}$ of each $q$-constraint $C_j$ is then a certificate that $C_j$ is satisfied by this assignment.
	The lemma then follows.
\end{proof}
 
Combining the previous two lemmas with \cref{lem:consistent_labels}, we obtain the following correspondence.
 
\begin{longlemma}
\label{lem:tw_lb_equiv}
	$\phi$ has a satisfying assignment if and only if $\pwinst'$ is a YES instance for some $\mu \in [n]_0$ and $\nu \in [n(\maxreq-1)]_0$.
\end{longlemma}
\begin{proof}
	By Lemmas~\ref{lem:tw_lb_equiv_1} and~\ref{lem:tw_lb_equiv_2}, $\phi$ has a satisfying assignment if and only if for some $\mu \in [n]_0$ and $\nu \in [n(\maxreq-1)]_0$, $\pwinst$ has a solution spanning tree~$T$ consistent to all the label gadgets in $G_{\mu,\nu}$.
	By \cref{lem:consistent_labels}, the latter holds if and only if $\pwinst'$ is a YES-instance.
	The lemma then follows.
\end{proof} 
We are now ready to prove \cref{thm:tw_lb}.

\begin{proof}[Proof of \cref{thm:tw_lb}]
	For the sake of contradiction, suppose there is an algorithm that solves \sdmstp in time $\ostar\Big((2\maxreq-\varepsilon)^{\pw}\Big)$.
	Combined with \cref{lem:tw_lb_pw}, this implies that we can solve $\pwinst'$ in the construction above in time $\ostar\Big((2\maxreq-\varepsilon)^{n + g(\maxreq,\varepsilon)}\Big)$.
	By \cref{lem:tw_lb_equiv}, trying all possible combinations of $\mu \in [n]_0$ and $\nu \in [n(r-1)]_0$, we can then solve any $\phi$ in $\ostar\big(n^2(2\maxreq-\varepsilon)^{n + g(\maxreq,\varepsilon)}\big)$.
	However, this contradicts \cref{thm:csp}.
\end{proof} 

\section{Cutwidth}\label{sec:ctw}
\noindent
\subsection{Upper Bound}
Our next aim is to prove the upper bound for cutwidth:

\ctwub*

Let $\ell = v_1,\dots,v_n$ be the given linear arrangement of $G$ of cutwidth $\ctw$.
For each $i\in[n]$, we define $A_i = E(H_i)$. We define the set $L_i$ as the set of all endpoints of $A_i$ in $V_i$, and the set $R_i$ as the endpoints of $A_i$ in $\overline{V}_i$.
 Let $E_i = E(G[V_i])\cup A_i$ and $N_i = R_i\setminus R_{i-1} = N_{H_i}(v_i)\setminus R_{i-1}$. Finally, we define the set $S_i = \{v\in R_i\colon \deg_{H_i}(v) = 1\}$.

We will define a path decomposition $\nctbags$ from $\ell$, and compute the tables $\hat{T}_x$ (from \cref{sec:pw}) for each node $x$ of $\nctbags$. We then carry out a careful analysis to show that the size of the families $\indices_x$ is bounded by $3^{\ctw}$. The theorem then follows from \cref{pw-ub-lem:running-time-relto-indices}.

By definition, the set $A_i$ for each $i\in[n]$ is a cut of $G$. Hence, the sets $R_i$ form separators of size at most $\ctw$. It is not hard to see that the sequence $\ctbags = B'_1,\dots,B'_n$, with $B'_i = R_i\cup \{v_i\}$ is a path decomposition of $G$. We complete $\ctbags$ into a nice path decomposition $\nctbags$ by replacing each set $B'_i$ by a sequence $B_i^1,\dots,B_i^{n_i}$, of introduce and forget bags, where $B_i^{n_i}=B'_i$.

\begin{longdefinition}\label{ctw-ub-def:ctbags}
    We define the sequence $\ctbags = B'_1, \dots, B'_{n}$ with $B'_i = \{v_i\}\cup \{R_i\}$. We define the path decomposition $\nctbags=B_1,\dots,B_{n'}$ as a super-sequence of $\ctbags$, where we replace each bag $B'_i$ by the following sequence: first, if $v_i\notin R_{i-1}$, then we add the bag $B_{i}^1$ as an introduce vertex bag $(v_i)$.
    After that, for each vertex $v_j \in N_i$ in an arbitrary order, we add an introduce vertex bag $(v_j)$. Then for each vertex $v_j \in N_G(v_i)$ with $j>i$ in an arbitrary order, we add an introduce edge  bag $(\{v_i, v_j\})$. We finish with a forget vertex bag ($v_i$). Let 
    $B_i^1,\dots,B_i^{n_i}$ be this sequence of bags, replacing $B'_i$.
\end{longdefinition}

\begin{longlemma}\label{ctw-ub-lem:linarr-to-pdec}
    It holds that $\nctbags$ is a nice path decomposition of $G$ of width at most $\ctw$.
    Moreover, it holds that $B_i^{n_i-1} = B'_i$ for all $i\in[n]$, and for $B_x$ the bag of $\nctbags$ corresponding to $B_i^{n_i-1}$, it holds that $E_x = E_i$.
\end{longlemma}

\begin{proof}
    Let $i\in[n]$. It holds that $v_i$ is either in $R_{i-1}$, and hence, is in $N_j$ for some $j<i$, which means it was already introduced in $B_j^k$ for some value $k$, or $v_i$ is introduced in $B_i^1$. For each value $i$, $v_i$ is included in all bags from its introduction bag to $B_i^{n_i}$ (exclusively) that forgets this vertex. Hence, the bags containing $v_i$ form a connected subgraph of the decomposition path.
    Finally, for each edge $e=\{v_i,v_j\}\in E$, let $i\leq j$. Then it holds that there is a bag $B_i^k$ introducing $e$ in the sequence replacing $B'_i$. Hence, $\nctbags$ is a nice path decomposition of $G$. Since all bags of the sequence replacing $B'_i$ only contain a subset of $R_i$ and the vertex $v_i$, it holds that the width of the decomposition is at most $R_i + 1 -1\leq \ctw$.

    The second claim follows by induction over $i$, where one can easily show that $B_i^{n_i} = R_i\setminus N_i$.
    Moreover, the edges introduced in the bags $B_i^k$ are exactly the edges $\{v_i, v_j\}$ for $j>i$. Hence, $E_x$ contains exactly the edges $\{v_i, v_j\}$, where $\min\{i,j\}\leq n$, which is exactly $E_i$.
\end{proof} 
\begin{lemma}\label{ctw-ub-lem:indices-size}
    Let $t\in[n']$. It holds that $|\indices_t| \leq 2n\cdot 3^{\ctw}$.
\end{lemma}

\begin{proof}
    We show that the claim hold for each bag $B_t = B_i^{n_i-1} = B'_i$. This suffices, since all bags $B_i^0,\dots,B_i^{n_i-2}$ are introduce bags, where $|\indices_t|$ can only increase, and the bag $B_i^{n_i}$ is a forget bag, where $|\indices_t|$ can only decrease.
    For a set of vertices $S\subseteq B^i_j$, we define 
    \[
        \indices_t[S] = \{(f_{|S},c_{|S})\colon (f,c) \in \indices_t\}
    \]
    Clearly, it holds that
    \[
    |\indices_t| \leq |\indices_t[S_i]|\cdot |\indices_t[B'_i\setminus S_i]|,
    \]
    since one can easily define an injective mapping from $\indices_t$ to $\indices_t[S_i]\times \indices_t[B'_i\setminus S_i]$, by $(f,c) \mapsto \big((f_{|S_i}, c_{|S_i}), (f_{|B'_i\setminus S_i}, c_{|B'_i\setminus S_i})\big)$. The mapping is injective, since $f$ is the disjoint union of $f_{|S_i}$ and $f_{|B'_i\setminus S_i}$, and similarly for $c$. We bound $|\indices_t[S_i]|$ and $|\indices_t[B'_i\setminus S_i]|$ independently.

    First, we start with $\indices_t[S_i]$. By definition of $\indices_t$, it holds for $v\in S_i$ and $(f,c)\in \indices_t$ that $f(v)\in[1]_0$, and that $c(v) = \noc$ if $f(v)=0$. It follows that
    \begin{align*}
        |\indices_t[S_i]| &\leq \sum\limits_{S\subseteq S_i} 1^{|S|}2^{|S_i|-|S|}\\
        &= \sum\limits_{k=0}^{|S_i|} \binom{|S_i|}{k}2^{|S_i|-k} = 3^{|S_i|}.
    \end{align*}

    On the other hand, it holds that 
    \begin{align*}
        |\indices_t[B'_i\setminus (S_i\cup v_i)]|
        &\leq \prod\limits_{v\in B'_i\setminus (S_i\cup \{v_i\})} (2(\deg_H(v)) + 1)\\
        &\leq \prod\limits_{v\in B'_i\setminus (S_i\cup \{v_i\})} (2(deg_H(v)) + 2)\\
        &\leq 2^{B'_i\setminus (S_i\cup \{v_i\})} \cdot \prod\limits_{v\in B'_i\setminus (S_i\cup \{v_i\})} (\deg_H(v)+1)\\
        &\leq 2^{\frac{\ctw-|S_i|}{2}} \cdot 2^{\ctw-|S_i|} \leq (2^{3/2})^{\ctw - |S_i|} \leq 3^{\ctw-|S_i|}.
    \end{align*}
    Note that the first equality follows from the fact, that each vertex $v \in R_i$ has degree at most $\deg_{H_i}(v)$ in $G_i$. The fourth inequality follows from the AM-GM inequality (see \cite{DBLP:journals/tcs/JansenN19}), and form the fact, that each vertex in $M_i$ is incident to at least two distinct edges of $H_i$.
    It follows that $|\indices_t[B'_i\setminus S_i]|\leq 3^{\ctw-|S_i|}\cdot 2n$, and $|\indices_t| \leq 3^{\ctw}\cdot 2n$.
\end{proof}

\begin{proof}[Proof of \cref{thm:ctw-ub}]
    We fix $Z = 2m = n^{\oh(1)}$.
    The algorithm starts by building the decomposition $\nctbags$ as defined in Definition~\ref{ctw-ub-def:ctbags}. By Lemma~\ref{ctw-ub-lem:linarr-to-pdec}, it holds that $\nctbags$ is a nice path decomposition of $G$ of width $\ctw$.
    The algorithm computes all tables $\hat{T}_x$ for all nodes $x$ in time $\ostar(3^{\ctw}\cdot W^{\oh(1)})$ by \cref{pw-ub-lem:running-time-relto-indices} and \cref{ctw-ub-lem:indices-size}.
    The algorithm accepts, if there exists a value $\omega_2\in\overline Z$ such that $\hat{T}_{\droot}[n-1,\omega_1,\omega_2][\phi] = 2$ and rejects otherwise.

    It holds by \cref{tw-ub-cor:isolation} that with probability at least $1/2$ there exists a value $\omega_2\in\overline Z$ such that $\mathcal{S}[\omega_1,\omega_2]$ contains a single solution only, if a solution exists, and none otherwise. It holds by \cref{pw-ub-cor:tables-count-cuts} that $\hat{T}_{\droot}[n-1,\omega_1,\omega_2][\phi] \equiv_4 C[\omega_1,\omega_2]$, and by \cref{tw-ub-lem:cnc} that $C[\omega_1,\omega_2] \equiv_4 2$ if $\big|\mathcal{S}[\omega_1,\omega_2]\big| = 1$. Hence, if a solution exists the algorithm accepts with probability at least $1/2$, and rejects otherwise.
\end{proof} 
\subsection{Lower Bound}
\begin{figure}[t!]
    \centering
    \begin{subfigure}[t]{\textwidth}
        \centering
        \begin{tikzpicture} [x=1cm, y=.8cm]
            \pic{ctw-lb-vg};
        \end{tikzpicture}
        \begin{tikzpicture}
		\pic{ctw-lb-absorber};
	\end{tikzpicture}
        \caption{Variable gadget (left) and degree absorbers (right)}
    \end{subfigure} 
    \begin{subfigure}[t]{\textwidth}
        \centering
        \begin{tikzpicture} [x=1cm, y=.8cm]
            \pic{ctw-lb-state-1};
        \end{tikzpicture} \\
                \begin{tikzpicture} [x=1cm, y=.8cm]
            \pic{ctw-lb-state-2};
        \end{tikzpicture} \\
                \begin{tikzpicture} [x=1cm, y=.8cm]
            \pic{ctw-lb-state-3};
        \end{tikzpicture}
        \caption{Three possible states of the variable gadget and the corresponding degree absorbers}
    \end{subfigure}

    \caption{Gadgets for the proof of \cref{thm:ctw_lb}}
    \label{fig:ctw_lb}
\end{figure}

In this section, we obtain a lower bound complementing the previous cutwidth-based algorithm:
      \ctwlb*

This proof has the general template as the proof of \cref{thm:tw_lb}.
In particular, we also reduce from a suitable $q$-CSP-$\cspdom$ instance.

Given $\varepsilon > 0$, let $q$ be the smallest integer such that $n$-variable $q$-CSP-$3$ does not admit an  $\ostar((3-\delta)^n)$ by \cref{thm:csp} for some appropriate $\delta$ that we will choose later. 
Let $\phi$ be a $q$-CSP-$3$ instance with $n$ variables $x_1, \dots, x_n$ and $m$ $q$-constraints $C_1, \dots, C_m$.

\para{Variable gadget.} For $i \in [n]$, we construct a variable gadget for $x_i$ as follows; see \cref{fig:ctw_lb}(a) for an illustration.
We first create $m+1$ vertices $v^{i,0}, \dots, v^{i,m}$.
For $j \in [m]$, we add a 5-cycle $v^{i,j} u^{i,j}_1 u^{i,j}_2 u^{i,j}_3 u^{i,j}_4$, and the edge $v^{i,j-1} u^{i,j}_3$.

\para{Constraint gadget.} 
The idea for the clause gadget is the same as the clause gadget in the proof of \cref{thm:tw_lb}, except that instead of having $r$ vertices to encode a state, we only have 1 vertex.
In particular, we use the same notations of $C_j$, $A_1, \dots, A_{p_j}$, $I_j$, and $a_{i,t}$ for $i \in I_j$ and $t\in [p_j]$.
Then we create $q$ new vertices $w^{t,j}_{i}$ for $i \in I_j$, and we connect them with a path such that the subscripts appear in the lexicographical order.
We call this path $P^j_t$.
For $i \in I_j$, we connect $w^{t,j}_{i}$ with $u^{i,j}_{a_{i,t}}$.

Next, for $i \in I_j$, we create a new vertex $s^j_{i}$ and connect it with the vertices $w^{t,j}_{i}$ for all $t \in [p_j]$.
We add $p_j$ new vertices $z^j_1, \dots, z^j_{p_j}$ and connect to the paths $P^j_t$ and assign labels to the edges incident to the vertices $w^{t,j}_i$ in a similar fashion as in the clause gadget for \cref{thm:tw_lb}.

\para{Dummy vertices.} For $j \in [m]$, $i \in [n] \setminus I_j$, we create a dummy vertex $d^{i,j}$ and connect it with $u^{i,j}_{1}$, $u^{i,j}_{2}$, and $u^{i,j}_{3}$.
In words, the three vertices $u^{i,j}_{1}$, $u^{i,j}_{2}$, and $u^{i,j}_{3}$ connect with the constraint gadget for $C_j$ if $x_i$ is a variable of $C_j$, and otherwise, they connect with the dummy vertex $d^{i,j}$.

\para{Degree absorbers.}
See \cref{fig:ctw_lb}(a) for an illustration of the description below.
Let $\gamma$ be a constant that we will choose later.
We split the set $\{v^{i,0} \colon i \in [n]\}$ into $\gamma$ disjoint subsets $S_1$, \dots, $S_{\gamma}$ of equal size.
For each subset $i \in [n/\gamma]$, we create a vertex $a_i$ and $n/\gamma$ vertices $b^i_1, \dots, b^i_{n/\gamma}$.
Let $B_i$ be the set of the latter $n/\gamma$ vertices.
We then add an edge between $a_i$ and each vertex in $S_i$ and $B_i$.
Then we add a matching between the set $S_i$ and $B_i$ (that is, we add $n/\gamma$ edges, each of which is incident to exactly one vertex in $S_i$ and one vertex in $B_i$, and each vertex in the two sets is incident to exactly one edge in the matching).

Similarly, we also split the set $\{v^{i,m} \colon i \in [n]\}$ into $\gamma$ disjoint subsets $\bar{S}_1$, \dots, $\bar{S}_{\gamma}$ of equal size.
For each subset $i \in [n/\gamma]$, we create a vertex $\bar{a}_i$ and $n/\gamma$ vertices $\bar{b}^i_1, \dots, \bar{b}^i_{n/\gamma}$.
Let $\bar{B}_i$ be the set of the latter $n/\gamma$ vertices.
We then add an edge between $\bar{a}_i$ and each vertex in $S_i$ and $\bar{B}_i$.
Then we add a matching between the set $S_i$ and $\bar{B}_i$.

\para{Root.} We have a special vertex called the \defi{root} $\rho$, which ensures that the connectivity of all gadgets in the construction.
We add four paths $(\rho, \rho^U_1, \dots, \rho^U_m)$, $(\rho, \rho^A_1, \dots, \rho^A_{n/\gamma})$, $(\rho, \rho^{\bar{A}}_1, \dots, \rho^{\bar{A}}_{n/\gamma})$ and $(\rho, \rho^Z_1, \dots, \rho^Z_m)$.

For $j \in [m]$, we add a path $(\rho^U_j, \rho^U_{j,1}, \dots, \rho^U_{j,n})$.
For $i \in [n]$, we add the edge $\rho^U_{j,i} u^{i,j}_4$ for $j \in [m]$.   For $j \in [m]$, we add the edge $\rho^Z_j z^j_0$.
For $t \in [n/\gamma]$, we add the edge $\rho^A_t a_{t}$ and $\rho^{\bar{A}}_t \bar{a}_{t}$.

Let $R$ be the set of all vertices $\rho$, $\rho^U_j$, $\rho^U_{j,i}$, $\rho^A_i$, and $\rho^{\bar{A}}_i$ for all possible $i$ and $j$.

\para{Degree requirements.} Lastly, we specify the degree requirements.
Here, each \sdmstp instance is identified with a sequence of parameters $M = (\mu_1, \dots, \mu_{\gamma}, \bar{\mu}_1, \dots, \bar{\mu}_{\gamma})$ such that each parameter is in the range $[2n/\gamma]_0$.
Then for $t \in [\gamma]$, the requirement for $a_t$ is $\mu_t$ while the requirement for $\bar{a}_t$ is $\bar{\mu}_t$.
The requirement for each vertex in $R$ is equal to its degree.
The requirement for $s^j_i, b^i_t, \bar{b}^i_t$ is one for all possible $i,j,t$.
The requirement for all other vertices is two. 

Let $\ctwinst = (G_M, \degreq)$ be unweighted \sdmstp instance obtained from the construction above.
Let $\ctwinst'$ be the unweighted \sdmstp instance obtained from $\ctwinst$ by applying \cref{lem:consistent_labels} to replace all label gadgets in $G_{M}$.

\begin{longlemma}
\label{lem:ctw_lb_pw}
	The cutwidth of $G'_M$ is at most $\frac{\gamma+1}{\gamma}n + g(\delta)$ for some computable function~$g$. 
\end{longlemma}
\begin{proof}
	We specify a linear arrangement.
	We start with the root $\rho$.
	For $t \in [\gamma]$, we add $\rho^A_t, a_t, b^t_1, \dots, b^t_{n/\gamma}$ and vertices in $S_t$ in some arbitrary order.
	For $j \in [m]$, we add $\rho^U_j, \rho^Z_j$, and then $v^{i,j}_1, u^{i,j}_1, u^{i,j}_2, u^{i,j}_3, u^{i,j}_4$, and $d^{i,j}$ (if exists) for $i \in [n]$, and then all vertices in the constraint gadget for $C_j$.
	For the analysis later, we denote by $L_j$ the set of all vertices listed in the previous sentence.
	Finally, for $t \in [\gamma]$, we add $\bar{b}^t_1, \dots, \bar{b}^t_{n/\gamma}, \bar{a}_t, \rho^{\bar{A}}_t$.
	
	For every vertex $v$, we denote by $\overleftarrow{V}(v)$ the set of all vertices up to and including $v$ in the linear arrangement.
	We denote by $\overrightarrow{V}(v)$ the set of all vertices after $v$ in the linear arrangement.
	Let $H(v) \coloneqq G'_M[\overleftarrow{V}(v), \overrightarrow{V}(v)]$.
	It remains to show that for any vertex $v$, 
	\begin{equation} 
	\label{eq:ctw_lb}
		E(H(v)) \leq \frac{\gamma+1}{\gamma}n + g(\delta) 
	\end{equation}
	for some computable function~$g$.
	
	This is true for $v = \rho$, since $\rho$ has degree four.
	Fix $t \in [\gamma]$.
	We show that (\ref{eq:ctw_lb}) holds for $v \in \{\rho^A_t, a_t, b^t_1, \dots, b^t_{n/\gamma}\}$.
	Note that the edges in $E(H(\rho^A_t))$ that are not incident to $\rho^A_t$ or $\rho$ have the form $v^{i',0} u^{i',1}_3$ where $v^{i',0}$ is a vertex in $\bigcup^{t-1}_{i=1} S_i$.
	Hence, there are $(t-1)n/\gamma$ such edges.
	Since there are three edges incident to $\rho^A_t$, $2n/\gamma+1$ edges incident to $a_t$, and since $t \leq \gamma$, it follows that (\ref{eq:ctw_lb}) holds for $v \in \{\rho^A_t, a_t\}$.
	For $i \in [n/\gamma]$, observe that in order to obtain $E(H(b^t_i))$ from $E(H(b^t_{i-1}))$ (or from $E(H(a_t))$ if $i=1$), we remove the edge $a_t b^t_i$ and add the edge $b^t_i v$ for some $v \in S_t$.
	Hence, $|E(H(b^t_i))| = |E(H(a_t))|$, and hence (\ref{eq:ctw_lb}) also holds for $v = b^t_i$.
	Next, for $v \in S_t$, in order to obtain $E(H(v))$ from the edge set of the immediately previous cut graph, we remove two edges $b^t_j v$ and $a_t v$ for some $j \in [n/\gamma]$ and add an edge $vu^{i',1}_3$ for some $i'$.
	Hence, $|E(H(v))| < |E(H(a_t))|$, and therefore (\ref{eq:ctw_lb}) holds for $v$.
	
	Now suppose that (\ref{eq:ctw_lb}) holds for the cut graph immediately preceding $H(\rho^U_j)$ for some $j \in [m]$.
	Note that each of $\rho^U_j$ and $\rho^Z_j$ has one adjacent vertex preceding itself in the linear arrangement, one adjacent vertex in $L_j$, and one remaining adjacent vertex.
	Hence, it is easy to see that (\ref{eq:ctw_lb}) also holds for $v \in \{\rho^U_j,\rho^Z_j \}$.
	For other vertices in $L_j$, their only neighbors outside $L_j$ are $v^{i,j-1}$ (which is the neighbor of $u^{i,j}_3$ and precedes all vertices in $L_j$) and $u^{i,j+1}_3$ (which is the neighbor of $v^{i,j}$).
	Furthermore, $|L_j|$ is a function of $q$, which in turn is a function of $\delta$.
	It then follows that (\ref{eq:ctw_lb}) holds for $v \in L_j$.
	
	Finally, for the remaining vertices, we use a similar argument as for the vertices in $\rho^A_t, a_t, b^t_1, \dots, b^t_{n/\gamma}$.
\end{proof} 
\begin{longlemma}
\label{lem:ctw_lb_equiv_1}
	If $\phi$ has a satisfying assignment, then we can construct a solution spanning tree~$T$ for $\ctwinst$ for some $M \in [2n/\gamma]^{2\gamma}_0$ such that $T$ is consistent to all the label gadgets in~$G_M$.
\end{longlemma}
\begin{proof}
	Suppose  $(x^*_i)_{i \in [n]}$ is the satisfying assignment of $\phi$.
	We construct $T$ starting from an empty graph.
	We first add all edges incident to vertices in $R$.

	For $i \in [n]$, if $x^*_i  = 1$, then for $j \in [m]$, we add $v^{i,j} u^{i,j}_4, u^{i,j}_1 u^{i,j}_2, u^{i,j}_2 u^{i,j}_3, u^{i,j}_3 v^{i,j-1}$ to $T$ (see \cref{fig:ctw_lb}(b)(top) for an illustration.).
	If $x^*_i  = 2$, then for $j \in [m]$, we add $v^{i,j} u^{i,j}_1, u^{i,j}_1 u^{i,j}_2, u^{i,j}_3 u^{i,j}_4, u^{i,j}_3 v^{i,j-1}$ to $T$ (see \cref{fig:ctw_lb}(b)(middle) ).
	If $x^*_i  = 3$, then for $j \in [m]$, we add $v^{i,j} u^{i,j}_4, v^{i,j} u^{i,j}_1, u^{i,j}_1 u^{i,j}_2, u^{i,j}_2 u^{i,j}_3$ to $T$ (see \cref{fig:ctw_lb}(b)(bottom)).
	For $i \in [n], j \in [m]$, if $d^{i,j}$ exists, then we add to $T$ the edge $d^{i,j} u^{i,j}_{x^*_i}$.
	
	For $j \in [m]$, suppose $A_1, \dots, A_{p_j}$ are the satisfying assignment of $C_j$, and $A_{p^*_j}$ is the assignment corresponding to $(x^*_i)_{i \in [n]}$.
	Then for each vertex in $P_{p^*_j}$, we add the edges labeled $\alpha$ to $T$.
	For other paths $P_t$ for $t \in [p_j] \setminus \{p^*_j\}$ in the constraint gadgets, we add all edges of the paths to $T$.
	For $t \in [p^*_j-1]$, we add $z_t\overleftarrow{\pi}_t$ and $z_t\overleftarrow{\pi}_{t-1}$ (the latter edge is only if $t > 1$).
	For $t \in [p_j] \setminus [p^*_j]$, we add $z_t\overleftarrow{\pi}_{t+1}$ and $z_t\overleftarrow{\pi}_{t}$ (the former edge is only if $t < p_j$).
	
	The remaining vertices are those in the degree absorbers (see \cref{fig:ctw_lb}(b) for an illustration).
	For $i \in [n]$, suppose $v^{i,0}$ is adjacent to $a_t$ and $b^t_j$ in $G$, and $v^{i,m}$ is adjacent to $\bar{a}_{\bar{t}}$ and $\bar{b}^{\bar{t}}_{\bar{j}}$, for some $t, j, \bar{t}, \bar{j}$.
	If $x^*_i = 1$, then we add $a_t v^{i,0}, a_t b^t_j, v^{i,m} \bar{b}^{\bar{t}}_{\bar{j}}$ to $T$.
	If $x^*_i = 2$, then we add $v^{i,0} b^t_j, v^{i,m} \bar{a}_{\bar{t}}, \bar{b}^{\bar{t}}_{\bar{j}} \bar{a}_{\bar{t}}$ to $T$.
	If $x^*_i = 3$, then we add $a_t v^{i,0}, b^t_j v^{i,0}, \bar{b}^{\bar{t}}_{\bar{j}} \bar{a}_{\bar{t}}$ to $T$.
	
	The degrees of $a_1, \dots, a_{\gamma}, \bar{a}_1, \dots, \bar{a}_{\gamma}$ in $T$ are trivially in the range $[2n/\gamma]_0$.
	Hence the degree requirements for these vertices are satisfied for some appropriate choice of $M \in [2n/\gamma]^{2\gamma}_0$.
	Further, by construction, it is easy to check that the degree requirements are always satisfied for all other vertices.
	
	Hence, it remains to show that $T$ is a spanning tree.
	Observe that the vertices in $R$ induce a tree in $T$.
	Hence, it is sufficient to prove that any other vertex~$v$ in $G$, the following property holds: (*) $v$ has a unique path to a vertex in $R$ and hence a unique path to the root $\rho$.	
	
	Firstly, by construction, for $i \in [n], j \in [m-1]$, it is easy to see that this unique path for $v^{i,j}$ is either (i) $(v^{i,j}, u^{i,j}_4, \rho^U_{j,i})$ or (ii) $(v^{i,j}, u^{i,j+1}_3, u^{i,j+1}_4, \rho^U_{j+1,i})$.
	This is because the other paths from $v^{i,j}$ only visit some other vertices in the variable gadgets, before either visiting the dummy vertex $d^{i,j}$ or $d^{i,j+1}$ (which has degree one in $T$) or visiting some $w^{p^*_{j'}, j'}_{i}$ and then $s^{j'}_{i}$ (which also has degree one).
	This also implies that property (*) also holds for all vertices $u^{i,j}_t, d^{i,j}, s^j_i$ for $i \in [n], j \in [m], t \in [4]$.
	Further it also holds for all vertices of the form $w^{p^*_j, j}_{i}$, where we recall that $p^*_j$ is the index of the satisfying assignment for constraint $C_j$.
	
	Secondly, for $j = 0$, it is also easy to see that the unique path for $v^{i,j}$ is either the path (ii) above or the path $(v^{i,j}, a_t, \rho^A_{t})$ for some $t$.
	For $j = m$, this unique path is either the path (i) or $(v^{i,j}, \bar{a}_{t'}, \rho^{\bar{A}}_{t'})$ by some $t'$.
	Note that each $a_t$ or $\bar{a}_t$ is connected to $R$ only via $\rho^A_t$ or $\rho^A_{\bar{t}}$, because its other potential neighbors in $T$ are of the form $b^t_j, \bar{b}^t_j$ (which have degree one in $T$) or of the form $v^{i,0}, v^{i,m}$ (but in these cases, the unique paths connecting these vertices to $R$ are via $a_t$ and $\bar{a}_t$ themselves.) 
	For a vertex of the form $b^t_j$, either it is connected to $a_t$ which is in turn adjacent to $\rho^A_t$, or it is connected to some $v^{i,0}$ which also has a unique path to a vertex in $R$.
	An analogous argument holds for vertices of the form $\bar{b}^t_j$.
	
	Thirdly, in the constraint gadget for a $q$-constraint $C_j$, there is a path that visits all $z^j_i$ and all the paths $P_t$ except for the path $P_{p^*_j}$.
	Note that all vertices in this path do not have any other incident edges in $T$ outside the path except for $z^j_1$ which is connected to $\rho^Z_j$.
	
	All the above imply that property (*) holds for all vertices in $G$. 
	Hence, $T$ is a spanning tree of $G$, as required.
\end{proof} 
\begin{longlemma}
\label{lem:ctw_lb_equiv_2}
	If for some $M \in [2n/\gamma]^{2\gamma}_0$, $\ctwinst$ has a solution spanning tree~$T$ consistent to all the label gadgets in $G_M$, then $\phi$ has a satisfying assignment.
\end{longlemma}
\begin{proof}
	Firstly, we start with a few similar arguments as in the proof of \cref{lem:tw_lb_equiv_2}.
	In particular, all edges incident to vertices in $R$ have to be in $T$, due to the degree requirements.
	Next, for $j \in [m]$, and $i \in [n]$, there is one $\alpha$-path of the form $(s^j_i, w^{p^*_{j}, j}_i, u^{i,j}_k)$ for some $p^*_j \in [p_j]$ and $k \in [3]$.
	Further, there is a path in $T$ that visits all $z^j_t$ for $t \in [p_j]$ and the vertices on the paths $P^j_t$ $t \in [p_j] \ \{p^*_j\}$.
	Except for $z^j_1$, no vertex can have other incident edges outside the path, and the other incident edge of $z^j_1$ has to be $z^j_1 \rho^Z_j$. 
	
	Secondly, for $i \in [n], j \in [m]$, we argue that the three states as depicted in \cref{fig:ctw_lb}(b)  are the only possible configurations we can have the vertices in the variable gadget between $v^{i,j-1}$ and $v^{i,j}$.
	Note that if one of $u^{i,j}_1, u^{i,j}_2, u^{i,j}_3$ has an incident edge outside of the cycle $\mathcal{C}^{i,j} \coloneqq v^{i,j}u^{i,j}_1u^{i,j}_2u^{i,j}_3u^{i,j}_4$, then this edge is either incident to the dummy vertex $d^{i,j}$ (which has degree one) or has to be part of the $\alpha$-path that visits $s^j_i$.
	Hence, there is exactly one of these vertices with an incident edge outside of the cycle~$\mathcal{C}^{i,j}$.
	Next, since the edge $u^{i,j}_4 \rho^U_{j,i}$ is incident to a vertex in $R$ and has to be in $T$, $u^{i,j}_4$ can only have either $v^{i,j}$ or $u^{i,j}_2$ as its other neighbor in $T$.
	Suppose this is $v^{i,j}$.
	If the edge $v^{i,j} u^{i,j}_1$ is not in $T$, then $u^{i,j}_1, u^{i,j}_2, u^{i,j}_3$ can only connect to the root $\rho$ via $v^{i,j-1}$; we must have \cref{fig:ctw_lb}(b)(top).
	If $v^{i,j} u^{i,j}_1$ is in $T$, then because exactly one of $u^{i,j}_1, u^{i,j}_2, u^{i,j}_3$ can have an incident edge outside of~$\mathcal{C}^{i,j}$, we must have \cref{fig:ctw_lb}(b)(bottom).
	Now suppose that $ u^{i,j}_4 u^{i,j}_3$ is in $T$.
	If the edge $u^{i,j}_3 v^{i,j-1}$ is in $T$, then $u^{i,j}_1, u^{i,j}_2$ can only connect to the root $\rho$ via $v^{i,j}$; we then have \cref{fig:ctw_lb}(b)(middle).
	Otherwise, the edge $u^{i,j}_3 u^{i,j}_2$ is in $T$. 
	However, on the one hand, if $u^{i,j}_2$ then has an incident edge outside of~$\mathcal{C}^{i,j}$, then $u^{i,j}_1$ cannot have degree two in $T$ without having also an incident edge outside of~$\mathcal{C}^{i,j}$, a contradiction.
	On the other hand, if $u^{i,j}_2 u^{i,j}_1$ is an edge in $T$, then $u^{i,j}_1$ has to have an incident edge outside of~$\mathcal{C}^{i,j}$, and $v^{i,j}$ does not have degree two in $T$ and hence does not meet the degree requirement.
	
	Finally, it is easy to see that for $i \in [n]$, there exists exactly one $k \in [3]$ such that for all $j \in [m]$, $u^{i,j}_k$ has an incident edge outside of the cycle $\mathcal{C}^{i,j}$.
	We can interpret this $k$ as the value assignment of the variable $x_i$.
	The $\alpha$-path $P^j_{p^*_j}$ of each $q$-constraint $C_j$ is then a certificate that $C_j$ is satisfied by this assignment.
	The lemma then follows.
\end{proof} 
Combining the previous two lemmas with \cref{lem:consistent_labels}, we obtain the following correspondence.

\begin{longlemma}
\label{lem:ctw_lb_equiv}
	$\phi$ has a satisfying assignment if and only if $\ctwinst'$ is a YES instance for some $M \in [2n/\gamma]^{2\gamma}_0$.
\end{longlemma}
\begin{proof}
	By Lemmas~\ref{lem:ctw_lb_equiv_1} and~\ref{lem:ctw_lb_equiv_2}, $\phi$ has a satisfying assignment if and only if for some $M \in [2n/\gamma]^{2\gamma}_0$, $\ctwinst$ has a solution spanning tree~$T$ consistent to all the label gadgets in $G_{M}$.
	By \cref{lem:consistent_labels}, the latter holds if and only if $\ctwinst'$ is a YES-instance.
	The lemma then follows.
\end{proof} 
We are now ready to prove \cref{thm:ctw_lb}.

\begin{proof}[Proof of \cref{thm:ctw_lb}]
	For the sake of contradiction, suppose there is an algorithm that solves \sdmstp in time $\ostar((3-\varepsilon)^{\ctw})$.
	Combined with \cref{lem:ctw_lb_pw}, this implies that we can solve $\ctwinst'$ in the construction above in time $\ostar((3-\varepsilon)^{\frac{\gamma+1}{\gamma}n + g(\delta)})$.
	By \cref{lem:ctw_lb_equiv}, we can then solve any $\phi$ in time $(2n/\gamma)^{2\gamma}\ostar((3-\varepsilon)^{\frac{\gamma+1}{\gamma}n + g(\delta)})$.
	We choose constants $\gamma, \delta$ such that 
	\[
	(2n/\gamma)^{2\gamma}\ostar((3-\varepsilon)^{\frac{\gamma+1}{\gamma}n + g(\delta)}) \leq \ostar((3-\delta)^{n}).
	\]
	However, this contradicts \cref{thm:csp}.
\end{proof} 

\section{Exact Leaf Spanning Tree}\label{sec:elst}

In this section, we will show that our shifting technique can also be applied to solve the \textsc{Exact Leaf Spanning Tree} problem in single exponential XP running time.

\problembox{\textsc{Exact Leaf Spanning Tree}}{Graph $G=(V,E)$, integer $\eltarget$}{Does $G$ admit a spanning tree $F$ with exactly $\eltarget$ leaves?}

In this problem, we are given a graph $G=(V,E)$ and a positive integer $t$. The goal is to decide whether there exists a spanning tree $F$ of $G$ with exactly $t$ leaves, where a leaf is a vertex of degree exactly $1$ in $F$. 

We assume that $|V|\geq 2$, and hence, $t\geq 2$ since a spanning tree of $G$ cannot consist of a single vertex.

For the remainder of this section,
let $(G, \eltarget)$ be the given instance of \textsc{Exact Leaf Spanning Tree},
where $G$ is an $n$-vertex graph of NLC-width $k$, and let us consider a given NLC-decomposition 
$\phi$ of $G$ using $k$ labels.

\subsection{Partial Solutions and Patterns}

For this problem, we define fixed forests simply as labeled forests. For technical reasons, it is convenient to introduce a label 0, which represents ``virtual'' vertices, i.e. vertices that cannot be used for connecting different components. This label will not appear in partial solutions, but will appear in canonical fixed forests.
For convenience, we call a vertex with label in $[k]$ a \defi{\pos} vertex.

 \begin{definition}
A \defi{fixed forest} is a forest $F$ together with an (extended) labeling function $\lab: V(F) \to [k]_0$.
For a labelled graph $H$, we define a \defi{partial solution} as a fixed forest $F$ that is a spanning subgraph of $H$.
\end{definition}

Although we allow the label $0$ in fixed forests, all considered partial solution, by definition, do not use this label. We allow this label in order to realize representations of patterns created by the shifting process.

We represent a fixed forest by a pattern which consists of a vector and a multiset. The vector represents the number of isolated vertices of each label, and the elements of the multiset correspond to connected components of the fixed forest. For each component, we store a vector representing the number of leaves of each label, and a set of internal vertices, excluding those of label zero (since they cannot be used to create connections to other components).  

Let us now formally define patterns.
For a labeled graph $F$, we denote by $\leaf(F)$ the set of degree-one vertices in $F$. Further, we denote by $\leaf_i(F)$ the set of degree-one vertices in $F$ with label $i$.

 \begin{definition}
    We define a \defi{$k$-pattern} as a pair $(\iso, \mset)$ where $\iso\in\mathbb{N}_0^k$ and $\mset$ is a multiset of pairs $(u,S)$ with $u \in \mathbb{N}_0^k$, and $S \subseteq [k]$ (we omit $k$ when clear from context). We denote by $\mathcal{P}$ the set of all patterns.

    For a fixed forest $F$, its \defi{corresponding pattern} $\repr(F)$ is the pair $(\iso, \mset)$ such that
    \[
    \iso_i = |\{v \in \lab^{-1}(i)\colon \deg_F(v) = 0\}|, \text{ for } i\in [k]
    \]
    and $\mset$ contains exactly the pairs $p(C)$ for each connected component $C$ of $F$, where $p(C) = (u,S)$ such that for $i\in[k]$,
    \begin{align*}
    u_i &= |\leaf_i(C)|, \\
    S &= \lab(C \setminus \leaf(C)) \setminus \{0\}.
    \end{align*} 
\end{definition}

Intuitively, $\iso$ counts the number of isolated vertices of each label. We have to distinguish isolated vertices from leaves because if we connect an isolated vertex to another component, it becomes a leaf, whereas if we connect a leaf to a component, it will no longer be a leaf. For each connected component, $u$ counts the number of leaves of each label, and $S$ is the set of non-zero labels among the internal vertices. Note that $S$ is not a multiset, i.e. we do not count the multiplicity of internal vertices. 

The definition above maps each fixed forest to a pattern. Next, we define a mapping in the other direction, i.e., a mapping that assigns a canonical fixed forest to each pattern. 

 \begin{definition}\label{def:elst:canonical_forest}
Given a pattern $A = (\iso,\mset) \in \mathcal{P}$, we define the \defi{canonical (fixed) forest}, 
denoted $F_A$, as the fixed forest constructed from an empty graph as follows.
For every $i \in [k]$, we add $\iso_i$ isolated vertices with label $i$ to $F_A$.
Next, for each pair $(u, S)\in \mset$, we construct a tree that contains a path with endpoints of label zero, and a vertex of label $i$ for each element $i$ of $S$. 
Then we attach to one endpoint of the tree $u_i$ vertices of label $i$ for each $i\in [k]$ and a single vertex of label 0.
Finally, we add the resulting tree to $F_A$.
 \end{definition}

In our dynamic programming table, at each node of the $k$-expression, the records that we store will contain the corresponding patterns of all partial solutions.

\begin{definition}\label{def:elst:record}
Given a labeled graph $H$ we define the \defi{record} $\mathcal{R}_{H}$ of $H$ as the set of the patterns corresponding to all partial solutions of $H$. 
For $x\in \nodes$, we define $\mathcal{R}_x\eqcolon\mathcal{R}_{G_x}$. 
\end{definition}

\begin{definition}
    \label{def:elst:sol-pat} We call a pattern $\elpat = (\iso, \mset)$ a \defi{solution pattern}, if it holds that $\iso = \mathbf{0}$, $|\mset|=1$, and for the unique element $(u,S)\in \mset$, it holds that $\sum_{i\in[k]} u_i = \eltarget$.
\end{definition}

\begin{observation}
\label{obs:elst:sol-pat}
The instance $(G,\eltarget)$ is a YES-instance if and only if the record $\mathcal{R}_G$ contains a solution pattern.
\end{observation}

\para{Compatibility and Equivalence.}
  
Analogously to Section~\ref{sec:cw}, we define the notions of compatibility and equivalence. Informally, two fixed forests are compatible if they can be completed to a solution, i.e. to a spanning tree with a given number of leaves. While this suffices for correctness of the algorithm and our reduction rules, we face challenge when joining partial solution in the algorithm. In order to overcome this challenge, we define a more fine-grained notion of compatibility that allows to add edges one by one, reducing the resulting family after each addition.

From now on, let $\allpairs = [k]\times[k]$, and let $\mathcal{A} = (\mathbb{N}_0)^{\allpairs}$ be the set of all mappings that assign to each pair $(i,j)\in [k]\times[k]$ a non-negative integer.
 \begin{definition}    
\label{def:elst:cw_compatible_forests}
For two fixed forests $F$ and $F'$, a mapping $\sigma\in \mathcal{A}$ and an integer $\ell\in\mathbb{N}_0$, we define a \defi{$(\sigma, \ell)$-spanning tree}
for $F$ and $F'$ as a spanning tree~$T$ of $F\oplus_{\allpairs} F'$ where the following holds:
\begin{itemize}
    \item  $E(F), E(F')\subseteq E(T)$,
    \item $T$ has exactly $\ell$ {\pos} leaves (i.e., leaves with non-zero labels), and
    \item For each $(i,j,z)\in \sigma$ it holds that $E(T)$ contains exactly $z$ edges between a vertex of label $i$ in $F$ and a vertex of label $j$ in $F'$.
\end{itemize}
    We say that $F$ and $F'$ are \defi{($\sigma, \ell)$-compatible} if there exists a $(\sigma, \ell)$-spanning tree for $F$ and $F'$. 
\end{definition}

Towards deriving reduction rules later, we now define the notion of equivalence, denoted by $\simeq$, for a range of different objects, starting from fixed forests. Informally, fixed forests $F$ and $F'$ are equivalent if for all $(\sigma, \ell)$,  $(\sigma, \ell)$-compatibility with $F$ implies $(\sigma, \ell)$-compatibility of $F'$ and vice versa.
If the implication only goes one way, we can speak of weaker fixed forests.

\begin{definition}
\label{def:elst:cw_equiv_ff}
    A set $R$ of fixed forests is \defi{weaker} than a set $R'$ of fixed forests, denoted by $R \preceq R'$, if for each mapping $\sigma\in \mathcal{A}$, for each integer $\ell\in \mathbb{N}_0$ and for each fixed forest $F^*$ it holds that if $F^*$ is $(\sigma, \ell)$-compatible with some fixed forest $F \in R$, then it is also $(\sigma, \ell)$-compatible with some fixed forest $F' \in R'$.
    We call $R$ and $R'$ \defi{equivalent} ($R \simeq R'$), if both $R \preceq R'$ and $R' \preceq R$ hold.

    Given two fixed forests $F$ and $F'$, we say that $F$ is \defi{weaker} than (resp. \defi{equivalent} to) $F'$, if the singleton set $\{F\}$ is weaker than (resp. equivalent to) the singleton set $\{F'\}$.
\end{definition}

Definition~\ref{def:elst:cw_equiv_ff} also allows us to extend the notion of equivalence to patterns and sets of patterns. Recall that each pattern $A$ is associated with a canonical fixed forest $F_A$.

\begin{definition}
\label{def:elst:equiv_patt}
    Let $A, A'\in\mathcal{P}$. We say $A$ is \defi{weaker} than $A'$ if $F_A$ is weaker than $F_{A'}$.
   We call $A$ and $A'$ \defi{equivalent}, if their canonical fixed forests are equivalent, i.e., $F_A\simeq F_{A'}$. 

   Let $R, R'\in 2^\mathcal{P}$ and let $F_R=\{F_A:\: A\in R \}$, $F_{R'}=\{F_A:\: A\in R' \}$. We say that $R$ is \defi{weaker} than $R'$ if $F_R$ is weaker than $F_{R'}$. We say $R$ and $R'$ are \defi{equivalent} if $F_R\simeq F_{R'}$.
\end{definition}

It is easy to see that all notions of equivalence above are equivalence relations.
In the following, \cref{cor:elst:cw_equiv_repr} states that in order to prove that two fixed forests are equivalent, it suffices to look at their patterns.
By definition, two patterns are equivalent if their \emph{canonical} fixed forests are equivalent.
We show that the reverse is true: two fixed forests are equivalent if their corresponding patterns are.
 For convenience, we first prove this for the case when the patterns are the same in the following lemma.

\begin{longlemma}
\label{lem:elst:cw_same_repr}
    Let $F,F'$ be two fixed forests. If $\repr(F) = \repr(F')$, then $F\simeq F'$.
\end{longlemma}
 \begin{proof}
    Fix an edge mapping $\sigma\in \mathcal{A}$, an integer $\ell\in\mathbb{N}$ and a fixed forest $\tilde{F}$. Suppose that $\tilde{F}$ is $(\sigma, \ell)$-compatible with $F$, and let $T$ be a $(\sigma, \ell)$-spanning tree for $\tilde{F}$ and $F$. We aim to construct a $(\sigma, \ell)$-spanning tree $T'$ of $\tilde{F}$ and $F'$. 
    We initialize $T'$ as the disjoint union of $\tilde{F}$ and $F'$.
    
    By the definition of a pattern corresponding to a fixed forest, since $\repr(F) = \repr(F')$, there is a bijection $\gamma$ that maps each connected component of $F$ to a connected component of $F'$.
    Further, for every connected component $C$ of $F$, $C$ and $\gamma(C)$ have the same number of leaves in every non-zero label class, and the same set of non-zero labels among their internal vertices.
    Define a mapping $\lambda$ between positive vertices of $F$ and $F'$ as follows. For each connected component $C$ of $F$, $\lambda$ bijectively maps positive leaves of each label to leaves of $\gamma(C)$ of the same label (i.e. $\lambda|_{L_i(C)}:L_i(C)\rightarrow L_i(\gamma(C))$ is a bijection for each $i\in [k]$). Each (positive) internal vertex of $C$ is mapped by $\lambda$ to the internal vertex of $\gamma(C)$ of the same label.  
         Now for each edge $e$ of $T$ between a vertex $v$ in $\tilde{F}$ and a vertex $w$, we add an edge $e'$ in $T'$ between $v$ and $\lambda(w)$.
    
    Clearly, $T'$ is a subgraph of $\tilde{F} \oplus_{\allpairs} F'$, since the cut between $F'$ and $\tilde{F}$ is a complete bipartite graph.
    First, we argue that $T'$ is a tree.
    For convenience, we extend $\gamma$ to map also each connected component of $\tilde{F}$ to itself.
    Then for any two connected components $C$ and $C'$ of $\tilde{F} \cup F$, $C$ and $C'$ are connected by an edge in $T$, if and only if $\gamma(C)$ and $\gamma(C')$ are connected by an edge in $T'$.
    This implies that $T'$ is connected.
    Further, since the connected components of $\tilde{F} \cup F'$ are trees, if $T'$ has a cycle, then some connected components of $\tilde{F} \cup F'$ are connected in a cyclic manner.
    This then implies the same cyclic connection between some component of $\tilde{F} \cup F$ in $T$, a contradiction to $T$ being a tree.
    Therefore, $T'$ is a tree as claimed.
    
    Now we show that $T'$ has $\ell$ positive leaves.
    First, by construction, it holds for every vertex $v$ of $\tilde{F}$, that $\deg_T(v) = \deg_{T'}(v)$.
    Hence, such a vertex is a leaf in $T$ if and only if it is a leaf in $T'$.
    Moreover, by the same argument, for every isolated vertex $v$ in $F$, its degree in $T$ is the same as the degree of $\gamma(v)$ in $F'$.
    Finally, consider a connected component $C$ of $F$.
    Note that a vertex of $C$ (or $\gamma(C)$, resp.) is a positive leaf of $T$ ($T'$, resp.) if and only if it is a leaf of $C$ ($\gamma(C)$, resp.) with the same label and it is not connected with a vertex in $\tilde{F}$ by an edge of $T$ ($T'$, resp.).
    Since $\repr(F) = \repr(F')$, the number of positive leaves in $F$ and $F'$ are the same.
    Further, by construction, a vertex $v$ of $C$ is adjacent to a vertex of $\tilde{F}$ in $T$ if and only if $\lambda(v)$ is adjacent to a vertex of $\tilde{F}$ in $T'$.
    The number of vertices in $C$ that are positive leaves in $T$ is the same as the number of vertices in $\gamma(C)$ that are positive leaves in $T'$.
    Overall, we obtain that $T'$ has the same number of positive leaves as $T$ (i.e., $\ell$).

    Finally, since we replace each edge $e$ of $T$ between $\tilde{F}$ and $F$ with an edge $e'$ between $\tilde{F}$ and $F'$ whose endpoints have the same labels, it holds that for each $(i,j)\in \allpairs$, the number of edges between a vertex of label $i$ in $\tilde{F}$ and a vertex of label $j$ in $F'$ is exactly $\sigma(i,j)$.
    
    Thereby we have shown that $T'$ is a $(\sigma, \ell)$-spanning tree of $\tilde{F}$ and $F'$.
    Consequently, it holds that $\tilde{F}$ is $(\sigma, \ell)$-compatible with $F'$. It follows that $F \preceq F'$, since we have chosen $\tilde{F}$, $\sigma$ and $\ell$ arbitrarily.
    It follows by a symmetric argument that $F' \preceq F$, and therefore, $F \simeq F'$.
     \end{proof} 
\begin{corollary}
\label{cor:elst:cw_equiv_repr}
    Let $F$ and $F'$ be two fixed forests. If $\repr(F)\simeq\repr(F')$ holds, then it holds that $F\simeq F'$.
\end{corollary}
\begin{proof}
    It holds by Lemma~\ref{lem:elst:cw_same_repr} that $F \simeq F_{\repr(F)}$, and $F' \simeq F_{\repr(F')}$.
    It follows from \cref{def:elst:equiv_patt} that $F_{\repr(F)} \simeq F_{\repr(F')}$, and hence, $F \simeq F'$.
\end{proof} 
\para{Nice Patterns and Reduction Rules.}

In order to get the claimed running time, we have to reduce the number of records in our dynamic programming tables. In order to achieve this, we define the family of nice patterns, whose size is bounded by $n^{\bigoh(k)}$. We describe reduction rules that can be used to turn any set of patterns into an equivalent set of nice patterns efficiently.

\begin{definition}\label{def:elst:nice_pattern}
Given a vector $u\in\mathbb{N}_0^k$ and a set $S$, let $\omega = |S|+\sum_{i\in[k]} u_i$. We call $(u, S)$ a \defi{unit pair} if $\omega = 1$, and $\defi{zero pair}$ if $\omega = 0$. We say a pattern $(\iso, \mset)$ is \defi{nice} if $\mset$ does not contain any zero pair, and for every $i\in [k]$, $\mset$ contains at most one non-unit pair $(u, S)$ with $i\in S$, and at most one non-unit pair $(u, S)$ with $u_i >0$.

We denote by $\mathcal{P}^*$ the set of all the nice patterns in $\mathcal{P}$.
\end{definition}
 Let us now bound the size of $\mathcal{P}^*$. It is easy to see that we have:
\begin{itemize}
    \item at most $\binom{n}{2k}$ choices for the non-singleton pairs,
    \item at most $n^k$ choices for the number of occurrences of each unit pair of the form $(\mathbf{0}, S)$ with $|S|=1$, where we bound the number of occurrences for each label in $[k]$ by $n$,
    \item at most $n^k$ choices for the ones of the form $(\idv_i, \emptyset)$, and
    \item at most $n^k$ choices for the vector $\iso$.
\end{itemize}
Thus we observe the following bound:
\begin{observation}
\label{obs:elst:sizenice}
It holds that $|\mathcal{P}^*|\leq n^{\bigoh(k)}$.
\end{observation}

\subsection{Requirement Shifting and Pattern Representation}

Now we define a new structure called ``ordered pattern'' that allows us to identify and manipulate specific vectors of a pattern more systematically.

\begin{definition}\label{def:elst:ord_pat}
    An \defi{ordered pattern} is a tuple $\ord = (\iso, O, b)$, where
\begin{itemize}
 \item $\iso\in\mathbb{N}_0^k$ is a vector;   
\item    $O := (u_1, S_1), \dots, (u_m, S_m)$ is a sequence of pairs, where $u_i\in\mathbb{N}_0^k$ and $S_i\subseteq [k]$ for all $i\in [m]$; and
 \item   $b : [k]\times\{\sleaf,\sinner\} \rightarrow [m]\cup\{\perp\}$ is a mapping      such that for $i\in [k]$ the following holds:
    \begin{itemize}
        \item For $j := b(\ell, \sleaf)$ it holds that $(u_{j})_{\ell} \geq 1$ if such a vector $u_{j}$ exists, or $\perp$ otherwise.
        \item For $j := b(\ell, \sinner)$ it holds that $\ell\in S_{j}$ if such a set $S_{j}$ exists or $\perp$ otherwise.
    \end{itemize}
\end{itemize}
    We then call $b(\ell,\sinner)$ the \defi{inner big index} of $\ell$ in $\ord$, and $b(\ell, \sleaf)$ the \defi{outer big index} of $\ell$ in $\ord$.
    Further, we call $b$ a \defi{big-index mapping} of $O$.

    Given a pattern $(\iso, \mset)\in\mathcal{P}$, an \defi{ordered pattern} of $(\iso, \mset)$ is an ordered pattern $\ord_A = (\iso, O_A, b_A)$, where $O_A$ is an arbitrary sequence of the pairs in $A$. Given an ordered pattern $(\iso, O, b)$ over $m$ pairs, we define the \defi{underlying pattern} $p(\iso, O,b)$ as the pair $(\iso, \mset)$ where $\mset$ is the multiset $\langle O_j \colon j\in [m] \rangle$. 
    We call a family of ordered patterns $\ord$ \defi{weaker} than (respectively equivalent to) another family $\ord'$, if the underlying pattern $p(\ord)$ is weaker than (respectively equivalent to) the underlying pattern $p(\ord')$.
\end{definition}

We first define bad indices and bad coordinates as the ``obstacles'' that we need to remove from an ordered pattern to make it nice:

\begin{definition}
    Let $\ord = (\iso, O, b)$ be an ordered pattern with $O=(u_1,S_1),\dots, (u_m, S_m)$.
    For $i\in[m]$ we define the set of \defi{bad coordinates} of $i$ as $Z_{\ord}(i) = \emptyset$ if $(u_i, S_i)$ is a unit pair, and otherwise as $Z_{\ord}(i) := Z^{\sleaf}_{\ord}(i) \cup Z^{\sinner}_{\ord}(i)$, where
	\begin{align*}
	  	Z^{\sleaf}_{\ord}(i) &\coloneq \big\{(q, \sleaf)\colon q \in [k] \land b(q,\sleaf)\neq i \land (u_i)_{q}\geq 1\big\}, \\
    		Z^{\sinner}_{\ord}(i) &\coloneq \big\{(q, \sinner)\colon q \in [k] \land b(q,\sinner)\neq i \land q \in S_i\big\}.
	\end{align*}
         We define the set of \defi{bad indices} of $\ord$ as $I_{\ord} \coloneq I^{\sleaf}_{\ord} \cup I^{\sinner}_{\ord}$ with $I^{\sleaf}_{\ord} \coloneq \{i \in [m]\colon Z^{\sleaf}_{\ord}(i)\neq \emptyset\}$ and $I^{\sinner}_{\ord} \coloneq \{i \in [m]\colon Z^{\sinner}_{\ord}(i)\neq \emptyset\}$.
\end{definition}

\begin{observation}\label{obs:elst:cw_ord_nice}
    Let $\ord = (\iso, O,b)$ be an ordered pattern. Then $p(\ord)$ is nice if and only if it holds that $I_{\ord} = \emptyset$ and $O$ doesn't contain any zero pairs.
\end{observation}

Now we introduce two operations that allow us to ``shift'' requirements between vectors in an ordered pattern. Intuitively, this corresponds to rerouting edges in a solution. Essentially, we show that the patterns resulting from these operations on a pattern correctly replace this pattern, i.e., the initial pattern behaves like an ``OR'' of the resulting patterns, which means that the initial pattern can be extended to a solution if and only if at least one of the resulting patterns can be extended to a solution.

Informally, applying $\pi_1$ and $\pi_2$ to an ordered pattern produces two new ordered patterns as follows. We select one bad index $i\in I_\ord$ and one bad coordinate $(q, X)\in Z_\ord(i)$. In the resulting patterns, the index $(q, X)$ will no longer be a bad coordinate for $i$.

\begin{definition}
\label{def:elst:cw_pi}
    We define the functions \defi{$\rrule_1, \rrule_2$} that take as input
    an ordered pattern $\ord \coloneq (\iso, O, b)$ with $O = (u_1, S_1), \dots, (u_m, S_m)$, an index $i\in[m]$ and a bad coordinate $(q, X)\in Z_{\ord}(i)$.
         Let $j = b(q, X)$.

    We define $\rrule_1(\ord, i, q, X)$ as the ordered pattern $(\iso, O^1,b^1)$  with $O^1 = (u'_1,S'_m),\dots (u'_m, S'_m)$ as follows:
    For $X = \sleaf$ we define
    \[
    u'_h = \begin{cases}
        u_j + u_i - \idv_{q} & \colon h = j,\\
        \idv_{q} & \colon h = i,\\
        u_h & \colon \text{otherwise,}
    \end{cases}
    \qquad\text{and}\qquad
    S'_h = \begin{cases}
        S_j \cup S_i & \colon h = j,\\
        \emptyset & \colon h=i,\\
        S_h & \colon \text{otherwise,}
    \end{cases}
    \]
    and for $X = \sinner$ we define 
    \[
    u'_h = \begin{cases}
        u_j + u_i & \colon h = j,\\
        \mathbf{0} & \colon h=i,\\
        u_h & \colon \text{otherwise.}
    \end{cases}
    \qquad\text{and}\qquad
    S'_h = \begin{cases}
        S_j \cup (S_i \setminus \{q\}) = S_j \cup S_i & \colon h = j,\\
        \{q\} & \colon h = i,\\
        S_h & \colon \text{otherwise,}
    \end{cases}
    \]
    We define $b^1$ by setting $b^1(q') = j$ for all values $q' \in b^{-1}(i)$, and $b^1(q') = b(q')$ for all other values of $q'$.

    We define the ordered pattern $(\iso, O^2, b^2)\coloneq\rrule_2(\ord, i, q, X)$ with $O^2 = (u'_1,S'_1),\dots (u'_m, S'_m)$ as follows:
    For $X=\sleaf$ we define $S'_h=S_h$ for all $h\in [m]$ and 
    \[
    u'_h = \begin{cases}
        u_j + (u_i)_{q} \cdot \idv_{q} & \colon h = j,\\
        u_i - (u_i)_{q} \cdot \idv_{q} & \colon h = i,\\
        u_h & \colon \text{otherwise,}
    \end{cases}
    \]
    and for $X = \sinner$, we define $u'_h = u_h$ for all $h\in [m]$ and 
    \[
    S'_h = \begin{cases}
        S_j \cup \{q\} = S_j & \colon h = j,\\
        S_i \setminus \{q\} & \colon h = i,\\
        S_h & \colon \text{otherwise.}
    \end{cases}
    \]
    Finally, we define $b^2 = b$.
\end{definition}

Intuitively, $\pi_1$ shifts all inner vertices and leaves except for a single vertex of label $q$ to the big index of $q$, while $\pi_2$ shifts all inner vertices or all leaves of label $q$, keeping all the other labels untouched.
Note that in the definition above, the first index of each of $\ord$, $\rrule_1(\ord, i, q, X)$, and $\rrule_2(\ord, i, q, X)$ is the same, i.e., both mappings $\rrule_1$ and $\rrule_2$ preserve the vector $\iso$.

In the following lemmas we show that the mappings $\rrule_1$ and $\rrule_2$ produce two ordered patterns equivalent to the original pattern, proving the correctness of our reduction rules.
For convenience, we say the \defi{type} of a vertex in a graph is either isolated, leaf, or internal, if its degree is 0, 1, or at least two, respectively.

\begin{longlemma}\label{lem:elst:cw_equiv_rrule_1}
    Let $\ord \coloneq (\iso, O, b)$ be an ordered pattern with $O = (u_1, S_1), \dots, (u_m, S_m)$, and $i\in[m]$. For a bad coordinate $(q, X) \in Z_{\ord}(i)$ where $q\in [k]$ and $X\in \{\sleaf, \sinner\}$, let $\ord^1 = \rrule_1(\ord, i, q, X)$.
   Then it holds that $\ord^1 \preceq \ord$.
\end{longlemma}
\begin{proof}
   Let $\ord^1 = (\iso, O^1, b^1)$ with $O^1 = (u'_1,S'_m),\dots (u'_m, S'_m)$ and let $j=b(q, X)$.
   Let $F$ and $F'$ be the corresponding fixed forests of $p(\ord)$ and $p(\ord^1)$, respectively.
   For $t \in [m]$, let $C_t$ (resp., $C'_t$) be the trees in $F$ (resp., $F'$) corresponding to $(u_t, S_t)$ (resp. $(u'_t, S'_t)$).
   By definition of $\rrule_1$ and by construction, for $t \in [m] \setminus \{i, j\}$, there is an isomorphism between $C_t$ and $C'_t$ that also preserves the labels.
   For convenience, we consider $C_t$ and $C'_t$ to be identical for these values of $t$.
   Further, observe that regardless of $X=\sleaf$ or $X=\sinner$, all positive vertices of $C'_i$ have label $q$.

   For a mapping $\sigma\in\mathcal{A}$, an integer value $\ell\in\mathbb{N}_0$ and a fixed forest $F''$, suppose that there exists a $(\sigma, \ell)$-spanning tree $T'$ for $F''$ and $F'$.
   We construct a $(\sigma, \ell)$-spanning tree for $F''$ and $F$ as follows.
   We start by removing all vertices of $C'_i$ and $C'_j$ and their incident edges from $T'$.
   Then we add $C_i$ and $C_j$ (by a disjoint union) to the resulting graph to obtain $\bar{T}$.
   We now aim to add some edges to $\bar{T}$ to form a $(\sigma, \ell)$-spanning tree for $F''$ and $F$.
   
   Consider a connected component $C$ of $\bar{T}$ that is not $C_i$ or $C_j$.
   There are three cases.
   
   \noindent
   \textbf{Case 1.} \textit{$C$ is incident to only $C'_i$ in $T$.}
   In this case, there is a removed edge in the construction of $\bar{T}$ that connects a vertex $u \in C$ and a vertex $v' \in C'_i$.
   As argued above $v'$ has label $q$.
   By assumption, $C_i$ has a vertex $v$ of the same type and label as $v'$.
   Then we add an edge to $\bar{T}$ connecting $u$ and $v$.

   \noindent   
   \textbf{Case 2.} \textit{$C$ is incident to only $C'_j$ in $T$.}
   In this case, there is a removed edge in the construction of $\bar{T}$ that connects a vertex $u \in C$ and a vertex $v' \in C'_j$.
   If $C_j$ has a vertex $v$ of the same type and label as $v'$, then we add an edge to $\bar{T}$ connecting $u$ and $v$.
   Otherwise, by definition of $\rrule_1$, there is a vertex $w$ with the same type and label as $v'$ in $C_i$, we then add an edge to $\bar{T}$ connecting $u$ and $w$.

   \noindent   
   \textbf{Case 3.} \textit{$C$ is incident to both $C'_i$ and $C'_j$ in $T$.}
   Note that there is only one such connected component, since $T$ is a tree.
   There must be two removed edges, one between a vertex $u_i \in C$ and a vertex $v'_i \in C'_i$ and another between a vertex $u_j \in C$ and a vertex $v'_j \in C'_j$.
   As argued above, $v'_i$ has label $q$, and $v'_i$ must be a leaf in $C'_i$ if $X = \sleaf$ and an internal vertex otherwise.
   Consider two subcases.

   \noindent   
   \textbf{Subcase 3.1.} \textit{$C_j$ has a vertex $v_j$ of the same type and label as $v'_j$.} 
   Then we add an edge to $\bar{T}$ connecting $u_j$ and $v_j$.
   Further, similar to Case 1, $C_i$ has a vertex $v_i$ of the same type and label as $v'_i$.
   We then also add an edge to $\bar{T}$ connecting $u_i$ and $v_i$.

   \noindent   
   \textbf{Subcase 3.2.} \textit{$C_j$ does not have a vertex $v_j$ of the same type and label as $v'_j$.} 
   Then similar to the argument in Case 2, by definition of $\rrule_1$, this implies that $C_i$ has a vertex $v_i$ of the same type and label as $v'_j$.
   Further, by assumption that $j = b(q, X)$, $C_j$ has a vertex $v_j$ with the same type and label as $v'_i$.
   Then we add two edges $u_j v_i$ and $u_i v_j$ to $\bar{T}$.
   
   Let $T$ be the resulting graph after adding the edges as described above.
   By construction, there is a unique component of $\bar{T}$ connecting to $C_i$ and $C_j$ in $T$, and other components of $\bar{T}$ only connects with exactly one of $C_i$ or $C_j$.
   Since we only use one edge for each of these connections, and since each component of $\bar{T}$ is a tree, $T$ then is a spanning tree of $F'' \oplus_\alpha F$.
   Further, the incident vertices of the added edges in the process above have the same type (in $\bar{T}$) and labels as the incident vertices of the removed edges from $T'$.
   Hence, the number of positive leaves in $T$ and $T'$ are the same.
   Finally, since we only replaced edges between $F''$ and $F'$ with edges between $F''$ and $F$ whose endpoints have the same labels, it holds for each $(i,j)\in \allpairs$, that the number of edges between a vertex of label $i$ in $F''$ and a vertex of label $j$ in $F$ is exactly $\sigma(i,j)$.
   Therefore, $T$ is a $(\sigma, \ell)$-spanning tree for $F''$ and $F$ which concludes the proof.
\end{proof}

\begin{longlemma}\label{lem:elst:cw_equiv_rrule_2}
    Let $\ord \coloneq (\iso, O, b)$ be an ordered pattern with $O = (u_1, S_1), \dots, (u_m, S_m)$, and $i\in[m]$. For a bad coordinate $(q, X) \in Z_{\ord}(i)$ for $q\in [k]$ and $X\in \{\sleaf, \sinner\}$, let $\ord^2 = \rrule_2(\ord, i, q, X)$.
   Then it holds that $\ord^2 \preceq \ord$.
\end{longlemma}
\begin{proof}
   We follow the same approach as in the proof of \cref{lem:elst:cw_equiv_rrule_1}.
   In particular, let $\ord^2 = (\iso, O^2, b^2)$ with $O^2 = (u'_1,S'_m),\dots (u'_m, S'_m)$. Let $j=b(q, X)$.
   Let $F$ and $F'$ be the corresponding fixed forests of $p(\ord)$ and $p(\ord^2)$, respectively.
   Then we define $C_t, C'_t, T', \bar{T}$ the same way as in the proof of \cref{lem:elst:cw_equiv_rrule_1}.
   We now also aim to add some edges to $\bar{T}$ to form a $(\sigma, \ell)$-spanning tree for $F''$ and $F$.
   
   Consider a connected component $C$ of $\bar{T}$ that is not $C_i$ or $C_j$.
   There are three cases.

   \noindent   
   \textbf{Case 1.} \textit{$C$ is incident to only $C'_i$ in $T$.}
   In this case, there is a removed edge in the construction of $\bar{T}$ that connects a vertex $w \in C$ and a vertex $v' \in C'_i$.
   By definition of $\rrule_2$, regardless of $X = \sleaf$ or $X = \sinner$, $C_i$ has a vertex $v$ of the same type and label as $v'$.
   Then we add the edge $wv$ to $\bar{T}$.

   \noindent   
   \textbf{Case 2.} \textit{$C$ is incident to only $C'_j$ in $T$.}
   In this case, there is a removed edge in the construction of $\bar{T}$ that connects a vertex $w \in C$ and a vertex $v' \in C'_j$.
   By definition of $\rrule_2$ and by the fact that $j = b(q,X)$, we can see that for every combination of type and label that appears in $C'_j$, it also appears in $C_j$.
   Hence, there is a vertex $v$ in $C_j$ with the same type and label as $v'$.
   Then we add the edge $wv$ to $\bar{T}$.

   \noindent   
   \textbf{Case 3.} \textit{$C$ is incident to both $C'_i$ and $C'_j$ in $T$.}
   Note that there is only one connected component in this case, since $T$ is a tree.
   There must be two removed edges, one between a vertex $w_i \in C$ and a vertex $v'_i \in C'_i$ and another between a vertex $w_j \in C$ and a vertex $v'_j \in C'_j$.
   Similar to the argument in Case 1, $C_i$ has a vertex $v_i$ with the same type and label as $v'_i$.
   Similar to the argument in Case 2, $C_j$ has a vertex $v_j$ with the same type and label as $v'_j$.
   Then we add the edges $v_i w_i$ and $v_j w_j$ to $\bar{T}$.
   
   The argument for the resulting $T$ to be a $(\sigma, \ell)$-spanning tree for $F''$ and $F$ is similar to the proof of \cref{lem:elst:cw_equiv_rrule_1}.
\end{proof}

\begin{longlemma}\label{lem:elst:cw_equiv_rrule_12}
     Let $\ord \coloneq (\iso, O, b)$ be an ordered pattern with $O = (u_1, S_1), \dots, (u_m, S_m)$, and $i\in[m]$. For a bad coordinate $(q, X) \in Z_{\ord}(i)$ for $q\in [k]$ and $X\in \{\sleaf, \sinner\}$, and let $\ord^1 = \rrule_1(\ord, i, q, X)$ and $\ord^2 = \rrule_2(\ord, i, q, X)$.
    Then it holds that $\{ \ord \} \preceq \{ \ord^1, \ord^2 \}$.
\end{longlemma}
\begin{proof}
   Let $\ord^1 = (\iso, O^1, b^1)$ with $O^1 = (u^1_1,S^1_m),\dots (u^1_m, S^1_m)$. Let $j=b(q, X)$.
   Let $\ord^2 = (\iso, O^2, b^2)$ with $O^2 = (u^2_1,S^2_m),\dots (u^2_m, S^2_m)$.
   Let $F$, $F^1$, and $F^2$ be the corresponding fixed forests of $p(\ord)$, $p(\ord^1)$, and $p(\ord^2)$, respectively.
   For $t \in [m]$, let $C_t$ be the tree of $F$ corresponding to $(u_t, S_t)$.
   We define $C^1_t$ and $C^2_t$ analogously for $F^1$ and $F^2$.
   Similar to the proof of \cref{lem:elst:cw_equiv_rrule_1}, we consider $C_t$, $C^1_t$, and $C^2_t$ to be identical for $t \in [m] \setminus \{i,j\}$.

   For a mapping $\sigma\in\mathcal{A}$, an integer value $\ell\in\mathbb{N}_0$ and a fixed forest $F''$,
   suppose there exists a $(\sigma, \ell)$-spanning tree $T$ for $F''$ and $F$.
   We need to construct a $(\sigma, \ell)$-spanning tree either for $F''$ and $F^1$ or for $F''$ and $F^2$.
   Let $T'$ be the forest resulting from removing all the vertices of $C_i$ and $C_j$ from $T$.
   Let $C^*$ be the unique connected component of $T'$ that is connected to $C_i$ and $C_j$ directly by an edge in $T$.
   That is, there exist two edges in $T$, one connecting a vertex $w_i \in C^*$ with a vertex $v_i \in C_i$ and one connecting a vertex $w_j \in C^*$ with a vertex $v_j \in C_j$.
   We distinguish two cases depending on the type and label of $v_i$.
   For a vertex $v$ in $F$, we define the \emph{label-type} of $v$ as the pair $(q, X)$ where $q = \lab_F(v)$, and $X=\sleaf$ if $v$ is a leaf of $F$, or $X=\sinner$ otherwise.

   \noindent   
   \textbf{Case 1.} If the label-type of $v_i$ is not equal to $(q, X)$:
   Since $C^2_i$ results from $C_i$ by removing all vertices of the label-type $(q, X)$, and adding them to $C^2_j$ (where we upper bound the multiplicity of inner vertices by one), it holds first that $C^2_i$ has a vertex $v^2_i$ of the same type and label as $v_i$, and $C^2_j$ has a vertex $v^2_j$ of the same type and label as $v_j$.
   We then add the edges $v^2_i w_i$ and $v^2_j w_j$ to $T'$.
   Second, We can define an injective mapping between the vertices of $C_j$ and the vertices of $C^2_j$ that preserves both type and label. We can extend this mapping to cover the vertices of $C_i$ of label-type $(q, X)$, where the restriction of the resulting mapping to leaves (on both sides) is even bijective. Let us extend this mapping by the identity over the rest of $T$. Then we can replace each edge between $F$ and $F''$ in $T$ that is not incident to $C_i$ or $C_j$ with the edge between the images of its endpoints under this mapping. Intuitively, we have shifted all edges whose endpoints in $C_j$ have label-type $(q, X)$ to $C^2_j$.
   
   By the same argument as in the end of the proof of \cref{lem:elst:cw_equiv_rrule_1}, we obtain that the resulting graph is a $(\sigma,\ell)$-spanning tree for $F''$ and $F^2$.
   
   \noindent   
   \textbf{Case 2.} Otherwise, the vertex $v_i$ has label-type $(q, X)$:
   Since $C^1_i$ results from $C_i$ by removing all vertices but a single vertex of label-type $(q, X)$, and adding them to $C^1_j$, it holds first that $C^1_i$ 
   has a vertex $v^1_i$ with the same type and label as $v_i$ and $C^1_j$ has a vertex $v^1_j$ of the same label and type as $v_j$.
   We then add the edges $v^1_i w_i$ and $v^1_j w_j$ to $T'$. 
   With analogous arguments as in Case 1, we can 
   define a mapping between all vertices of $C_i$ and $C_j$ except for a single vertex of the label-type $(q, X)$ and the vertices of $C^1_j$ that preserve both type and label, and such that the restriction of this mapping to leaves is bijective. We extend this mapping by the identity over the rest of $T$.
   We then replace each edge between $F$ and $F''$ in $T$ that is not incident to $C_i$ or $C_j$ with the edge between the images of its endpoints under this mapping. Intuitively we shift all edges from $C_i$ to $C_j$, except for the single edge tha connects them.
   Again, we conclude that the resulting graph is a $(\sigma,\ell)$-spanning tree for $F''$ and $F^1$.
   Thus we obtain a $(\sigma, \ell)$-spanning tree for $F''$ and $F^1$.
\end{proof} 
Combining Lemmas~\ref{lem:elst:cw_equiv_rrule_1}, \ref{lem:elst:cw_equiv_rrule_2} and~\ref{lem:elst:cw_equiv_rrule_12} immediately yields the following corollary:
\begin{corollary}
\label{cor:elst:cw_equiv_rrule12}
     Let $\ord \coloneq (\iso, O, b)$ be an ordered pattern with $O = (u_1, S_1), \dots, (u_m, S_m)$, an index $i\in[m]$ and a coordinate $q\in[k]$ and $X\in \{\sleaf, \sinner\}$ such that $(q, X)\in Z_\ord (i)$.
           Let $j = b(q, X)$.
    Further, let $\ord^1 = \rrule_1(\ord, i, q, X)$ and $\ord^2 = \rrule_2(\ord, i, q, X)$. 
    Then $\{ \ord \} \simeq \{ \ord^1, \ord^2 \}$.
\end{corollary}

\para{Subroutines.}
Based on the functions $\pi_1$ and $\pi_2$, we formulate a subroutine that transforms a single pattern into an equivalent family of nice patterns.
The first two subroutines $\operatorname{ShiftVectorL}$ and $\operatorname{ShiftVectorI}$ take as input an ordered pattern $\ord \coloneq (\iso, O, b)$, and a bad index $i\in I_{\ord}$, and output a set of ordered patterns, that result from an exhaustive application of the operations $\rrule_1$ and $\rrule_2$ to all bad coordinates of either type ($\sleaf$ or $\sinner$) of $i$. We attach to each resulting pattern a single index. This is the index we shifted to by an application of $\rrule_1$, and is the only index that may have turned bad in the resulting pattern.

\begin{procedure}
\label{def:elst:cw_shift_vector_L}
    The operation \defi{$\operatorname{ShiftVectorL}$} takes as input a tuple $(\ord, i)$, where $\ord \coloneq (\iso, O, b)$ is an ordered pattern and  $i\in I_{\ord}^{\sleaf}$.
    Let $O = (u_1, S_1), \dots (u_m, S_m)$.
         Further, let $Z_\ord^\sleaf(i)=\{(z_1,\sleaf),\dots, (z_t, \sleaf)\}$.
    The operation outputs a set of pairs $\{(\ord^1, p_1),\dots,(\ord^{t+1}, p_{t+1})\}$, where $p_j=b(z_j, \sleaf)$ for $j\in [t]$ and $p_{t+1}=\perp$. Let  
    We define the pattern $\ord^j = (\iso, O^j, b^j)$ for $j\in[t+1]$ with 
    $O^j = (u^j_1, S^j_1), \dots, (u^j_m, S^j_m)$ as follows:

    For $j \in [t]$ and $d\in[m]$ we define the vector $(u^j_d, S^j_d)$ as follows:
    If $d = i$, we define $(u^j_d, S^j_d) = (\idv_{z_j}, \emptyset)$. 
    If $d = b(z_j, \sleaf)$, we define $S^j_d = S_d \cup S_i$ and 
    \[
    (u^j_d)_z = \begin{cases}
        (u_d)_z & \colon z \in \{z_1, \dots, z_{j-1}\} \text{ and } d \neq b(z, \sleaf),\\
        (u_d)_z + (u_i)_z - 1 & \colon z = z_j,\\
        (u_d)_z + (u_i)_z & \colon \text{otherwise.}
    \end{cases}
    \]
    Otherwise, we define $S^j_d = S_d$ and
    \[
    (u^j_d)_z = \begin{cases}
        (u_d)_z + (u_i)_z & \colon z \in \{z_1, \dots, z_{j-1}\} \text{ and } d = b(z, \sleaf),\\
        (u_d)_z & \colon \text{otherwise.}
    \end{cases}
    \]
    We define $b^j$ as the mapping such that $b^j(z', X') = b(z_j,  \sleaf)$ for all values $(z',X') \in b^{-1}(i)$, and $b^j(z', X') = b(z', X')$ for all other pairs $(z', X') \in [k]\times\{\sinner, \sleaf\}$.
    
    Finally, we define $b^{t+1} = b$. 
    For $d\in[m]$, we define $S^{t+1}_d = S_d$ and
    \[
    (u^{t+1}_d)_z = \begin{cases}
        0 & \colon d = i \land (z, \sleaf) \in Z_{\ord}(i),\\
        (u_j)_z + (u_i)_z & \colon d = b(z, \sleaf) \land (z, \sleaf) \in Z_{\ord}(i),\\
        (u_d)_z & \colon \text{otherwise.}
    \end{cases}
    \]
    \xqed
\end{procedure}

\begin{observation}\label{obs:elst:shift_vector_L}
    Let $(\ord', p)$ be some pair in the output of $\operatorname{ShiftVectorL}(\ord, i)$ with $\ord'= (\iso, O', b')$. If $p\neq perp$, then all following hold:
    \begin{itemize}
        \item $I^{X}_{\ord'} = I^{X}_{\ord}\setminus \{i\}$ or $I^{X}_{\ord'} = (I^{X}_{\ord}\setminus \{i\}) \cup \{p\}$ for each value $X \in \{\sleaf, \sinner\}$, and
        \item $b'([k]\times \{\sleaf, \sinner\})= b([k]\times \{\sleaf,\sinner\})\setminus \{i\}$.
    \end{itemize}

    For $p=\perp$, the following points hold:
    \begin{itemize}
        \item $I^{\sleaf}_{\ord'} = I^{\sleaf}_{\ord}\setminus \{i\}$,
        \item $I^{\sinner}_{\ord'} = I^{\sinner}_{\ord}$,
        \item and $b' = b$.
    \end{itemize}
\end{observation}

\begin{procedure}
\label{def:elst:cw_shift_vector_I}
    The operation \defi{$\operatorname{ShiftVectorI}$} takes as input a tuple $(\ord, i)$, where $\ord \coloneq (\iso, O, b)$ is an ordered pattern and $i\in I_{\ord}^{\sinner}$.
    Let $O = (u_1, S_1), \dots (u_m, S_m)$.
    Further, let $Z_\ord^\sinner(i)=\{(z_1,\sinner),\dots,(z_t, \sinner)\}$, where $z_1<\dots<z_t$. 
    The operation outputs a set of pairs $\{(\ord^1, p_1),\dots,(\ord^{t+1}, p_{t+1})\}$, where $p_j=b(z_j, \sinner)$ for $j\in [t]$ and $p_{t+1}=\perp$. We define the patterns
    $\ord^j = (\iso, O^j, b^j)$ for $j\in[t+1]$ with 
    $O^j = (u^j_1, S^j_1), \dots, (u^j_m, S^j_m)$ as follows:
    
    For $j \in [t]$ and $d \in [m]$, we define
    \[
    u^{j}_d = \begin{cases}
        \mathbf{0} & \colon d = i,\\
        u_d + u_i & \colon d = b(z_j, \sinner),\\
        u_d & \colon \text{otherwise,}
    \end{cases}
    \qquad
    S^{j}_d = \begin{cases}
        \{ z_j \} & \colon d = i,\\
         S_d \cup (S_i \setminus \{z_1, \dots, z_{j-1}\}) & \colon d = b(z_j, \sinner),\\
        S_d & \colon \text{otherwise.}
    \end{cases}
    \]

    We define $b^j$ as the mapping such that $b^j(z', X') = b(z_j,  \sinner)$ for all values $(z',X') \in b^{-1}(i)$, and $b^j(z', X') = b(z', X')$ for all other pairs $(z', X') \in [k]\times\{\sinner, \sleaf\}$.

    Finally, we define $b^{t+1} = b$. 
    For $d\in[m]$, we define $u^{t+1}_d = u_d$ and
    \[
        S^{t+1}_d = \begin{cases}
        S_d \setminus \{z_1, \dots, z_t\} & \colon d = i \\
        S_d & \colon \text{otherwise.}
    \end{cases}
    \]
\end{procedure}

\begin{observation}\label{obs:elst:shift_vector_I}
    Let $(\ord', p)$ be some pair in the output of $\operatorname{ShiftVectorI}(\ord, i)$ with $\ord'= (\iso, O', b')$. If $p\neq \perp$, then all of the following hold:
    \begin{itemize}
        \item $I^{X}_{\ord'} = I^{X}_{\ord}\setminus \{i\}$ or $I^{X}_{\ord'} = (I^{X}_{\ord}\setminus \{i\}) \cup \{p\}$ for each value $X \in \{\sleaf, \sinner\}$, and
        \item $b'([k]\times \{\sleaf, \sinner\})= b([k]\times \{\sleaf,\sinner\})\setminus \{i\}$.
    \end{itemize}

    For $p=\perp$, the following points hold:
    \begin{itemize}
        \item $I^{\sinner}_{\ord'} = I^{\sinner}_{\ord}\setminus \{i\}$,
        \item $I^{\sleaf}_{\ord'} = I^{\sleaf}_{\ord}$,
        \item and $b' = b$.
    \end{itemize}
\end{observation}

\begin{lemma}\label{lem:elst:cw_shift_vector_equiv}
    Given an ordered pattern $\ord$ and an index $i\in I^{\sleaf}_{\ord}$, let $P$ be the set resulting from $\operatorname{ShiftVectorL}(\ord, i)$ and let $S=\{\ord\colon (\ord, p)\in P\}$. Then $S$ is equivalent to $\ord$. The same holds if $i\in I^{\sinner}_{\ord}$ and we use $\operatorname{ShiftVectorI}(\ord, i)$ instead.
\end{lemma}

\begin{proof}
    We prove the claim for $\operatorname{ShiftVectorL}(\ord, i)$ and $i\in I^{\sleaf}_{\ord}$.
    The proof for $\operatorname{ShiftVectorI}(\ord, i)$ follows from the analogous arguments applied to $Z^{\sinner}_{\ord}(i)$ instead of $Z^{\sleaf}_{\ord}(i)$.

    Let $Z_\ord^\sleaf(i)=\{(z_1,\sleaf),\dots, (z_t, \sleaf)\}$, and let $\ord^1, \dots, \ord^{t+1}$ be the set of ordered patterns as defined in \cref{def:elst:cw_shift_vector_L}.
         We define the sequence of ordered patterns $\gamma^0, \gamma^1, \dots, \gamma^t$ with
    $\gamma^0 = \ord$, and $\gamma^j=\rrule_2(\gamma^{j-1}, i, z_j, X)$ for $j\in[t]$.
    Then it holds for each $j\in[t]$ that $\ord^j=\rrule_1(\gamma^{j-1}, i, z_j, X)$.
    We show by induction over $j$ that $\{\ord^1, \dots, \ord^j\}\cup \{\gamma^j\}$ is equivalent to $\{\ord\}$. 
    The lemma then follows from the observation that $\ord^{t+1} = \gamma^t$.

    For $j=0$ the claim holds trivially since $\gamma^0 = \ord$.
    Assume now that the claim holds for $j-1$, i.e., that $\{\ord^1, \dots, \ord^{j-1}\}\cup \{\gamma^{j-1}\}$ is equivalent to $\{\ord\}$. 
    Since it holds that $\ord^j=\rrule_1(\gamma^{j-1}, i, z_j, X)$ and $\gamma^j=\rrule_2(\gamma^{j-1},i,z_j,X)$, 
    it holds again by \cref{cor:elst:cw_equiv_rrule12} that $\{\ord^1, \dots, \ord^{j}\}\cup \{\gamma^j\}$ is equivalent to $\{\ord\}$.
\end{proof}

Now we define the procedure $\operatorname{ShiftVector}$ that combines the two operations above. 

\begin{procedure}
    \label{def:elst:cw_shift_vector}
    The operation \defi{$\operatorname{ShiftVector}$} takes as input a tuple $(\ord, i)$, where $\ord \coloneq (\iso, O, b)$ is an ordered pattern and $i\in I_{\ord}$. Let $\{(\ord^1, p_1), \dots, (\ord^{t+1}, p_{t+1})\}=\operatorname{ShiftVectorL}(\ord,i)$, where $(\ord^{t+1}, p_{t+1})$ is the unique pair such that $p_{t+1}=\perp$.   
    Let $\{(\tilde{\ord}^1, \tilde{p}_1), \dots, (\tilde{\ord}^{s+1}, \tilde{p}_{s+1})\} = \operatorname{ShiftVectorI}(\ord^{t+1}, i)$. Then the operation outputs the set \[\{(\ord^1, p_1), \dots, (\ord^{t},p_t), (\tilde{\ord}^1, \tilde{p}_1), \dots, (\tilde{\ord}^{s+1}, \tilde{p}_{s+1})\}.\]
\end{procedure}

Then the following statement is an immediate consequence of \cref{lem:elst:cw_shift_vector_equiv}.

\begin{corollary}
\label{cor:elst:cw_shift_vector_equiv}
    Given an ordered pattern $\ord$ and an index $i\in I_{\ord}$, let $P$ be the set resulting from $\operatorname{ShiftVector}(\ord, i)$ and let $S=\{\ord\colon (\ord, p)\in P\}$. Then $S$ is equivalent to $\ord$.
\end{corollary}

The following observation follows by Observations~\ref{obs:elst:shift_vector_I} and \ref{obs:elst:shift_vector_L}, since the union operation preserves equivalence.

\begin{observation}
    Let $(\ord', p)$ be some pair in the output of $\operatorname{ShiftVector}(\ord, i)$ with $\ord'= (\iso, O', b')$. Then all following hold:
    \begin{itemize}
        \item either $I_{\ord'} = I_{\ord}\setminus \{i\}$ or $I_{\ord'} = (I_{\ord}\setminus \{i\}) \cup p$ for each value $X\in\{\sleaf,\sinner\}$; and
        \item $b'([k]\times \{\sleaf,\sinner\})\subseteq b([k]\times\{\sleaf,\sinner\})$.
    \end{itemize}
    In particular, it holds for the unique pair $(\ord, p)$ with $p=\perp$ that $I_{\ord'} = I_{\ord}\setminus \{i\}$.
\end{observation}

Now we define the recursive operation $\operatorname{ReduceVector}$ that exhaustively applies the operation $\operatorname{ShiftVector}$ to remove a single bad index, without creating further bad indices.

\begin{procedure}
    \label{def:elst:cw_reduce_vector}
    The recursive operation \defi{$\operatorname{ReduceVector}$} takes as input an ordered pattern $\ord$ over $m$ pairs, an index $i\in[m]$, and outputs a family of ordered patterns $C$, defined as follows.
    If $i \notin I_{\ord}$, we define $C \coloneq \{\ord\}$.
         Otherwise, let $\{(\ord^1, p_1),\dots,(\ord^{t+1}, p_{t+1})\}= \operatorname{ShiftVector}(\ord,i)$, and let $(\ord^{t+1}, p_{t+1})$ be the unique pair with $p_{t+1}=\perp$.
    Then for $j\in[t]$ we define the sets $C_j \coloneq \operatorname{ReduceVector}(\ord^j, p_j)$.
    Finally, we define $C \coloneq \{\ord^{t+1}\} \cup \bigcup_{j\in[t]} C_j$.
\end{procedure}

Now we prove some properties of this operation that we need later to show the correctness of our procedures. We start with the following definition:

\begin{definition}
    Let $\ord$ be an ordered pattern over $m$ pairs and $i\in[m]$.
    We define the tree $\mathcal{T}_{\ord, i}$ as the recursion tree of the operation $\operatorname{ReduceVector}(\ord, i)$ (we omit the subscript $\ord$ when clear from context). 
    We also label each node of $\mathcal{T}_{\ord,i}$ corresponding to a call $(\ord', i')$ with the index $i'$.
\end{definition}

\begin{observation}\label{obs:elst:cw_reduce_vector}
              Let $\ord=(\iso, O, b)$ and $\ord'=(\iso, O', b')$ be ordered patterns such that $(\ord', i')$ is a child of the root of the recursion tree $\mathcal{T}_{\ord,i}$.
    It holds that:
    \begin{enumerate}
        \item $(O')_i$ is a unit pair. In particular, it holds that $i\notin I_{\ord'}$;
        \item $(b')^{-1}(i) = \emptyset$;
        \item $i' = b(q, X)$ for some bad coordinate $(q, X)$ of $i$ with respect to $b$; and
        \item $b'([k] \times \{\sleaf, \sinner\}) \subseteq b([k] \times \{\sleaf, \sinner\})$.
    \end{enumerate}
\end{observation}

\begin{lemma}\label{lem:elst:cw_reduce_vector_distinct_labels}
    There cannot exist two distinct nodes on the same root-leaf path of the recursion tree $\mathcal{T}_i$ sharing the same label $j\in[m]$.
\end{lemma}

\begin{proof}
    It holds by induction over the depth of the nodes in the recursion tree, using item (2) of \cref{obs:elst:cw_reduce_vector}, that for a node corresponding to a call $\operatorname{ReduceVector}(\ord', i')$, and for a descendant node corresponding to a call $\operatorname{ReduceVector}(\ord'', i'')$, it holds that $(b'')^{-1}(i')= \emptyset$. The lemma follows then by item (3) of \cref{obs:elst:cw_reduce_vector}.
\end{proof}

\begin{lemma}\label{lem:elst:cw_reduce_vector_size}
    Let $\mathcal{T}_{\ord,i}$ be the recursion tree of a call $\operatorname{ReduceVector}(\ord, i)$.
    Then $\mathcal{T}_{\ord,i}$ has at most $(2k)^k$ nodes. In particular, the procedure runs in time $k^{\oh(k)}$, and it holds that $|\operatorname{ReduceVector}(\ord, i)|\leq k^{\bigoh(k)}$.
\end{lemma}

\begin{proof}
    A call of $\operatorname{ReduceVector}(\ord, i)$ makes at most $|Z_{\ord}(i)|\leq 2k$ recursive calls, and hence, each node of $\mathcal{T}_{\ord,i}$ has at most $2k$ children. Moreover, by \cref{lem:elst:cw_reduce_vector_distinct_labels}, the labels assigned to the nodes of each root-leaf path of the recursion tree are pairwise distinct. Hence, $\mathcal{T}_{\ord,i}$ has height at most $k$. It follows that $\mathcal{T}_{\ord,i}$ has at most $(2k)^{k}$ nodes. The bound on the running time follows then, since each node of the recursion tree, aside from the recursive calls, only makes a call of $\operatorname{ShiftVector}$, which runs in polynomial time.
    
    Now let $T_h$ be the maximum size of a family returned by a call $\operatorname{ReduceVector}(\ord, i)$ with a recursion tree of height $h$. Then it holds that $T_0 = 1$, and for $h\geq 1$, it holds that $T_h\leq k\cdot T_{h-1} + 1\leq (k+1)T_{h-1}$. This implies that $T_k\leq (k+1)^k\leq k^{\bigoh(k)}$.
\end{proof}

\begin{lemma}\label{lem:elst:cw_reduce_vector_equiv}
    Let $\ord$ be an ordered pattern over $m$ pairs, $i\in[m]$. Then the operation $\operatorname{ReduceVector}(\ord, i)$ produces a family of ordered patterns equivalent to $\{\ord\}$, such that $I_{\ord'}\subseteq I_{\ord}\setminus\{i\}$ and $b'([k] \times \{\sleaf, \sinner\})\subseteq b([k] \times \{\sleaf, \sinner\})$ for each pattern $\ord'$ in the resulting family.
\end{lemma}

\begin{proof}
    Let $(\ord^0,i_0),\dots,(\ord^r,i_r)$ with be some sequence of calls of the $\operatorname{ReduceVector}$ operation corresponding to some root-leaf path of the recursion tree $\mathcal{T}_{\ord,i}$. 
    For $t \in [r]_0$, let $\ord^t = (\iso, O^t, b^t)$.
    Note that $i_0,\dots,i_r$ are the labels assigned to the nodes of this path. Then in each iteration, the coordinates of $i_j$ are shifted to either their big indices or to $i_{j+1}$. 
    Hence, $i_{j+1}$ is the only index that can turn bad in $O^{j+1}$ with respect to $b^{j+1}$, i.e., $I_{\ord^{j+1}}\subseteq (I_{\ord^j}\setminus\{i_j\})\cup \{i_{j+1}\}$.
    It follows by induction, that $I_{\ord^r}\subseteq I_{\ord}\setminus\{i\}$ holds for the last pattern $(\ord^r)$ in the sequence.

    Moreover, it holds by \cref{obs:elst:cw_reduce_vector}, $b^{j+1}([k] \times \{\sleaf, \sinner\})\subseteq b^j([k] \times \{\sleaf, \sinner\})$, and hence, $b^{r}([k] \times \{\sleaf, \sinner\})\subseteq b([k] \times \{\sleaf, \sinner\})$ holds also by induction. 
    The equivalence of the families also holds by induction using \cref{cor:elst:cw_shift_vector_equiv}, since we only replace a pattern with the family resulting from an application of $\operatorname{ShiftVector}$ on this pattern.
\end{proof}

 Finally, we define the $\operatorname{ReduceToNice}$ operation that, given a pattern $A$, outputs a family of nice patterns equivalent to $A$. The first step of this operation is a subroutine called $\operatorname{ProcessBig}$ that takes care of the big vectors, i.e., ensures that the big vectors have no bad coordinates.

\begin{procedure}\label{def:elst:cw_process_big}
    We define the subroutine \defi{$\operatorname{ProcessBig}$} that takes an ordered pattern $\ord$ as input, and outputs a family of ordered patterns as follows:
    Let $H_0 = \{\ord\}$. We compute $H_j$ from $H_{j-1}$ for $j\in[k]$ as follows:
    For each pattern $\ord'\in H_{j-1}$, we define the set $C_{\ord'}$, where $C_{\ord'} \coloneq \{\ord'\}$ if $I_{\ord'}\cap b'([k] \times \{\sleaf, \sinner\}) = \emptyset$, and $C_{\ord'} \coloneq \operatorname{ReduceVector}(\ord', i)$ for the smallest $i \in I_{\ord'}\cap b'([k] \times \{\sleaf, \sinner\})$ otherwise.

    We define $H_j = \bigcup_{\ord'\in H_{j-1}} C_{\ord'}$. Let $s$ be the smallest integer such that $H_j = H_{j+1}$. The operation outputs the family $H_s$.
\end{procedure}

\begin{lemma}\label{lem:elst:cw_process_big}
    It holds for each ordered pattern $\ord = (\iso, O, b)$, that $\operatorname{ProcessBig}(\ord)$ runs in time $k^{\bigoh(k)}$ and outputs a family $H$ of ordered patterns equivalent to $\ord$, such that for each pattern $\ord' = (\iso, O', b') \in H$, it holds that $I_{\ord'}\cap b'([k] \times \{\sleaf, \sinner\}) = \emptyset$, $b'([k] \times \{\sleaf, \sinner\})\subseteq b([k] \times \{\sleaf, \sinner\})$ and $O_j = O'_j$ for all $j\notin b([k] \times \{\sleaf, \sinner\})$. Moreover, the process finishes after at most $2k$ iterations.
\end{lemma}

\begin{proof}
    We start with the last claim. 
    For each ordered pattern $\ord' \in H_j$ with $\ord' \in \operatorname{ReduceVector}(\ord'',i)$ for some pattern pattern $\ord''\in H_{j-1}$, it holds that $i \in I_{\ord'}\cap b'([k], X)$. By \cref{lem:elst:cw_reduce_vector_equiv}, $I_{\ord'} \subseteq I_{\ord''}\setminus \{i\}$. Hence,  $I_{\ord'}\cap b'([k], \{\sleaf, \sinner\})$ is strictly smaller than $I_{\ord''}$. Therefore, it holds by item (3) of \cref{obs:elst:cw_reduce_vector} that the largest size of such a set over all patterns in $H_j$ decreases after each iteration. Hence, it must hold after $2k$ iterations for all patterns $\ord'\in H_j$ that $I_{\ord'}\cap b'([k]\times \{\sleaf, \sinner\}) = \emptyset$.

    The claim $b'([k]\times \{\sleaf, \sinner\})\subseteq b([k]\times \{\sleaf, \sinner\})$ holds by \cref{obs:elst:cw_reduce_vector}. The claim $O_j = O_j'$ for all $j\notin b([k]\times \{\sleaf, \sinner\})$ also holds by \cref{obs:elst:cw_reduce_vector},
    since we only call $\operatorname{ReduceVector}$ on indices in $b([k] \times \{\sleaf, \sinner\})$, which only shifts indices to $b([k]\times \{\sleaf, \sinner\})$. $H_j$ is equivalent to $H_{j-1}$ by \cref{lem:elst:cw_reduce_vector_equiv}, since we only replace a pattern with the family resulting from an application of $\operatorname{ReduceVector}$ on this pattern. It follows by induction over $j$ that $H$ is equivalent to $\ord$.

    We claim that $|H_j|\leq k^{\bigoh(k)}$ for all $j\geq 0$. The lemma then follows by \cref{lem:elst:cw_reduce_vector_size}. For $\ord'\in H_j$ with $\ord' = (\iso, O', b')$, by \cref{lem:elst:cw_reduce_vector_equiv} for each index $i\in I_{\ord'}$,  $O'_i = O_i$. This holds since each time we shift to an index $i$ (possibly making it a bad index), we directly call $\operatorname{ReduceVector}$ on $i$, and hence, by \cref{lem:elst:cw_reduce_vector_equiv}, $i$ cannot be a bad index in the patterns resulting from such a call. Therefore, each pattern of $H_j$ can be uniquely identified by the indices of $b([k] \times \{\sleaf, \sinner\})$ that were ``reduced'', and by their respective values. Let $R = \{i_1, \dots, i_r\}$ be the set of these indices that do not correspond to unit or zero pairs. 
    Let $(u_{i_1}, S_{i_1}), (u_{i_2}, S_{i_2}), \dots, (u_{i_r}, S_{i_r})$ be the pairs corresponding to the indices in $R$.
    Since all the indices of $R$ are not bad indices, the non-zero coordinates of the $u_{i_j}$'s for $j \in [r]$ must form a partition of a subset of $[k]$; likewise, every element of $[k]$ appears in at most one set among $S_{i_1}, \dots, S_{i_r}$. 
    Hence, we have at most $3^{(2k)}\cdot (2k)^{(2k)} \cdot (2k)^{(2k)} \cdot (2k+1)^{(2k)} = k^{\bigoh(k)}$ possible patterns in $H_j$, where the first term counts for each index $i\in b([k]\times\{\sinner, \sleaf\})$ whether $O'_i = O_i$, $O'_i$ is a unit or zero pair, or $i \notin I_{\ord'}$ and $O_i$ is not a unit or zero vector. 
    The second term counts the number of partitions of the non-zero coordinates, while the third term counts the number of ways to distribute $[k]$ among $S_{i_1}, \dots, S_{i_r}$.
    Finally, the last term counts the different unit and zero pairs for each index corresponding to a unit or zero vector $O'_i$.
\end{proof}

Now we are ready to describe the procedure $\operatorname{ReduceToNice}$, that takes as input a pattern and outputs a family of nice patterns equivalent to the input. 

\begin{procedure}\label{def:elst:cw_reduce_to_nice}
    The procedure \defi{$\operatorname{ReduceToNice}$} takes as input a pattern $A\in\mathcal{P}$, and outputs a family of patterns.
    The procedure first fixes an arbitrary ordered pattern $\ord_A\coloneqq(\iso, O,b)$ of $A$. It computes the family $S_0=\operatorname{ProcessBig}(\ord_A)$.
     
    Let $\{i_1,\dots,i_s\} \coloneqq I_{\ord}\setminus b([k]\times\{\sleaf, \sinner\})$. 
    For each value $j\in[s]$, we define $ S_j$ as follows: First, let $ S'_j = \bigcup_{\ord'\in  S_{j-1}} \operatorname{ReduceVector}(\ord', i_j)$.
    Then the family $ S_j$ results from $ S'_j$ as follows: we iterate over all patterns $\ord' \in  S'_j$, and we add $\ord'$ to $ S_j$ if the following two conditions hold: first, if $\ord'$ does not contain any zero pair, and second if $ S_j$ does not already contain a pattern $\ord''$, that shares the same underlying pattern with $\ord'$, i.e., $p(\ord') = p(\ord'')$. In other words, for each pattern $p\in\mathcal{P}$, we only add one ``representative'' ordered pattern of $p$ to $ S_j$ if such a pattern exists in $ S'_j$.
    The operation outputs the family
    $\big\{p(\ord)\colon \ord\in S_s\big\}$.
    \xqed
\end{procedure}

In the following we show that this subroutine runs in time $n^{\bigoh(k)}$, and correctly produces a family of nice patterns equivalent to the input pattern.

\begin{lemma}\label{lem:elst:cw_reduce_to_nice}
    It holds for all $j\in[s]_0$ and $\ord'\in  S'_j$ with $\ord' = (\iso, O', b')$ that $I_{\ord'} = \{i_{j+1},\dots, i_s\}$ and $b'([k]\times \{\sleaf, \sinner\}) \subseteq b([k]\times \{\sleaf, \sinner\})$. Further, it holds that $O'_{\ell} = O_{\ell}$ for all $\ell\in\{i_{j+1},\dots,i_s\}$.
\end{lemma}
\begin{proof}
    We prove the claim by induction over $j\in[s]$.   
    For $j=0$, the claim holds by \cref{lem:elst:cw_process_big}, since $I_{\ord'} = I_{\ord}\setminus b([k]\times \{\sleaf, \sinner\})$ for all $\ord' \in  S_0$. Moreover, $b'([k]\times\{\sleaf, \sinner\}) \subseteq b([k]\times \{\sleaf, \sinner\})$ holds since each pattern resulting from $\operatorname{ShiftVector}$ only shifts big vector assignments to big indices in $b([k]\times \{\sleaf, \sinner\})$.

    Now assume the claim holds for $j-1$. Let $\ord' = (\iso, O', b')$ be a pattern in $ S'_{j}$, and let $\ord^* = (\iso, O^*, b^*)$ be a pattern in $ S_{j-1}$ such that $\ord'\in \operatorname{ReduceVector}(\ord^*, i^{\sleaf}_j)$.
    Then by \cref{lem:elst:cw_reduce_vector_equiv}, we have that 
    $I_{\ord'}\subseteq I_{\ord^*}\setminus\{i_j\}$ and $b'([k\times \{\sleaf, \sinner\}])\subseteq b^*([k\times \{\sleaf, \sinner\}])$. Moreover, it holds by the induction hypothesis that $I_{\ord^*} = \{i_{j},\dots,i_s\}$, and $b^*([k\times \{\sleaf, \sinner\}])\subseteq b([k\times \{\sleaf, \sinner\}])$. It follows that $I_{\ord'} \subseteq \{i_{j+1},\dots,i_s\}$, and $b'([k\times \{\sleaf, \sinner\}])\subseteq b([k\times \{\sleaf, \sinner\}])$.
    
    Finally, the inclusion $\{i_{j+1},\dots,i_s\}\subseteq I_{\ord'}$ follows from the fact that $\operatorname{ShiftVector}$ only shifts to big indices, and that it holds by induction hypothesis that $b^*([k\times \{\sleaf, \sinner\}])\subseteq b([k\times \{\sleaf, \sinner\}])$ is distinct from all indices $i_1,\dots, i_s$. By the same observation, we also have that $O'_{i_{\ell}} = O^*_{i_{\ell}} = O_{i_{\ell}}$ for all $\ell\in\{j+1,\dots,s\}$.
\end{proof}

\begin{corollary}\label{cor:elst:cw_reduce_to_nice_output}
    For each $A\in\mathcal{P}$, the operation $\operatorname{ReduceToNice}(A)$ produces a family of nice patterns equivalent to $A$.
\end{corollary}
\begin{proof}
    It holds by \cref{lem:elst:cw_reduce_to_nice} that $I_{\ord'} = \emptyset$ for all $\ord' \in \bar{ S}_{\bar{s}}$. Since we remove all patterns containing a zero pair, by \cref{obs:elst:cw_ord_nice}, the output family is nice. 
    
    It holds by \cref{lem:elst:cw_process_big} that $S_0$ is equivalent to $\{\ord_A\}$ 
    We show that $S_j$ is equivalent to $S_{j-1}$ for all $j\in[s]$, and hence, it hold by induction over $j\in[s]_0$ that $ S_j$ is equivalent to $\{\ord_A\}$.
    It holds that $ S'_j$ is equivalent to $ S_{j-1}$ for all $j\in[s]$ by \cref{lem:elst:cw_reduce_vector_equiv}. The family $ S_j$ is then equivalent to $ S'_j$ by the fact that we only remove patterns that are either equivalent to other patterns in the family, or contain at least a zero pair. Since the latter patterns are not compatible with any partial solution, they are equivalent to the empty family, and they can be removed safely.
\end{proof}

\begin{lemma}\label{lem:elst:cw_reduce_to_nice_size}
    For each $j\in[s]_0$, $| S_j| \leq |\mathcal{P^*}| \leq n^{\bigoh(k)}$.
    Similarly, for each $j\in[\bar{s}]_0$, $|\bar{ S}_j| \leq |\mathcal{P^*}| \leq n^{\bigoh(k)}$.
\end{lemma}
\begin{proof}
    For $j=0$ this holds by \cref{lem:elst:cw_process_big}. Now assume $j>0$.
    By \cref{lem:elst:cw_reduce_to_nice}, all patterns in $ S_j$ have the same set of bad indices and identical vectors at bad indices, i.e., it holds for all $(O',b')\in  S_j$, $I_{O',b'} = \{i_{j+1}, \dots, i_s \}$, and for $i\in I_{O',b'}$, $O'_i = O_i$.
    Let $ S_j^*$ be the family resulting from $ S_j$ by removing all pairs $\{O_i\colon i\in \{i_{j+1}, \dots, i_s \}\}$. Then $ S_j^*$ contains nice patterns only, each containing at most $n$ vectors. Moreover, for any two patterns $(O',b'),(O'',b'')\in  S_j^*$, it holds that $p(O',b') \neq p(O'',b'')$, since otherwise, the original ordered patterns in $ S_j$ would share the same underlying pattern, a contradiction to our choice of of $ S_j$. Hence, the size of $ S_j^*$ (and hence, of $ S_j$) is bounded by the number of nice patterns with at most $n$ vectors (i.e., $|\mathcal{P^*}|$), which is at most $n^{\bigoh(k)}$, by \cref{obs:elst:sizenice}.
\end{proof}

\begin{corollary}\label{cor:elst:cw_reduce_to_nice_time}
    The operation $\operatorname{ReduceToNice}(A)$ runs in time $n^{\bigoh(k)}$ for each pattern $A\in\mathcal{P}$.
\end{corollary}
\begin{proof}
    By \cref{lem:elst:cw_process_big}, $\operatorname{ProcessBig}(\ord)$ runs in time $k^{\bigoh(k)}$.
    For each $j\in [s]$, we apply $\operatorname{ReduceVector}$ on each element of the family $ S_{j-1}$. This runs in time $| S_{j-1}|\cdot k^{\bigoh(k)}\poly(n) \leq n^{\bigoh(k)}$ by Lemmas~\ref{lem:elst:cw_reduce_vector_size} and \ref{lem:elst:cw_reduce_to_nice_size}, since each call of the $\operatorname{ShiftVector}$ runs in polynomial time. 
    Moreover, it holds by these two lemmas that $ S'_j$ has size $n^{\bigoh(k)}$. Hence, $ S_j$ can be computed from $ S'_j$ in time $(n^{\bigoh(k)})^2\poly(n) = n^{\bigoh(k)}$.
    Since $s\leq n$, the total running time is $n^{\bigoh(k)}$.
\end{proof}

\subsection{The Algorithm}

Before defining the algorithm, we describe an auxiliary operator that corresponds to relabeling vertices in the NLC-expression, i.e., it describes how a pattern corresponding to a partial solution changes after a relabel operation.
 \begin{definition}
    Let $\beta : [k] \to [k]$
    be some relabeling function. For a vector $v\in \mathbb{N}^{k}$, we define $\rho_{\beta}(v) = v'$, where $v'_i = \sum_{j \in \beta^{-1}(i)} v_j$, and for a set $S \subseteq [k]$, we define $\rho_{\beta}(S) = \big\{\beta(i)\colon i\in S\big\}$.
    Given a pattern $(\iso, \mset)\in \mathcal{P}$, we define
    $\rho_{\beta}(\iso, \mset) = \big(\rho_{\beta}(\iso), \rho_{\beta}(\mset)\big)$,
    where $\rho_{\beta}(\mset)=\big\langle\big(\rho_{\beta}(u), \rho_{\beta}(S)\big)\colon (u, S)\in \mset\big\rangle$.
\end{definition}

Now, we describe our algorithm. Intuitively, at each join node $x$ of the NLC-decomposition, in order to create all partial solutions at $x$ we iterate over all pairs of partial solutions at its children, and for each such pair we iterate over all possible ways of adding edges between them. However, to keep the number of created patterns bounded, we add the edges one by one, applying the $\operatorname{ReduceToNice}$ procedure on all resulting patterns after each addition. After that we apply the corresponding relabeling procedure on the resulting patterns and make our record nice one final time by applying $\operatorname{ReduceToNice}$ again.

Along this algorithm, we will deal with $2k$-labeled graphs, distinguishing between the labels of the vertices coming from the first and the second child of a join node. This will allow us to unify the partial solutions from the two children without conflicts, allowing us to add edges between them correctly.
 \begin{algorithm}
\label{def:elst:cw_algo}
    For each $x\in\nodes$, we define the families $T_x \subseteq \mathcal{P}^*$ recursively over $\syntaxtree$ as follows. For an initial node $\bullet_i$, we define $T_x = \left\{\big(\idv_i, \langle\rangle\big)\right\}$.

    For a join node $x$ with children $x_1$ and $x_2$, let $\alpha\subseteq [k]^2$ and $\beta : [k]\to[k]$ such that $G_x = G_{x_1} \oplus_{\alpha}^{\beta} G_{x_2}$.
    The algorithm computes the table $T_x$ as follows:
    
    First, the algorithm computes the set $T'$ that corresponds to a representation of all combinations of partial solutions from $T_{x_1}$ and $T_{x_2}$ after applying the join operation $\oplus_{\alpha}$. In order to do this, we first define the following sets $T^{(i)}$ for $i\in\mathbb{N}_0$. Let $\beta_0 : [2k]\rightarrow[2k]$ be the relabeling function defined as $\beta_0(i) = i + k$ for $i\in[k]$ and $\beta_0(i) = i$ otherwise, and $\beta_1 : [2k]\rightarrow[2k]$ be defined as $\beta_1(i) = i$ for $i\in[k]$ and $\beta_1(i) = i - k$ otherwise.
    Moreover, let $\alpha' := \{(x,\beta_0(y)) \colon (x,y) \in \alpha \}$.
    We define the set
    \[
    T^{(0)} := \bigcup_{(A_1, A_2) \in T_{x_1} \times T_{x_2}} \operatorname{ReduceToNice}(\repr(F_{A_1} \dot \cup \rho_{\beta_0}(F_{A_2}))).
    \]
    This forms all disjoint unions of partial solutions from the two children, where the labels of the second child are shifted by $k$.
    For $z > 0$, we define the sets $T^{(z)}$ from $T^{(z-1)}$ as follows:
    First, we define $T^{(z)}_0$, where for each pattern $A\in T^{(z-1)}$, and for each pair $\{u,v\} \in \big(V(F_A) \times V(F_A)\big) \setminus E(F_A)$
         with $(\lab_{F_A}(u), \lab_{F_A}(v) ) \in \alpha'$
    and such that the graph $F'$ resulting from adding the edge $\{u,v\}$ to $F_A$ is acyclic, the algorithm adds $\repr(F')$ to $T_0^{(z)}$. Finally, the algorithm defines $T^{(z)} = \bigcup_{A\in T_0^{(z)}} \operatorname{ReduceToNice}(A)$.
    
    We define the sets  $T^{(z)}_R$ for each value $z\geq 0$ as the set $T^{(z)}_R = \rho_{\beta_1}(T^{(z)})$, and $T^{(z)}_N = \bigcup_{A\in T^{(z)}_R} \operatorname{ReduceToNice}(A)$.
    Let $r$ be the smallest integer such that $T^{(r)} = T^{(r+1)}$. The algorithm computes $T' =  \bigcup_{z\in[r]_0} T^{(z)}_N$.
    Finally, the algorithm computes the set $T'' = \bigcup_{A\in T'} \rho_{\beta}(A)$ and the set $T_x =\bigcup_{A\in T''} \operatorname{ReduceToNice}(A)$.
    The algorithm outputs YES if $T_r$ contains a solution pattern $\elpat$ at the root node $r$, or outputs NO otherwise. 
    \xqed
\end{algorithm}

Now we show that the computed set $T_x$ is equivalent to the record $\mathcal{R}_x$ at each node $x\in\nodes$. First, we prove that adding a single edge between two labels preserves the representation of partial solutions. This lemma requires the finer notion of compatibility defined in \cref{def:elst:cw_compatible_forests}, which is essential for proving the correctness of the algorithm. We also show that the relabeling operation preserves equivalence.

\begin{lemma}\label{lem:elst:cw_single_edge_equiv}
    Let $\mathcal{R}_1$ and $\mathcal{R}_2$ be two families of patterns such that $\mathcal{R}_1\simeq \mathcal{R}_2$. For $i,j\in[k]$ with $i\neq j$, we define the families $\mathcal{R}'_1$ and $\mathcal{R}'_2$, where for each pattern $A\in \mathcal{R}_1$ and for each pair $(u,v)$ such that $\lab_{F_A}(u) = i$, $\lab_{F_A}(v) = j$, $uv \notin E(F_A)$, and the graph $F'$ obtained from $F_A$ by adding the edge $uv$, we add $\repr(F')$ to $\mathcal{R}'_1$. We define $\mathcal{R}'_2$ similarly from $\mathcal{R}_2$.
    Then, it holds that $\mathcal{R}'_1 \simeq \mathcal{R}'_2$.
\end{lemma}

\begin{proof}
    Let $\delta'\in \mathcal{A}$, $\ell\in\mathbb{N}_0$ and $F^*$ be a fixed forest such that there exists $F'_2 \in \mathcal{R}'_2$ that is $(\delta', \ell)$-compatible with $F^*$. We need to show that there exists $F'_1 \in \mathcal{R}'_1$ that is $(\delta', \ell)$-compatible with $F^*$. The other direction then follows by symmetry.
    
    Let $F^{**}$ be the forest obtained from $F^*$ by adding two new vertices $w_1$ and $w_2$ with labels $k+1$ and $k+2$ with an edge between them, and let $\delta:[k+2]\times[k+2]\rightarrow \mathbb{N}_0$ be the extension of $\delta'$ by setting $\delta(i, k+1) = \delta(j, k+2) = 1$ and $\delta(t,t') = 0$ for all other values of $t$ and $t'$ with $\{t, t' \} \cap \{k+1,k+2\} \neq \emptyset$.
    
    We claim that a fixed forest $F$ is $(\delta, \ell)$-compatible with the forest $F^{**}$ if and only if there exists a forest $F'$ that is $(\delta', \ell)$-compatible with $F^*$, such that $F'$ is obtained by adding an edge between a vertex labeled $i$ and a vertex labeled $j$ in $F$.

    Before we prove the claim, we show that the lemma follows from this claim. Since $F'_2$ is $(\delta', \ell)$-compatible with $F^*$, there exists by this claim a forest $F_2 \in \mathcal{R}_2$ that is $(\delta, \ell)$-compatible with $F^{**}$. But since $\mathcal{R}_1\simeq \mathcal{R}_2$, there exists a forest $F_1$ in $\mathcal{R}_1$ that is $(\delta, \ell)$-compatible with $F^{**}$. It follows again by the claim above that $\mathcal{R}'_1$ contains a forest $F_1'$ that is $(\delta', \ell)$-compatible with $F^*$ which proves the lemma.
    
    Now we prove the claim. Let $T$ be a $(\delta, \ell)$-spanning tree for $F$ and $F^{**}$. In $T$, $w_1$ has exactly one edge to a vertex in $F$ of label $i$, and similarly $w_2$ has exactly one edge to a vertex in $F$ of label $j$. Let $u$ and $v$ be these two vertices respectively.
          By replacing the path $uw_1w_2v$ path with the edge $\{u,v\}$, we obtain a $(\delta', \ell)$-spanning tree $T'$ for the forest $F'$ obtained from $F$. The other direction of the claim follows by reversing this argument, where we replace the edge $\{u,v\}$ added to create $F'$ in $T'$ with the path $uw_1w_2v$. Note that the number of leaves is the same in $T$ and $T'$, since the degrees of $u$ and $v$ do not change, and since $w_1$ and $w_2$ both have degree two in $T$, and hence, are not leaves.
\end{proof}

\begin{lemma}
\label{lem:elst:cw_relabel_equiv}
    Let $\mathcal{R}_1$ and $\mathcal{R}_2$ be two families of patterns such that $\mathcal{R} \simeq \mathcal{R}'$. For a relabeling function $\beta : [k]\to[k]$, we define the families $\mathcal{R}'_1 := \rho_{\beta}(\mathcal{R})$ and $\mathcal{R}'_2 := \rho_{\beta}(\mathcal{R}')$.
    Then, it holds that $\mathcal{R}'_1 \simeq \mathcal{R}'_2$.
\end{lemma}

\begin{proof}
    Let $\delta \in \mathcal{A}$, $\ell\in\mathbb{N}_0$ and $F^*$ be a fixed forest such that there exists $A'_1 \in \mathcal{R}'_1$ with $F_{A'_1}$ $(\delta, \ell)$-compatible with $F^*$. We need to show that there exists $A'_2 \in \mathcal{R}'_2$ such that $F_{A'_2}$ is $(\delta, \ell)$-compatible with $F^*$. The other direction then follows by symmetry.

    For a fixed forest $F$ and for $F' = \rho_{\beta}(F)$, we claim that $F'$ is $(\delta, \ell)$-compatible with $F^*$ if and only if there exists a mapping $\delta_0$ such that $\delta(i,j) = \sum\limits_{i' \in[k]\\ \beta(i') = i} \delta_0(i',j)$ for all $i,j\in[k]$, and $F$ is $(\delta_0, \ell)$-compatible with $F^*$. Before we prove the claim, we show that the lemma follows from this claim. Since $F_{A'_1}$ is $(\delta, \ell)$-compatible with $F^*$, there exists by this claim a mapping $\delta_0$ and a forest $F_{A_1}$ in $\mathcal{R}$ that is $(\delta_0, \ell)$-compatible with $F^*$. But since $\mathcal{R}\simeq \mathcal{R}'$, there exists a forest $F_{A_2}$ in $\mathcal{R}_2$ that is $(\delta_0, \ell)$-compatible with $F^*$. It follows again by the claim above that $\mathcal{R}'_2$ contains a pattern $A_2$ where $F_{A'_2}$ is $(\delta, \ell)$-compatible with $F^*$ which proves the lemma.

    Now we prove the claim. Let $T$ be a $(\delta, \ell)$-spanning tree for $F'$ and $F^*$. Then the tree $T_0$ resulting from $T$ by reverting the labels of the vertices of $F'$ to their original label in $F$ is a $(\delta_0, \ell)$-spanning tree for $F$ and $F^*$, where $\delta_0$ is defined as the number of cut edges between each pair of labels between $F$ and $F^*$ in $T_0$. It is not hard to see that $\delta$ and $\delta_0$ satisfy the relation in the claim. The other direction of the claim follows by reversing this argument, where we relabel the vertices of $F$ in $T_0$ according to $\beta$ to obtain $F'$, and we define $\delta$ as in the claim.
\end{proof}

\begin{lemma}
\label{lem:elst:cw_alg_corr}
    It holds for each $x\in \nodes$ that $T_x$ is equivalent to the record $\mathcal{R}_x$, i.e., $T_x\simeq \mathcal{R}_x$.
\end{lemma}

\begin{proof}
    We prove this lemma by induction over $\syntaxtree$. 
    Let $x$ be a $\bullet_i$ node and let $v$ be the only vertex of $G_x$. 
    Clearly, $G_x$ admits a single partial solution, namely, the graph $G_x$ itself, consisting of the isolated vertex $v$. Hence, it holds by the definition of $\repr$ that $\mathcal{R}_x=T_x$.

    For a join node $x\in\nodes$, let $x_1$ and $x_2$ be the children of $x$. 
    Let $\alpha_0 = \{(i, j+k)\colon (i,j)\in \alpha\}$, and $\hat{G} = G_{x_1}\oplus_{\alpha_0} \rho_{\beta_0}(G_{x_2})$. Let $\mathcal{R}'$ be the record of $\hat{G}$. By the induction hypothesis, $T_{x_1}\simeq \mathcal{R}_{x_1}$ and $T_{x_2}\simeq\mathcal{R}_{x_2}$.
    \begin{longclaim}\label{clm:elst:alg_equiv}
    It holds that $T' \simeq \mathcal{R}'$.
    \end{longclaim}
    \begin{claimproof}[Proof of \cref{clm:elst:alg_equiv}]
        Let $\mathcal{R}^{(0)}$ be the record corresponding to the disjoint union of all pairs of partial solutions from $\mathcal{R}_{x_1}$ and $\mathcal{R}_{x_2}$.
        We define the families $\mathcal{R}^{(z)}$ defined as the sets of all partial solutions that result from adding exactly $z$ edges between vertices whose labels belong to $\alpha_0$ to the partial solutions in $\mathcal{R}^{(0)}$.
        
        We claim that $T^{(z)} \simeq \mathcal{R}^{(z)}$ holds for all values of $z$. Before we prove the claim, we note that \cref{clm:elst:alg_equiv} follows from it by~\cref{cor:elst:cw_reduce_to_nice_output}, since it follows from the definition of equivalence and by \cref{lem:elst:cw_relabel_equiv} that both the union operation and relabeling preserve equivalence.
        
        Now we prove the claim by induction over $z$.
        It holds by \cref{cor:elst:cw_reduce_to_nice_output} that $T^{(0)} \simeq \mathcal{R}^{(0)}$. Now assume that $T^{(z-1)} \simeq \mathcal{R}^{(z-1)}$ holds for some $z>0$. It holds by \cref{lem:elst:cw_single_edge_equiv} that the partial solution resulting from adding a single edge between two labels preserves equivalence. Again, since it holds that the union operation preserves equivalence, and $\mathcal{R}^{(z)}$ can be defined from $\mathcal{R}^{(z-1)}$ by the union over all $(i, k+j) \in \alpha_0$ of adding a single edge between a vertex of label $i$ and a vertex of label $k+j$ to some partial solution, and since $T_0^{(z)}$ results from $T^{(z-1)}$ in the same way, it holds that $T^{(z)}_0 \simeq \mathcal{R}^{(z)}$.
        
        Hence, it follows from the transitivity of equivalence, and from \cref{cor:elst:cw_reduce_to_nice_output} that $T^{(z)} \simeq T^{(z)}_0 \simeq \mathcal{R}^{(z)}$, which completes the proof.
    \end{claimproof}

    Since $\mathcal{R}_x$ results from $\mathcal{R}'$ by applying the relabeling function $\beta$, it holds by \cref{lem:elst:cw_relabel_equiv} that $T'' \simeq \mathcal{R}_x$.
    Again, it follows from \cref{cor:elst:cw_reduce_to_nice_output} that $T_x\simeq \mathcal{R}_x$.
\end{proof} 
We now characterize the YES-instances:

\begin{corollary}
\label{cor:elst:cw_zero_acc}
    The instance $(G, \eltarget)$ is a YES-instance if and only if $T_r$ contains a solution pattern.
\end{corollary}

\begin{proof}
    Let $\elpat'$ be the empty pattern. Then it holds that $\elpat$ is $(\emptyset, \eltarget)$-compatible with $\elpat'$, if and only if $\elpat$ is a solution pattern.
    Since it holds by \cref{lem:elst:cw_alg_corr} that $T_r$ is equivalent to $\mathcal{R}_r$, it holds that $T_r$ contains a solution pattern if and only if $\mathcal{R}_r$ contains such a pattern. Finally, it holds by \cref{obs:elst:sol-pat} that this is the case if and only if $G$ is a YES-instance.
\end{proof} 
Finally, we bound the running time of our algorithm:
\begin{lemma}\label{lem:elst:cw_time}
    Algorithm~\ref{def:elst:cw_algo} runs in time $n^{\bigoh(k)}$.
\end{lemma}

\begin{proof}
    Since the size of the decomposition tree is polynomial in $n$, it suffices to bound the running time for a single node $x$. This is trivial for a leaf node. For an inner node $x$ with children $x_1$ and $x_2$, since $T_{x_1}$ and $T_{x_2}$ contain at most $n^{\bigoh(k)}$ patterns each, by \cref{cor:elst:cw_reduce_to_nice_time}, the computation of $T^{(0)}$ takes time $n^{\bigoh(k)}$. Since it holds by \cref{cor:elst:cw_reduce_to_nice_output} that $T^{(z)}$ contains only nice patterns, and since for each such pattern we iterate over at most $n^2$ different ways of adding a single edge, and since we call $\operatorname{ReduceToNice}$ once after adding each edge, it follows by \cref{cor:elst:cw_reduce_to_nice_time} that the computation of each $T^{(z)}$ takes time $n^{\bigoh(k)}$ for each value of $z$. Since $r$ is bounded by the number of added edges, and hence, by $n^2$, computing $T'$ can be done in $n^{\bigoh(k)}$. Since relabeling a single pattern can be done in polynomial time, also $T''$ can be computed in time $n^{\bigoh(k)}$. Thus, by \cref{cor:elst:cw_reduce_to_nice_time}, computing $T_x$ takes time $n^{\bigoh(k)}$.
\end{proof}

Now we are ready to prove the main theorem of this section:

\begin{proof}[Proof of Theorem~\ref{thm:elst-cw-ub}]
    Given an instance $(G,\eltarget)$ of \textsc{Exact Leaf Spanning Tree} together with an NLC-decomposition of width $k$, we run Algorithm~\ref{def:elst:cw_algo}. By Corollary~\ref{cor:elst:cw_zero_acc}, the input instance is a YES-instance if and only if $T_r$ contains a solution pattern. The total running time is $n^{\bigoh(k)}$, by Lemma~\ref{lem:elst:cw_time}.
\end{proof}

\section{Concluding Remarks}
\noindent
Our results not only provide a comprehensive and in-depth overview of the fine-grained complexity of computing a fundamental class of spanning trees, but also develop new techniques which we believe are applicable to other problems, including those related to graph factoring. 
One task left for future work is closing the small remaining gap between the upper bound for treewidth (Theorem~\ref{thm:tw-ub}) and the lower bound for pathwidth (Theorem~\ref{thm:tw_lb}); presently, it seems difficult to obtain further progress in either of these two directions. We remark that a similar gap also exists for the classical \textsc{Hamiltonian Cycle} problem~\cite{CyganKN18}, which has a tight $\bigoh^*((2+\sqrt{2})^{\pw(G)})$ lower and upper bound, but no matching tight bound for treewidth is known.

Alternatively, one may ask whether using more restrictive structural parameterizations---such as the \emph{treedepth} of the input graph---can be used to design fixed-parameter algorithms for instances with large $r$, similarly as was achieved in Theorem~\ref{thm:ctw-ub}.
Finally, a long-standing open question in the field~\cite{CyganNPPRW22,NederlofPSW23} is whether Cut\&Count based algorithms can be derandomized while achieving the same asymptotic running time. To the best of our knowledge, all known deterministic techniques for connectivity problems parameterized by pathwidth (and treewidth) yield single exponential running time but with a blow up in the base of the exponent \cite{BodlaenderCKN15, Wlodarczyk19, FominLS14}.

\bibliography{ref}

\end{document}